\documentclass[11pt]{article}

\usepackage[T5]{fontenc}

\usepackage{amsmath,amssymb,amsthm,algpseudocode,float,caption,mathtools,thmtools,thm-restate,scalerel,stackengine,verbatim}
\usepackage{algorithm,algpseudocode}
\usepackage[shortlabels]{enumitem}

\stackMath
\newcommand\reallywidehat[1]{%
\savestack{\tmpbox}{\stretchto{%
  \scaleto{%
      \scalerel*[\widthof{\ensuremath{#1}}]{\kern-.6pt\bigwedge\kern-.6pt}%
          {\rule[-\textheight/2]{1ex}{\textheight}}
            }{\textheight}%
            }{0.5ex}}%
            \stackon[1pt]{#1}{\tmpbox}%
            }
\usepackage[margin=1in]{geometry}
\usepackage{multirow}
\usepackage{subcaption}
\usepackage{afterpage}
\usepackage{bm}
\usepackage{tikz}
\usetikzlibrary{decorations.pathreplacing, positioning}

\usepackage{bbm}
\usepackage{hyperref}
\usepackage{nameref}

\newcommand{\wt}[1]{\widetilde{#1}}

\newcommand{\wh}[1]{\widehat{#1}}
\newcommand{\abs}[1]{\left|#1\right|}
\newcommand{\eps}{\varepsilon}

\newcommand{\Nbb}{\mathbb{N}}

\newcommand{\norm}[1]{\left\lVert#1\right\rVert}

\newcommand{\floor}[1]{\left\lfloor#1\right\rfloor}

\DeclarePairedDelimiter\brac{\lbrack}{\rbrack}
\DeclarePairedDelimiter\set{\lbrace}{\rbrace}
\DeclarePairedDelimiter\paren{\lparen}{\rparen}

\newcommand{\E}[2][]{\operatorname*{\mathbb{E}}_{#1 }\brac*{#2}}
\newcommand{\Pb}[2][]{\operatorname*{Pr}_{#1 }\brac*{#2}}
\newcommand{\Var}[2][]{\operatorname*{\normalfont{\text{Var}}}_{#1 }\paren*{#2}}
\newcommand{\bO}[1]{\operatorname*{O}\paren*{#1}}
\newcommand{\bOt}[1]{\operatorname*{\wt{O}}\paren*{#1}}
\newcommand{\bOm}[1]{\operatorname*{\Omega}\paren*{#1}}

\newcommand{\bTh}[1]{\operatorname*{\Theta}\paren*{#1}}
\newcommand{\lO}[1]{\operatorname*{o}\paren*{#1}}

\DeclareMathOperator*{\1bb}{\mathbbm{1}}

\newcommand{\Ab}{\mathbf{A}}
\newcommand{\Bb}{\mathbf{B}}
\newcommand{\Db}{\mathbf{D}}
\newcommand{\Sb}{\mathbf{S}}
\newcommand{\bb}{\mathbf{b}}
\newcommand{\Tb}{\mathbf{T}}
\newcommand{\Xb}{\mathbf{X}}
\newcommand{\Yb}{\mathbf{Y}}
\newcommand{\Zb}{\mathbf{Z}}
\newcommand{\Hb}{\mathbf{H}}
\newcommand{\Ub}{\mathbf{U}}
\newcommand{\Fb}{\mathbf{F}}
\newcommand{\Cb}{\mathbf{C}}
\newcommand{\vb}{\mathbf{v}}

\newcommand{\rb}{\mathbf{r}}
\newcommand{\tb}{\mathbf{t}}
\newcommand{\cb}{\mathbf{c}}

\newcommand{\gb}{\mathbf{g}}
\newcommand{\Ob}{\mathbf{O}}
\newcommand{\Mb}{\mathbf{M}}
\newcommand{\Qb}{\mathbf{Q}}

\newcommand{\Ac}{\mathcal{A}}
\newcommand{\Hc}{\mathcal{H}}
\newcommand{\Uc}{\mathcal{U}}
\newcommand{\Sc}{\mathcal{S}}
\newcommand{\Ec}{\mathcal{E}}
\newcommand{\Kc}{\mathcal{K}}
\newcommand{\Bc}{\mathcal{B}}

\newcommand{\Fc}{\mathcal{F}}
\newcommand{\Gc}{\mathcal{G}}
\newcommand{\Dc}{\mathcal{D}}

\newcommand{\Gbc}{\mathbf{\mathcal{G}}}
\newcommand{\Fbc}{\mathbf{\mathcal{F}}}

\newcommand{\bPi}{\boldsymbol{\Pi}}
\newcommand{\bpi}{\boldsymbol{\pi}}
\newcommand{\bsigma}{\boldsymbol{\sigma}}
\newcommand{\bbeta}{\boldsymbol{\beta}}

\newcommand{\bsim}{\boldsymbol{\sim}}

\DeclareMathOperator*{\bi}{\text{Bi}}

\newcommand{\ColComp}{\normalfont{\textproc{CollectComponent}}}
\newcommand{\ColCompIdeal}{\normalfont{\textproc{CollectComponentIdeal}}}
\DeclareMathOperator{\CountComp}{\normalfont{\textproc{CountComponents}}}
\DeclareMathOperator{\FindComp}{\normalfont{\textproc{FindComponent}}}

\newcommand{\SC}{\normalfont{\textsc{StreamingCycles}}}

\DeclareMathOperator{\plog}{\normalfont{polylog}}

\newcommand{\Ibb}{\mathbb{I}}

\newcommand{\return}{\textbf{return}~}

\newcommand{\zpp}{\mathbb{Z}_+^{\lbrace  \rbrace}}

\renewcommand{\next}{\text{next}}

\newcommand{\grow}{\textbf{Grow}}

\newcommand{\e}{{\varepsilon}}
\newcommand{\bool}{{\{0, 1\}}}
\newcommand{\R}{{\mathbb{R}}}

\newcommand{\expect}{{\mathbb E}}
\newcommand{\F}{{\mathcal{F}}}

\DeclareMathOperator{\poly}{poly}

\newtheorem{theorem}{Theorem}[section]
\newtheorem*{theorem*}{Theorem}
\newtheorem{lemma}[theorem]{Lemma}
\newtheorem{definition}[theorem]{Definition}
\newtheorem*{definition*}{Definition}
\newtheorem*{lemma*}{Lemma}

\newtheorem*{corollary*}{Corollary}
\newtheorem{claim}[theorem]{Claim}
\newtheorem*{claim*}{Claim}

\newcommand{\nth}{^\text{th}}

\newenvironment{proofof}[1]{\noindent{\bf Proof of #1:}}{$\qed$\par}

{\makeatletter
	\gdef\xxxmark{%
		\expandafter\ifx\csname @mpargs\endcsname\relax 
		\expandafter\ifx\csname @captype\endcsname\relax 
		\marginpar{xxx}
		\else
		xxx 
		\fi
		\else
		xxx 
		\fi}
	\gdef\xxx{\@ifnextchar[\xxx@lab\xxx@nolab}
	\long\gdef\xxx@lab[#1]#2{{\bf [\xxxmark #2 ---{\sc #1}]}}
	\long\gdef\xxx@nolab#1{{\bf [\xxxmark #1]}}
}

\makeatletter
\let\orgdescriptionlabel\descriptionlabel
\renewcommand*{\descriptionlabel}[1]{%
  \let\orglabel\label
  \let\label\@gobble
  \phantomsection
 \edef\@currentlabel{#1\unskip}%
  \let\label\orglabel
  \orgdescriptionlabel{#1}%
}
\makeatother

\title{Factorial Lower Bounds for (Almost) Random Order Streams}
\date{}
\author{Ashish Chiplunkar\\IIT Delhi \and John Kallaugher\\Sandia National Labs \and Michael Kapralov\\EPFL \and Eric Price\\UT Austin}

\begin{document}
\pagenumbering{gobble}
\hypersetup{pageanchor=false}
\maketitle

\begin{abstract}
In this paper we introduce and study the \textsc{StreamingCycles} problem, a
random order streaming version of the Boolean Hidden Hypermatching problem that
has been instrumental in streaming lower bounds over the past decade. In this
problem the edges of a graph $G$, comprising $n/\ell$ disjoint length-$\ell$
cycles on $n$ vertices, are partitioned randomly among $n$ players. Every edge
is annotated with an independent uniformly random bit, and the players' task is
to output, for some cycle in $G$, the sum (modulo $2$) of the bits on its
edges, after one round of sequential communication.

Our main result is an $\ell^{\Omega(\ell)}$ lower bound on the communication
complexity of \textsc{StreamingCycles}, which is tight up to constant factors
in the exponent. Applications of our lower bound for \textsc{StreamingCycles}
include an essentially tight lower bound for component collection in (almost)
random order graph streams, making progress towards a conjecture of Peng and
Sohler [SODA'18] and the first exponential space lower bounds for random walk
generation. 
\end{abstract}

\newpage
\tableofcontents
\newpage

\hypersetup{pageanchor=true}
\pagenumbering{arabic}

\section{Introduction}

The streaming model of computation has been a core model for small space
algorithms that process large datasets since the foundational paper
of~\cite{AlonMS96} showed how to (approximately) compute basic statistics of
large datasets using space only polylogarithmic in the size of the input. While
many such stream statistics are tractable to calculate in this
model---including frequency moments, number of distinct elements, and heavy
hitters---for graphs of $n$ edges even some basic problems are known to require
space $\Omega(n)$ in streams in which the edge arrival order is chosen
adversarially. At the same time, a recent line of work on testing graph
properties~\cite{KapralovKS14,CormodeJMM17,MonemizadehMPS17,PS18,KapralovMNT20,CzumajF0S20}
shows that when the edges of the input graph are presented in a uniformly
random order, one can approximate many fundamental graph properties (e.g.,
matching size, number of connected components, constant query testable
properties in bounded degree graphs) in polylogarithmic or even constant space.

These algorithms make use of small space graph exploration primitives, including:

\paragraph{Component collection.} Given a graph $G=(V, E)$ presented as a
random order stream and a budget $k$, collect the connected components of a
representative sample of vertices of $G$ assuming that many vertices in $G$
belong to components of size bounded by $k$\footnote{We note that our
definition of component collection here is not formal by design. This is
because our lower bound applies to a very weak formal version of component
collection, in which one is promised that a constant fraction of vertices of
the input graph belong to components of size at most $k$, and the task is to
output one such vertex (see Definition~\ref{def:comp-collection} in
Section~\ref{sec:lb}). Our algorithmic results, on the other hand, solve a
stronger version of the problem, that of outputting a `representative sample'
of vertices of the graph together with components that they belong to (as long
as those are of size bounded by $k$), in any graph. In particular, our
algorithmic primitive naturally leads to an algorithm that additively
approximates the number of connected components in the input graph, as
in~\cite{PS18}---see Section~\ref{sec:comp}.}. The work of~\cite{PS18} designs
a component collection primitive that uses $k^{O(k^3)}$ space, and then uses it
to obtain an additive $\e n$-approximation to the number of connected
components in $G$.

\paragraph{Random walk generation.}
Given a graph $G=(V, E)$ presented as a random order stream, a target walk
length $k$ and a budget $s$, generate a sample of (close to) independent random
walks of length $k$ from $s$ vertices in $G$ selected uniformly at random. The work of~\cite{KallaugherKP21} designs a primitive that outputs such walks
(with constant TVD distance to the desired distribution, say), using space
$2^{O(k^2)} s$.

\paragraph{$k$-disc estimation.}
Given a bounded degree graph $G=(V, E)$ presented as a random order stream, an
integer $k$, estimate $k$-disk frequencies\footnote{A $k$-disc is the subgraph induced by vertices at shortest path distance at most $k$ from a given vertex, and the set of $k$-disc frequencies corresponds to the numbers of occurrences of all such graphs, up to isomorphism. }.  The work of~\cite{MonemizadehMPS17} designs a primitive for estimating $k$-disc
frequencies and then uses it to show that any constant query testable property
of bounded degree graphs is random order streamable. 

The above primitives
perform depth-$k$ exploration in random order graph streams using space
exponential in $k$. Our work is motivated by the natural question:
\begin{quote}
Does depth-$k$ exploration in random order streams require space exponential in $k$?
\end{quote}
This question was originally raised by~\cite{PS18}, who wrote
\begin{quote}
{\em ...it will also be interesting to obtain lower bounds for random order streams. It seems to be plausible to conjecture that approximating the number of connected components requires space exponential in $1/\e$. It would be nice to have lower bounds that confirm this conjecture.}
\end{quote}

In this paper we make progress towards this conjecture, showing that this
dependence is indeed necessary, at least in graph streams that admit some
amount of correlation. Our lower bound is based on a new communication problem
that we refer to as the \SC{} problem, a relative of the well-studied Boolean
Hidden Hypermatching problem. We introduce the \SC{} problem next, then give
reductions from component collection and random walk generation. Our reduction
generates instances of the component collection and random walk generation
problem that are not quite random order streams, but rather allow for small
{\em batches} of edges as opposed to edges themselves to arrive in a random
order. We argue that this is in fact a very natural {\em robust} analog of the
idealized random order streaming model, and show that corresponding random
order streaming algorithms extend to the batch random order setting. Finally,
we give an overview of our lower bound, which is the main technical
contribution of the paper.

\paragraph{The {\normalfont\textsc{StreamingCycles}} problem.} Our main technical contribution is a tight lower bound for the \SC{} problem, which we now define. In an instance of $\SC(n, \ell)$ a graph $G = (V,E)$ made up of $n/\ell$
length-$\ell$ cycles is received as a stream of edges $e$ with bit labels
$x_e$. There are $n$ players, indexed by the edges of $G$. Upon arrival of an edge $e$ the corresponding bit label $x_e$ is given to the corresponding player as private input, together with a message from the previous player. The edges $e$ are posted on a common board as the edges arrive.
 The last player must return a vertex $v \in V$ and
the \emph{parity} of the cycle $C \subseteq E$ containing $v$, i.e.\ $\sum_{e
\in C} x_e$. We consider a distributional version of the problem, in which the bits $x_e$ are chosen independently and uniformly at random, and the ordering of the edges in the stream is uniformly random.

A na\"{\i}ve protocol for the \SC{} problem is for the players to track the
connected component of a vertex $u\in G$. This strategy succeeds if and only if
the edges of the component arrive ``in order'', which happens
with probability $\ell^{-\Theta(\ell)}$. Thus, it suffices to track
$\ell^{O(\ell)}$ vertices, which results in an $\ell^{O(\ell)}$ communication
per player protocol. Our main result shows that this is essentially best
possible:

\begin{theorem}[Main result; informal version of
Theorem~\ref{thm:cycleslb}]\label{thm:main-inf}
Any protocol for the \SC$(n, \ell)$ problem that succeeds with $2/3$
probability requires $\min(\ell^{\Omega(\ell)}, n^{0.99})$ bits of
communication from some player.  
\end{theorem}

\paragraph{Relation to the Boolean Hidden Hypermatching problem.} We note that
this is related to the search version of the Boolean Hidden Hypermatching
problem, in which a vector $x\in \bool^n$ is given to Alice, who sends a single
message to Bob. Bob, in addition to the message from Alice, is given a perfect
hypermatching with hyperedges of size $\ell$ and must output the parity of $x$
on one of the hyperedges. The Boolean Hidden Hypermatching problem admits a
protocol with $O(n^{1-1/\ell})$ communication (Alice simply sends Bob the
values of $x$ on a uniformly random subset of coordinates of size
$n^{1-1/\ell}$), and this bound is tight. Note that in the \SC{} problem the
bits of $x$ are associated with edges in the cycles, and therefore every cycle
can naturally be associated with a hyperedge in Boolean Hidden Hypermatching
problem. In contrast to the Boolean Hidden Hypermatching problem, in which the
bits are presented first and then the hyperedges are revealed, in the \SC{}
problem the bits and the identities of the hyperedges are gradually revealed to
the algorithm. Similarly to the Boolean Hidden Hypermatching problem, a
`sampling' protocol turns out to be nearly optimal. The na\"{\i}ve protocol
mentioned above, and considered in more detail in Section~\ref{sec:overview},
solves \SC$(n, \ell)$ using $\bO{\ell!}$ samples, and we prove a nearly
matching lower bound of $\min\set{\ell^{\Omega(\ell)}, n^{0.99}}$ bits.

\paragraph{Applications to component collection and random walk generation.} We now give a natural way for the players to produce a graph stream based on their inputs to the communication game:

\begin{itemize}
\item Define the vertex set $V' = V \times \bool$.
\item On receiving the $t^{\text{th}}$ edge $(uv, b_{uv})$, insert edges $(u,0)(v,b_{uv})$ and
$(u,1)(v,\overline{b_{uv}})$ into the stream.
\end{itemize}

In other words, every edge of the graph $G$ in the \SC{} problem becomes a pair of edges in $G'$. The order in which edges of $G'$ are presented is not quite random as {\em pairs of edges} as opposed to individual edges arrive in a uniformly random order.

We refer to such streams, in which {\em batches} of edges arrive uniformly at
random in the stream as opposed to edges themselves, as  {\em (hidden-)batch}
random order streams. Note that the reduction above generates a stream with
batches of size two, corresponding to the pairs of edges arriving at the same
time. We argue in Section~\ref{sec:hidden-batch} below that the batched model,
in which arrival times of edges could be correlated, but the correlations are
restricted by bounded size batches, is a very natural {\em robust} analog of
the idealized uniformly random streaming model. In particular, we show that,
surprisingly, some existing random order streaming algorithms for estimating
graph properties are quite robust, and can be made to work even when the
structure of the batches is not known to the algorithm, i.e.\ in the {\em
hidden-batch} random order model. 

Using the reduction above together with Theorem~\ref{thm:main-inf}, we get
\begin{theorem}[Component collection lower bound; informal version of Theorem~\ref{thm:compestlb}]\label{thm:comp-lb}
Component collection  requires $k^{\Omega(k)}$ bits of space in hidden-batch random order streams.
\end{theorem}
\begin{proof}
If the component collection algorithm returns a vertex $(v,b)$ together with a component of size $\ell$ that contains $(v, b)$, return $v$ and $\text{parity} = 0$. If it returns $(v, b)$ and a component of size $2\ell$ containing $(v, b)$, return $v$ and $\text{parity} = 1$. Otherwise fail.
\end{proof}

The bound provided by Theorem~\ref{thm:comp-lb} is {\bf tight up to constant factors in the exponent}. We give an algorithm with $k^{O(k)}$ space complexity in Section~\ref{sec:comp}:
\begin{theorem}[Component collection upper bound; informal version of Theorem~\ref{theorem:compfinding}]
There exists a component collection algorithm in (hidden-batch) random order streams with space complexity $k^{O(k)}$ (words). 
\end{theorem}
Similarly to the work of~\cite{PS18}, our component collection algorithm can be used to estimate the number of connected components to additive precision $\varepsilon n$. The space complexity of our estimation algorithm is $(1/\varepsilon)^{O(1/\varepsilon)}$, similarly improving upon on~\cite{PS18}. The details are provided in Section~\ref{sec:comp}.

Similarly, we obtain exponential lower bounds for the random walk generation problem:

\begin{theorem}[Random walk generation lower bound; informal version of Theorem~\ref{thm:rwlb}]
Generation of a random walk of length $k$ started at any vertex in a graph given as a hidden-batch random order stream requires $k^{\Omega(\sqrt{k})}$ space. Generation of $C\cdot 4^k$ random walks for a sufficiently large constant $C>0$ requires $k^{\Omega(k)}$ space.
\end{theorem}
\begin{proof}
For the first lower bound, let $\ell=\sqrt{k/C}$ for a sufficiently large absolute constant $C$, so that $k=C\ell^2$. Generate a walk of length $k$  with precision $\e=1/10$ in total variation distance.  The walk loops around the cycle that it starts in with probability at least $2/3$. Let $(v, b)$ denote the starting vertex. If the cycle is of length $\ell$, output $v$ and $\text{parity} = 0$.  If the cycle is of length $2\ell$, output $v$ and $\text{parity}=1$.  Thus, random walk generation requires at least $\ell^{\Omega(\ell)}=k^{\Omega(\sqrt{k})}$ space.

For the second bound, let $\ell=k/2$ and run $C4^k=C2^{2\ell}$ random walks  of length $k$ started at uniformly random vertices, with precision $1/10$ in total variation distance (for the joint distribution), for a sufficiently large constant $C>0$. With probability at least $2/3$ at least one of the walks will loop around the cycle that it started in. Let $(v, b)$ denote the starting vertex.  If the cycle is of length $\ell$, output $v$ and $\text{parity} = 0$.  If the cycle is of length $2\ell$, output $v$ and $\text{parity}=1$.
\end{proof}

\paragraph{Boolean Fourier Analysis for Many-Player Games} Our lower bound for
\SC{} is based on the techniques of \emph{Boolean Fourier analysis}. The
application of these techniques to one-way communication complexity goes back
to~\cite{GavinskyKKRW07}, but previous applications have either involved two
players or at most a small number relative to the size of the input, meaning
that they can afford to lose factors polynomial or even exponential in the
player count. Our application involves $n$ players for an $\bOt{n}$-sized input,
requiring a careful consideration of how the Fourier coefficients associated
with the players' messages evolve as each player passes to the next. We give an
overview of these techniques in Sections~\ref{sec:warmup}
and~\ref{sec:overview}.

\subsection{Related work}
The random order streaming model has seen a lot of attention recently. Besides the aforementioned work of~\cite{PS18} that gives small space algorithms for component counting, small space approximations to matching size have been given in~\cite{KapralovKS14,CormodeJMM17,MonemizadehMPS17,KapralovMNT20} (naturally, the problem has also attracted significant attention in adversarial streams, but the space complexity of known algorithms in this model is significantly higher than in random order streams~\cite{EsfandiariHLMO15,BuryS15,AssadiKL17,V18,BuryGMMSVZ19,McGregorV16,ChitnisCEHMMV16,EsfandiariHM16}). The work of~\cite{MonemizadehMPS17} shows that constant query testable graph properties can be tested in constant space in random order streams in bounded degree graphs.

The Fourier-analytic approach to proving communication complexity lower bounds
pioneered by~\cite{GavinskyKKRW07} has been instrumental in lower bounds for
many graph problems, including cycle counting~\cite{VerbinY11}, estimating
MAX-CUT value~\cite{KapralovKS15,KapralovKSV17,KapralovK19} more general
CSPs \cite{GuruswamiVV17,GuruswamiT19,CGV20}, and subgraph
counting~\cite{KallaugherKP18}. Communication problems inspired by the Boolean
hidden matching problem (and therefore related to the \textsc{StreamingCycles}
problem that forms the basis of our lower bound) have also been recently used
to obtain lower bounds for multipass algorithms for several fundamental graph
streaming
problems~\cite{AssadiKL17,AssadiKSY20,DBLP:journals/corr/abs-2104-04908}. Lower
bounds for statistical estimation problems (e.g., distinct elements, frequency
moments and quantile estimation) in random order streams were given
in~\cite{ChakrabartiCM08,ChakrabartiJP08}.

\section{Warm-up: Boolean Hidden Hypermatching with interleaving}\label{sec:warmup}
Let $\Xb \in \{0,1\}^n$ be uniformly random, and revealed one bit at a
time to our algorithm.  Let our algorithm's state at time $t$ be
$\Mb_t$, which is at most $c$ bits long.  The Fourier-analytic lower
bound approach studies quantities corresponding to the following question: what
is known about the \emph{parity} of each set of bits at time $t$?  For the
indicator $z \in \{0, 1\}^t$ of any subset of bits that arrive before time $t$,
define 
\begin{align}\label{eq:Ft}
  \wt{\Fb}_t(z) \coloneqq \E[x \mid \Mb_t]{(-1)^{z \cdot x}}
\end{align}
which is $\pm 1$ if the corresponding parity is specified by the
message and $0$ if it is completely unknown.  One can think of
$\wt{\Fb}_t(z)^2 \in [0, 1]$ as an estimate of how well the parity $z$
is remembered at time $t$.  A method that stores $c$ individual bits
would have
\[
  \forall k \le c, \sum_{\abs{z} = k} \wt{\Fb}_t(z)^2 = \binom{c}{k},
\]
where $\abs{z}$ denotes the Hamming weight of $z$, because it remembers
exactly each subset of those bits.  The foundation of the
Fourier-analytic approach is that a similar bound typically holds for
\emph{any} protocol that generates $c$-bit messages $\Mb_t$:
\begin{align}\label{eq:KKL}
  \forall k \le c, \sum_{\abs{z} = k} \wt{\Fb}_t(z)^2 \leq \binom{O(c)}{k}
\end{align}
with very good probability over $\Mb_t$.  This inequality (Lemma~3
in~\cite{GavinskyKKRW07}) is a consequence of the hypercontractive inequality
(see Lemma~3.4 in~\cite{KKL88}).

In Boolean Hidden Hypermatching, one first receives the bits $x$ and
then receives $\bTh{n/\ell}$ ``important'' sets $z^{(i)}$, each of size
$\ell$.  Since the sets are uniform and independent of the message
$\Mb_t$ (and so, of $\wt{\Fb}_t$), the expected amount known about
them is
\begin{align}\label{eq:desiredBHH}
  \E[z^{(i)}]{ \sum_{i \in [n/\ell]} \wt{\Fb}_t(z^{(i)})^2} = (n/\ell)
  \frac{1}{\binom{n}{\ell}} \sum_{\abs{z} = \ell} \wt{\Fb}_t(z)^2  \leq
  (n/\ell) \frac{\binom{O(c)}{\ell}}{\binom{n}{\ell}} = (n/\ell)
  \left(\frac{O(c)}{n}\right)^\ell.
\end{align}
If $c \ll n^{1 - 1/\ell}$, this is $o(1)$ so the algorithm probably
does not remember any of the important parities.

The challenge we face in adapting this approach is that our important
sets (the components of the graph) are revealed over time, interleaved
with the bits of $x$ rather than at the end.

To see how this can be an issue, consider a two-stage version of
Boolean Hidden Hypermatching: the first $n/2$ bits of $x$ are given,
then at time $s=n/2$ we receive $z_{\leq s}^{(i)}$ (elements $z^{(i)}$ with indices at most $s$) for each $i$, then the
rest of $x$, and finally at time $t=n$ we receive the rest of the
important indices $z_{[s+1:t]}^{(i)}$ (elements $z^{(i)}$ with indices between $s+1$ and $t$).  For simplicity, suppose
$\abs{z^{(i)}_{\leq s}} = \abs{z^{(i)}_{[s+1:t]}} = \ell/2$ always.  The
algorithm that stores a random subset of bits still needs $c \gtrsim n^{1 -
1/\ell}$.  Solving either half of the stream (determining the parity of one of
the half-sets $z^{(i)}_{\leq s}$, $z^{(i)}_{[s+1:t]}$) requires only $n^{1 -
2/\ell}$ space in general, but how can we get a tight $n^{1 - 1/\ell}$ bound?

The problem is that~\eqref{eq:KKL} does not give strong enough control
over the higher-order moments to show~\eqref{eq:desiredBHH}.  The sets
$z^{(i)}$ at the end are no longer independent of $\wt{\Fb}_t$, because the
algorithm's behavior in the second half can depend on the $z_{\leq s}$. One
could instead apply~\eqref{eq:KKL} to each half of the stream and take the
product, getting
\[
  \sum_{z} \wt{\Fb}_s(z_{\leq s})^2 \wt{\Fb}_t(0^sz_{[s+1:t]})^2 = \paren*{\sum_{\abs{z_{\leq s}}=\ell/2} \wt{\Fb}_s(z_{\leq s})^2}\paren*{\sum_{\abs{z_{[s+1:t]}}=\ell/2} \wt{\Fb}_t(0^s z_{[s+1:t]})^2} \leq \binom{O(c)}{\ell/2}^2
\]
and so, on average over $z$,
\[
  \sum_{i\in [n/\ell]} \wt{\Fb}_s(z_{\leq s}^{(i)})^2 \wt{\Fb}_t(0^sz_{[s+1:t]}^{(i)})^2  \leq (n/\ell) \left(\frac{O(c)}{n}\right)^\ell.
\]
For algorithms that store individual bits this
implies~\eqref{eq:desiredBHH}, since in that case
\[
\wt{\Fb}_t(z) = \wt{\Fb}_t(z_{\leq s}0^{t-s})\wt{\Fb}_t(0^sz_{[s+1:t]})
\]
and \[
\wt{\Fb}_t(z_{\leq s}0^{t-s})^2 \leq \wt{\Fb}_s(z_{\leq s})^2.
\] However, for general
algorithms, \[\wt{\Fb}_t(z) \neq \wt{\Fb}_t(z_{\leq s}0^{t-s})\wt{\Fb}_t(0^sz_{[s+1:t]}).
\]
To solve this, we need to relate $\wt{\Fb}_t(z)$ to bounds
involving $z_{\leq s}$ and $z_{[s+1:t]}$ individually.  We define
\[
  \wt{\rb}_{s,t}(z_{[s+1:t]}) \coloneqq \E[x \mid \Mb_s, \Mb_t, \Bb_t]{(-1)^{z_{[s+1:t]} \cdot x_[s+1:t]}}
\]
as a ``double-ended'' version of~\eqref{eq:Ft}: it asks about the
knowledge of $H$ given the states before and after $H$ arrives, as
well as the ``board'' $\Bb_t$ at time $t$ (which is the information
about important sets revealed by time $t$, namely the
$z_{\leq s}^{(i)}$ ).  Since $x$ is independent of $\Bb_t$, this is
specified by $2c$ bits ($\Mb_s$ and $\Mb_t$), so it also
satisfies~\eqref{eq:KKL}.  Our key observation,
Lemma~\ref{lm:decomposition}, is that
\begin{align}\label{eq:simple-decomposition}
\wt{\Fb}_t(z) = \E[\Mb_s \mid \Mb_t,\Bb_t]{\wt{\Fb}_s(z_{\leq s}) \wt{\rb}_{s,t}(z_{[s+1:t]})}.
\end{align}
This lets us relate $\wt{\Fb}_t(z)$ to $\wt{\Fb}_s(z_{\leq s})$ and
$\wt{\rb}_{s,t}(z_{[s+1:t]})$, each of which are bounded
by~\eqref{eq:KKL}.

Specifically,  for any fixed index $i$, the average amount that is remembered about $z^{(i)}$ is:
\begin{align*}
  \E[z^{(i)}, \Mb_s, \Mb_t]{\wt{\Fb}_t(z^{(i)})^2}
  &\leq \E[\Mb_s, \Mb_t]{\E[z^{(i)}]{\wt{\Fb}_s(z^{(i)}_{\leq s})^2 \wt{\rb}_{s,t}(z^{(i)}_{[s+1:t]})^2}} \\
  &= \E[\Mb_s, \Mb_t]{\E[z_{\leq s}^{(i)}]{\wt{\Fb}_s(z^{(i)}_{\leq s})^2 \E[z_{[s+1:t]}^{(i)}]{\wt{\rb}_{s,t}(z^{(i)}_{[s+1:t]})^2}}} \\
  &\leq \frac{\binom{O(c)}{\ell/2}}{\binom{n/2}{\ell/2}} \cdot \frac{\binom{O(2c)}{\ell/2}}{\binom{n/2}{\ell/2}}  = \left(\frac{O(c)}{n}\right)^\ell.
\end{align*}
Thus, for the two-stage version of Boolean Hidden Hypermatching, we
need $c \gtrsim n^{1 - 1/\ell}$ to remember one of the $n/\ell$ important sets
on average.

These equations,~\eqref{eq:simple-decomposition} and~\eqref{eq:KKL},
form the Fourier-analytic basis of our lower bound.  The rest of the challenge for
our setting comes from random order streams having much more
complicated combinatorics than two-stage hypermatching, spread across
$n$ stages.  As edges arrive, components appear, extend, and merge to
eventually form the final cycles.

\section{Technical Overview}
\label{sec:overview}
\paragraph{A basic ``sampling'' protocol.} Suppose we only permit ourselves to remember
one parity (so using one bit of space, in addition to whatever space we need to
know which parity this is). We want to eventually learn the parity of one cycle
in the stream, and so the natural strategy is as follows:
\begin{enumerate}
\item Arbitrarily choose some edge to start with, and record its parity.
\item Whenever we see a new edge that is incident to the parity we are storing,
add that edge to the parity, and hope our parity eventually grows to encompass
an entire cycle.
\end{enumerate}
This strategy will succeed with probability $\ell^{-\bTh{\ell}}$, as it only works if no
edge of the cycle arrives before a path to it from the first edge has already
arrived (we refer to such a cycle as a ``single-seed'' cycle---see Fig.~\ref{fig:singleseed} for an illustration). So we would have
to repeat this process $\ell^{\bTh{\ell}}$ times in order to achieve a constant probability of success. 

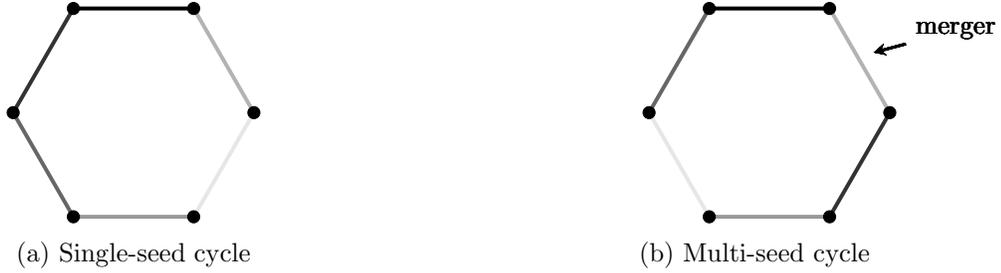
\begin{figure}
  \begin{center}
    \begin{subfigure}{0.5\textwidth}
      \centering
      \begin{tikzpicture}[scale=0.8]
        \foreach[count=\i, evaluate=\i as \x using int(\i)] \a in {1,0.8,0.6,.4,0.1,0.3} {
          \draw[fill=black!100] (\x*60: 2) circle (0.1);
          \draw[opacity=\a,line width=1.5pt] (\x*60: 2) -- (\x*60+60: 2);
        }
      \end{tikzpicture}
      \caption{Single-seed cycle}
      \label{fig:singleseed}
    \end{subfigure}%
    \begin{subfigure}{0.5\textwidth}
      \centering
      \begin{tikzpicture}[scale=0.8]
        \path (-4.5,0) -- (4,0);  
        \foreach[count=\i, evaluate=\i as \x using int(\i)] \a in {1,0.6,0.1,.4,0.8,0.3} {
          \draw[fill=black!100] (\x*60: 2) circle (0.1);
          \draw[opacity=\a,line width=1.5pt] (\x*60: 2) -- (\x*60+60: 2);
          \node[anchor=south west] (m) at (25:2.5) {\small merger};
          \draw[line width=1pt, ->, >=stealth] (m) -- (30:2);
        }
      \end{tikzpicture}
      \caption{Multi-seed cycle}
      \label{fig:multi-seed}
    \end{subfigure}
    \caption{Illustration of two possible edge arrival orders.
      Lighter edges arrive later. }
    \label{fig:seed}
  \end{center}
	
\end{figure}

Can we hope to do better by not considering the parities independently? Suppose
we maintain the parities of $c$ paths at a time, which may merge with each
other as we process the stream. We now have some chance of finding
``multi-seed'' cycles, i.e.\ cycles in which several disjoint paths arrive
before eventually being merged by later edge arrivals---see
Fig.~\ref{fig:multi-seed} for an illustration. If we happen to have remembered
the parity of each of the components that eventually merged into a given
multi-seed cycle, we will find the parity of the cycle. The chance that any $k$
of them are from the same cycle is ${\sim}\binom{c}{k}\paren*{\ell/n}^{k-1}$,
and the chance of any given cycle having only $k$ seeds is
$(\ell/k)^{-\bTh{\ell}}$, as it requires $k$ paths of average length $\ell/k$
to arrive in order. Until $c$ is $n^{1 - \bTh{1/\ell}}$, the probability of
finding the parity of a cycle will therefore by dominated by the single-seed
case.

However, so far we have assumed we can only store individual parities. To
extend these arguments to algorithms that maintain arbitrary state, we make use
of the tools of \emph{Boolean Fourier analysis}.

\paragraph{Fourier-analytic lower bound for \SC.}

We construct a hard
instance for the problem in which the bit labels $\Xb$ for the edges are chosen
uniformly at random, and in order to simplify the analysis we allow the
algorithm to remember which edges it has seen for free (although not the bit
labels). We say that these edges are posted on the ``board'' $\Bb$. The state
of the board at time $t$, i.e.\ after receiving $t$ edges, is denoted by
$\Bb_t$.

For any subset $z \in \bool^t$ of the edges that have arrived so far
(given by the appropriate bit mask), we can associate the expectation of the
parity of $z$ given $\Mb_t$ with the \emph{normalized Fourier coefficient} \[
  \wt{\Fb}_t(z) \coloneqq \E[x \mid \Mb_t]{(-1)^{z\cdot x}}.
\]

Note that if the algorithm returns $v$, and $C$ is the cycle containing $v$ (written as an
element of $\bool^n$), the algorithm's best guess for the parity of $C$ will be
$1$ if $\wt{\Fb}_n(C) > 0$ and $-1$ otherwise. Moreover, the probability that
this guess will be correct is $\frac{1 + \abs{\wt{\Fb}_n(C)}}{2}$. Therefore,
for a lower bound, it will suffice to prove that with good probability \[
 \abs{\wt{\Fb}_n(C)} = \lO{1}\text{.}
\]

\paragraph{Fourier mass on collections of component types.}  Writing $\zpp$ for the set of multisets of integers, for $\beta \in \zpp$ and a $z\in \bool^t$ we define 
$$
z \bsim_t \beta
$$ to be true iff $z$ corresponds to a set of edges which contains
$\beta\brac{a}$ components (i.e., paths or cycles) of length $a$ for each $a$
(here for $\beta\in \zpp$ and an integer $a$ we write $\beta[a]$ to denote the
number of occurrences of $a$ in $\beta$).  Our main object of study in the
lower bound proof is 
\[
\Hb_\beta^t \coloneqq \sum_{z\in \bool^t, z\bsim_t \beta} \wt{\Fb}_t(z)^2, 
\]
which can be viewed as the amount of certainty that the algorithm has about parities of unions of components whose sizes are prescribed by $\beta$. 

In order to build some intuition, we consider our prototypical sampling
protocol from the start of Section~\ref{sec:overview}. Let $c$ be the number of
edges sampled at the beginning and for every $j\geq 1$ let $Y_j^t$ denote the
number of components of size $j$ (i.e., with $j$ edges) that the sampling
algorithm was able to construct at time $t$. Then for every multiset $\beta$
and every $t$ one simply has 
\[
\Hb_\beta^t=\prod_{a} \binom{Y_a^t}{\beta[a]}.
\]
and
$$
\Hb_{\{1\}}^{c}=c,
$$ 
where the last equality holds because the sampling algorithm grows components out of the first $c$ arriving edges, and for every $j$ we have that $\Hb_{\{j\}}^t$ is the number of components of size $j$ that the algorithm knows the parity of at time $t$. In order to prove that the sampling protocol does not succeed, we need to prove that $Y_\ell^n=0$, and for general protocols we need to prove 

\begin{restatable}{lemma}{lemfinalH}
\label{lem:finalH}
For all $\varepsilon > 0$, there is a $D > 0$ depending only on $\varepsilon$
such that, if $c < \min(\ell^{\ell/D}, n^{1 - \varepsilon})$ and $D <
\ell < D^{-1}\log n $,
\[
\E[\Xb, \Bb_n]{ \Hb_{\lbrace \ell \rbrace}^n } \le \varepsilon\text{.}
\]
\end{restatable}

In order to establish Lemma~\ref{lem:finalH}, we bound the expected evolution of $\E{\Hb^t_\beta}$ as a function of $t$. In what follows we first analyze the evolution of $\Hb_\beta^t$ for the simple ``sampling'' protocol, and then present the main ideas of our analysis.

\paragraph{Evolution of Fourier coefficients of the ``sampling'' protocol.} Fix a component of size $j$. At time $t$ the probability that it gets extended is about $\frac{2}{n-t}$, and therefore 
$$
\expect[Y^t_j|\Bb_{t-1}, \Fb_{t-1}]\approx Y^{t-1}_j+\frac2{n-t}\cdot Y^{t-1}_{j-1}+\text{(contribution from merges of smaller components)}.
$$
In order to derive the asymptotics of $Y^t_j$, we first ignore the contribution of merges, and later verify that they do not affect the result significantly. In particular, ignoring the contribution of merges, we get
$$
\expect[Y^t_j|\Bb_{t-1}, \Fb_{t-1}]\approx Y^{t-1}_j+\frac2{n-t}\cdot Y^{t-1}_{j-1}.
$$

We assume for intuition that $t\leq n/2$, i.e .we are only looking at the first
half of the stream.  Since the initial conditions are (essentially) $Y^1_1=c$
and $Y^1_j=0$ for $j>1$, because the algorithm can remember $c$ single edges at
the beginning of the stream, and no larger components (since they typically do
not form at the very beginning of the stream). This now yields that \[
Y^t_j\leq c\cdot 4^{j-1} (t/n)^{j-1}/(j-1)!
\] 
for $t\leq n/2$ and all $j\geq 1$. This is
because $Y^1_j=c$ as required, and for $j\geq 2$
\begin{equation}\label{eq:rec-y}
\begin{split}
\expect[Y^t_j]&\leq \sum_{s=1}^{t-1} \frac2{n-s} Y^s_{j-1}\\
&\leq (c\cdot 4^{j-1}/(j-2)!) \frac{1}{n}\sum_{s=1}^{t-1} (t/n)^{j-2}\\
&\approx c\cdot 4^{j-1}/(j-2)!\cdot \int_0^{t/n} x^{j-2}dx\\
&=c \cdot 4^{j-1} (t/n)^{j-1} /(j-1)!
\end{split}
\end{equation}
This in particular implies that $Y_\ell^{t/2}\ll 1$ if $c=\ell^{o(\ell)}$, and in general that for the sampling protocol we have, at least for $t\leq n/2$,
\begin{equation*}
\begin{split}
\expect[\Hb_\beta^t]&=\expect\left[\prod_{j\in \beta} Y_j^t\right]\\
&\approx \prod_{j\in \beta} c \cdot 4^{j-1} (t/n)^{j-1} /(j-1)!\\
&\lesssim \left(\prod_{j\in \beta} \frac{1}{j!}\right) \cdot
Q^{\abs{\beta}_*} \cdot \left(\frac{t}{n}\right)^{\abs{\beta}_* -
\abs{\beta}}\cdot c^{\abs{\beta}},
\end{split}
\end{equation*}
where $Q$ is an absolute constant, $\abs{\beta}_*=\sum_{i\in \beta} i$ and
$\abs{\beta}$ is the number of elements in the multiset $\beta$ (counting
multiplicities).

\paragraph{Evolution of Fourier coefficients of a general protocol.} The outline of the simple ``sampling'' protocol above provides a good model for our general proof. Specifically, in Lemma~\ref{lem:hbetabound} (see Section~\ref{sec:main-lemma}) we show that there exists a constant $Q>0$ such that for (almost) all $\beta$ and $t$ (the near-end of
the stream and going from $\beta = \set{\ell-1}$ to $\beta = \set{\ell}$
require some special treatment) one has

\begin{equation}\label{eq:evolution}
\E{\Hb^t_\beta} \lesssim \left(\prod_{j\in \beta} \frac{1}{j!}\right) \cdot
Q^{\abs{\beta}_*} \cdot \left(\frac{t}{n}\right)^{\abs{\beta}_*-\nu(\beta)}\cdot c^{|\beta|},
\end{equation}
where $\abs{\beta}_*=\sum_{i\in \beta} i$ and $\nu(\beta)=\sum_{i\in \beta} \lceil \frac{i}{2}\rceil$. Lemma~\ref{lem:finalH} then follows by essentially summing the above bound over all component types (the proof is presented in Section~\ref{sec:main-lemma}).

The result of Lemma~\ref{lem:finalH} can then be seen to imply that the chance of successfully guessing the parity of any cycle is $\lO{1}$ whenever $c= \ell^{o(\ell)}$, and specifically that the \SC~problem requires $\ell^{\bOm{\ell}}$ space. 

 To establish~\eqref{eq:evolution}, we bound the (expected) evolution of $\Hb^t_\beta$ as a
function of $t\in [n]$. Specifically, we show in Section~\ref{sec:evolution}
that for every $t\in [n]$ the expectation of $\Hb^t_\beta$ can be upper bounded
in terms of expectations of $\Hb^s_\alpha$ for $s<t$ and $\alpha$ corresponding
to ``subsets'' of $\beta$ (see Section~\ref{sec:basic-defs} for the formal
definitions). This is a natural extension of our analysis of the ``sampling''
protocol above.  In full generality, however, this requires showing that if
the algorithm has limited information about parities of collections of type
$\alpha$ at time $s$ (i.e.,  $\Hb^s_\alpha$ is small), then it is unlikely to
know too much about collections of type $\beta$ obtained as a result of merging
several components in $\alpha$ or extending them by  edges arrived between $s$
and $t$. 

Crucially, the probability that a collection of components of type $\alpha$
grows into a collection of components of type $\beta$ at any given time depend
only on the \emph{collection type} (i.e., the multisets $\alpha$ and $\beta$).
Specifically,  for $s \in [n]$, a pair of collection types $\alpha, \beta\in
\zpp$ such that $\alpha\lbrack 1\rbrack = \beta\lbrack 1\rbrack$ (the number of
single edge components in $\alpha$ and $\beta$ is the same) and a realization
$B_s$ of the board $\Bb_s$ at time $s$ we write
\[
p_s(\alpha, \beta, B_s) = \Pb[\Bb_{s+1}]{ z\cdot 1 \bsim_{s+1} \beta | \Bb_s =
B_s}
\] for any $z \in \lbrace 0, 1\rbrace^s$ such that $z \bsim_s \alpha$ to denote
the probability that a collection of type $\alpha$ at time $s$ becomes a
collection of type $\beta$ at time $s+1$ through one of the following ``growth events'':
\begin{description}[labelindent=\parindent]
\item[Extension] An edge arrives at time $s+1$ that is incident to exactly one
component in the collection.
\item[Merge] An edge arrives at time $s+1$ that is incident to two components
in the collection.
\end{description}
We will use 
\begin{lemma}[Informal version of Lemma~\ref{lm:ext-prob}]\label{lm:ext-prob-inf}
With high probability over the board state $\Bb_s$, for $s$ not too close to $n$ one has for (almost) every $\alpha, \beta$ 
\[ 
p_{s}(\alpha, \beta, \Bb_s) \le \begin{cases}
\frac{O(\alpha\lbrack a\rbrack)}{n} &\mbox{if $\alpha\to \beta$ is an extension
of a path of size $a$}\\
\frac{O(\alpha \lbrack a\rbrack \cdot \alpha \lbrack b\rbrack)}{(n - s)^2}
&\mbox{if $\alpha\to \beta$ is a merge of paths of size $a$ and $b$.}
\end{cases}
\]
\end{lemma}

We will also need

\begin{definition}[Down set of $\beta\in \zpp$]
For $\alpha, \beta\in \zpp$ we write $\alpha\in \beta-1$ if $\beta$ can be
obtained from $\alpha$ by either an {\em extension} or a {\em merge}  followed
by possibly adding an arbitrary number of $1$'s to $\alpha$.

For $\beta\in \zpp$ and $\alpha\in \beta-1$ we write $|\beta-\alpha|$ to denote the number of ones that need to be added to $\alpha$ after a merge or extension to obtain $\beta$.
\end{definition}

Equipped with the above, and writing $T$ for the set of all edge arrival times,
we can state our main bound on the evolution of Fourier coefficients:
\begin{restatable}{lemma}{lmevolution} \label{lm:evolution}
For every $t\in T$ one has for $\beta\in \zpp$ that contain at least one component of
size more than $1$ and have $\abs{\beta}_* \le \ell - 2$,
\begin{equation*}
\begin{split}
\E[\Xb, \Bb_t] {\Hb^t_\beta}&\leq \sum_{s = 1}^{t-1} \sum_{\alpha\in
\beta-1} q(|\beta-\alpha|)\cdot\E[\Xb, \Bb_s]{\Hb_\alpha^s\cdot 
p(\alpha, \beta, \Bb_s)}\\
\end{split}
\end{equation*}
and for $\beta$ with all components of size $1$
\begin{equation*}
\begin{split}
\E[\Xb, \Bb_t]{ \Hb^t_\beta } &\leq q(|\beta|),
\end{split}
\end{equation*}
where 
\[
q(k) = \begin{cases}
\paren*{\frac{8c}{k}}^k & \mbox{if $k \le c$}\\
2^c & \mbox{otherwise.}
\end{cases}
\]
\end{restatable}
The function $q(k)$ comes from hypercontractivity (as
in~\eqref{eq:KKL},~\cite{KKL88}), and bounds the total amount of
Fourier mass a $c$-bit message can place on size-$k$ parities; when
$\alpha \to \beta$, $q(|\beta - \alpha|)$ appears because the term
involves remembering $|\beta - \alpha|$ of the isolated edges
that arrive between $s$ and $t$.

The proof of Lemma~\ref{lm:evolution} is based on a function $\rb$ that
functions similarly to the $\rb_{s,t}$ used in the Section~\ref{sec:warmup}
warm-up.  Lemma~\ref{lm:evolution} allows us to bound the Fourier mass on
various collections of components as a function of their evolution in the
stream. We consider two prototypical examples now.

\noindent{\bf Example 1: single-seed components.}  To obtain some
intuition for Lemma~\ref{lm:evolution}, we first consider a simplified
setting where $p(\alpha, \beta, \Bb_s)=0$ unless $\alpha\to \beta$ is
an extension, i.e., we ignore the effect of merges.  Without merges,
since we eventually care about the single-element set $\{\ell\}$, we
only need to track the mass on other single-element sets
$\beta = \{a\}$, so $\beta - 1 = \{a - 1\}$.  Then

\begin{equation*}
\begin{split}
\E[\Xb, \Bb_t] {\Hb^t_\beta}&\leq \sum_{s = 1}^{t-1} \sum_{\substack{\alpha\in
\beta-1\\ \alpha\to \beta \text{~is an extension}}} q(|\beta-\alpha|)\cdot\E[\Xb, \Bb_s]{\Hb_\alpha^s\cdot 
p(\alpha, \beta, \Bb_s)}\\
&= \sum_{s = 1}^{t-1} q(0) \E[\Xb, \Bb_s]{\Hb_{\{a - 1\}}^s\cdot 
p(\{a-1\}, \{a\}, \Bb_s)}\\
&\leq \sum_{s = 1}^{t-1} 1 \cdot \frac{O(1)}{n} \E[\Xb, \Bb_s]{\Hb_{\{a-1\}}^s}.
\end{split}
\end{equation*}
where the last step uses the first bound from
Lemma~\ref{lm:ext-prob-inf}.  One observes that the above recurrence
is quite similar to~\eqref{eq:rec-y} and can be upper bounded
similarly by $c \frac{2^{O(a)}}{a!}(t/n)^a$---as needed for~\eqref{eq:evolution}.

\noindent{\bf Example 2: multi-seed components.} To illustrate the way
Lemma~\ref{lm:evolution} handles merges, suppose we want to bound
$H_{\{4\}}$.  There are three paths from $\{4\}$ via down-set
relations to our base cases:
\begin{align*}
  \{4\} \to \{3\} \to \{2\} \to \{1\}\\
  \{4\} \to \{3\} \to \{1, 1\}\\
  \{4\} \to \{2,1\} \to \{1\}
\end{align*}
The first is a series of extensions, so bounded by about $c/a!$
according to Example 1; the other two involve merges, and we show give
negligible contribution.  We show this for the
$\{4\} \to \{2,1\} \to \{1\}$ path here.

The only path to $\beta = \{2, 1\}$ is an extension from $\{1\}$, so
Lemma~\ref{lm:evolution} and Lemma~\ref{lm:ext-prob-inf} show
\begin{equation*}
\begin{split}
\E[\Xb, \Bb_t] {\Hb^t_\beta}&\leq \sum_{s = 1}^{t-1} \sum_{\alpha\in
\beta-1} q(|\beta-\alpha|)\cdot\E[\Xb, \Bb_s]{\Hb_\alpha^s\cdot 
p(\alpha, \beta, \Bb_s)}\\
&= \sum_{s = 1}^{t-1} q(1)\cdot\E[\Xb, \Bb_s]{\Hb_{\{1\}}^s\cdot 
p(\{1\}, \{2, 1\}, \Bb_s)}\\
&\leq \sum_{s = 1}^{t-1} 8c \cdot \frac{O(1)}{n}\cdot\E[\Xb, \Bb_s]{\Hb_{\{1\}}^s}\\
&\lesssim c^2 (t/n).
\end{split}
\end{equation*}
Then the contribution to $\E[\Xb, \Bb_s]{\Hb^t_\beta}$ for
$\beta = \{4\}$ from the $\{4\} \to \{2,1\} \to \{1\}$ path is at most
\begin{equation*}
\begin{split}
 \sum_{s = 1}^{t-1} q(|\{4\}-\{2,1\}|)\cdot\E[\Xb, \Bb_s]{\Hb_{\{2,1\}}^s\cdot 
p(\{2,1\}, \{4\}, \Bb_s)}
&\leq \sum_{s = 1}^{t-1} 1 \cdot O(c^2 (s/n)) \cdot \frac{O(1)}{(n-s)^2}\\
&\lesssim c^2\sum_{s = 1}^{t-1} \frac{1}{(n-s)^2}\\
&\eqsim \frac{c^2}{n-t}
\end{split}
\end{equation*}
For almost the entire stream, say $t < n - n^{0.99} < n - \omega(c)$, this
term is much less than the $\Theta(c)$ contribution from extensions.
The fact that merges are about a $c/n$ factor less likely can be used
in general to show that the evolution is ultimately dominated by
extensions and establish~\eqref{eq:evolution}, for $t < n - n^{0.99}$.

\paragraph{Handling the end of the stream.}  Our
bound~\eqref{eq:evolution} gives roughly a $\frac{c}{a!}$ bound on
$H_{\{a\}}^t$, but only up to time $t = n - n^{0.99}$.  By this time,
however, probably every single cycle will be missing only $O(1)$
edges, and so have at most $O(1)$ components.  Thus each cycle will
have an $\Omega(\ell)$-long component that has arrived,
and~\eqref{eq:evolution} gives a
$\frac{c^{O(1)}}{\Omega(\ell)!} \ll \eps / \poly(\ell)$ bound for the
Fourier mass of its collection type at $t$.  Summing over the
$\poly(\ell)$ possible types leads to Lemma~\ref{lem:finalH}.

We now discuss a key technical insight that allows us to establish
Lemma~\ref{lm:evolution}.  It is the one illustrated in the warm-up
example of Section~\ref{sec:warmup}.

\paragraph{Key tool in proving Lemma~\ref{lm:evolution}: decomposition of a typical message.} 
In order to establish Lemma~\ref{lm:evolution}, we need an approach to
expressing the Fourier transform of the typical message $\Fb_t$ at
time $t$ in terms the Fourier transform of the typical messages
$\Fb_s$ for $s<t$. This is achieved by Lemma~\ref{lm:decomposition}
below. Intuitively, this lemma allows us to exploit communication
bottlenecks arising at every $s<t$ that preclude various Fourier
coefficients from becoming large.

Let  $\rb(x_{[s+1:t]}; F_s, F_t)$ denote the indicator function for
$x_{[s+1:t]}$ taking $F_s$ to $F_t$. Note that $\rb$ is a random function
depending only on $\Bb_t$ ($F_s$ is a function on $s$ bits and $F_t$ on $t$
bits, so its dependence on $\Bb_t$ is given implicitly by its arguments).  The
decomposition of typical messages is given by

\begin{restatable}{lemma}{lmdecomposition} \label{lm:decomposition}
For every $s, t\in T$ with $s < t$, and any $z_{\leq t}$, we have:
\[
\wt{\Fb}_t(z_{\leq t})=\E[\Fb_s]{\wt{\Fb}_s(z_{\leq s})
\cdot \wt{\rb}(z_{[s+1:t]}; \Fb_s, \Fb_t) \middle| \Fb_t, \Bb_t}.
\]
\end{restatable}

Since $\rb$ is a function partitioning its input based on at most $2c$ bits,
its normalized Fourier transform $\wt{\rb}$ is subject to a bound similar
to~\eqref{eq:KKL}, which enables us to establish Lemma~\ref{lm:evolution}.

\subsection{Hidden-Batch Random Order Streams}\label{sec:hidden-batch}

We introduce a new random order  streaming model that allows
for (limited) correlations whose structure is unknown to the algorithm (the
hidden-batch random order streaming model, Definition~\ref{def:model-informal} below).  We give
new algorithms for estimating local graph structure using small space in this
model, and show that existing results in this space translate to our model with
only a very mild loss in parameters.

\begin{definition}[Hidden-batch random order stream model; informal]\label{def:model-informal}
In the $(b,w)$-hidden-batch random order stream model the edges of the input graph $G=(V,
E)$ are partitioned adversarially into batches of size bounded by $b$, after
which every batch is presented to the algorithm in a time window of length
$w\geq 0$ starting at a uniformly distributed time in the interval $[0, 1]$.
\end{definition}
To motivate this model consider observing, say, a network traffic stream or a
stream of friendings in a social network. In each case, there are many events
(say, a login attempt, or a group of people meeting each other at a party) that
will trigger a bounded number of updates (the back and forward of packets in a
login protocol, or people adding friends they met) that might have very
complicated temporal correlations with each other, but that occur over a
bounded period and are mostly independent of other events being observed in the
stream.

To simulate this, we think of the division of observations into events (our
batches) being adversarial but limited by a maximum batch size $b$, while the
times of the events are chosen at random but the observations associated with
the events are adversarially distributed about the event time, subject to the
event duration limit $w$. Note that this means that observations from multiple
events may (and often will be) interleaved---more than one person may be
logging onto the same network at the same time and more than one party may
be taking place at once.

It is worth stressing that the partitioning of edges into batches is unknown to
the algorithm (consequently, we refer to our model as the {\em hidden} batch
model). 

Some existing random order streaming algorithm can be readily ported to the
hidden-batch random order streaming model.

\begin{restatable}[Component Collection; informal version of Theorem~\ref{theorem:compfinding}]{theorem}{compfinding-inf}
\label{theorem:compfinding-inf}
There is a $(b,w)$-hidden batch streaming algorithm that, if at least a
$\Omega(1)$ fraction of the vertices of $G$ are in components of size at most
$\ell$, returns a vertex in $G$ and the component containing it with
probability $9/10$ over its internal randomness and the order of the stream,
using $\ell^{\bO{\ell}}(b + wm)\plog n$
bits of space.
\end{restatable}

We show that results of~\cite{PS18} on counting connected components in random graph streams can be easily extended to our  hidden-batch random order model with only a mild loss in parameters. Specifically, in~\cite{PS18} it was shown that the
number of connected components $c(G)$ in a graph $G$ can be approximated up to an
$\varepsilon n$ additive term using
$\paren{1/\varepsilon}^{\bO{1/\varepsilon^3}}$ words of space.  We show that
this can be improved to $\paren{1/\varepsilon}^{\bO{1/\varepsilon}}$ and that
(up to log factors) it can be extended to hidden-batch streaming with only
linear loss in the parameters.

\begin{restatable}[Counting Components; informal version of Theorem~\ref{theorem:compcounting}]{theorem}{compcounting-inf}
\label{theorem:compcounting-inf}
For all $\varepsilon \in (0, 1)$, there is a $(b,w)$-hidden batch
streaming algorithm that achieves an $\varepsilon n$ additive approximation to
$c(G)$ with $9/10$ probability, using $\paren{1/\varepsilon}^{\bO{1/\varepsilon}} (b + wm)\plog(n)$ bits of space.
\end{restatable}
We note that the $wm$ term in the space complexity corresponds to the expected number of edges arriving in a given time window of length $w$. Since the arrival times of edges are adversarially chosen in a window of length $w$ started at the arrival time of the corresponding batch, it is natural to expect the algorithm to store these edges.

Our algorithm above, similarly to the approach of~\cite{PS18}, proceeds by first sampling a few nodes in the graph uniformly at random, and then constructing connected components incident on those nodes explicitly. Such a sample of component sizes for a few vertices selected uniformly at random from the vertex set of the input graph can then be used to obtain an estimator for the number of connected components.  The details are given in Section~\ref{sec:comp}.

\section{Lower Bounds}\label{sec:lb}

In this section we prove our main result, which is a lower bound for
$\SC$, even when the edge arrival order and bit labels are chosen uniformly at
random.
\begin{restatable}[$\SC$ Lower Bound]{theorem}{cycleslb}
\label{thm:cycleslb}
For all constants $\eps > 0$, solving the distributional version of the
\SC$(n, \ell)$ problem with probability at least $2/3$ requires at least one
player to send a message of size at least $\min\paren*{\ell^{\Omega(\ell)},n^{1
- \varepsilon}}$.
\end{restatable}

We use this to prove a lower bound for the component collection problem. 
\begin{definition}[Component Collection]\label{def:comp-collection}
In the $(\beta, \ell)$ component collection problem, we are given a graph
$G$ as a stream of edges, with at least $\beta |V(G)|$ of its vertices in
components of size at most $\ell$, and we must return a vertex $v \in V(G)$ and
the size of the component containing $v$.
\end{definition}

Our lower bound is given by Theorem~\ref{thm:compestlb}:
\begin{restatable}[Component Collection Lower Bound]{theorem}{compestlb}
\label{thm:compestlb}
For all constants $C, \varepsilon > 0$, solving the $(1,\ell)$ component
estimation problem in the $(2, 0)$-batch random order streaming model with
probability at least $2/3$ requires at least $\min\left(\ell^{\Omega(\ell)},
n^{1 - \varepsilon}\right)$ space.
\end{restatable}

The theorem gives a tight lower bound for the $(1,\ell)$ component collection problem in $(2,0)$-hidden batch
streams. 

We also prove a lower bound for the random walk generation problem in random streams, which we define formally first. Our definition matches the one in~\cite{KallaugherKP21}. 

\begin{definition}[Pointwise $\e$-closeness of
distributions]\label{def:eps-close} We say that a distribution $p\in
\R_+^{\mathcal U}$ is $\e$-close pointwise to a distribution $q\in
\R_+^{\mathcal U}$ if for every $u\in \mathcal U$ one has $$
p(u)\in [1-\e, 1+\e]\cdot q(u).
$$
\end{definition}

We now define the notion of an $\e$-approximate sample of a $k$-step random walk:
\begin{definition}[$\e$-approximate sample]\label{def:eps-sample}
Given $G=(V, E)$ and a vertex $u\in V$ we say that $(X_0, X_1,\ldots, X_k)$ is an {\em $\e$-approximate sample} of the $k$-step random walk started at $u$ if 
the distribution of $(X_0, X_1,\ldots, X_k)$ is $\e$-close pointwise to the distribution of the $k$-step walk started at $u$ (see Definition~\ref{def:eps-close}).
\end{definition}

\begin{definition}[Random walk generation]\label{def:rw-generation}
In the $(k, s, \e, \delta)$-random walk generation problem one must, given a graph $G=(V, E)$ presented as a stream, generate $s$ independent $\e$-approximate samples of the walk of length $k$ in $G$ started at a uniformly random vertex, with error bounded by $\delta$ in the total variation distance. 
\end{definition}
The work of~\cite{KallaugherKP21} designs a primitive that outputs such walks
 using space $(1/\e)^{O(k)}2^{O(k^2)} s$, with $\delta=1/10$, say.

\begin{restatable}[Random Walk Generation Lower Bound]{theorem}{rwlb}
\label{thm:rwlb}
There exists an absolute constant $C>1$ such that for sufficiently large $k\geq 1$, {\bf (1)} solving the $(k, 1, 1/10, 1/10)$ component
estimation problem in the $(2, 0)$-batch random order streaming model requires at least $\min\left(k^{\Omega(\sqrt{k})},
n^{0.99}\right)$ space and {\bf (2)} solving the $(k, C4^k, 1/10, 1/10)$ random walk generation
 problem in the $(2, 0)$-batch random order streaming model requires at least $\min\left(k^{\Omega(k)},
n^{0.99}\right)$ space. 
\end{restatable}

We start by setting up basic notation in Section~\ref{sec:notation}, then
define our communication problem, namely the \SC{} problem, in
Section~\ref{sec:sc}. Notation relating the underlying graph in the \SC{} problem
to the observable graph (the graph that the players see on their common board)
together with basic results on this relation is presented in
Section~\ref{sec:oaugraph}. Definitions relating to properties of collections
of components observed on the board, such as the definition of extension and
merge events, are presented in Section~\ref{sec:basic-defs}. Technical lemmas
on the probabilities of these events are presented in Section~\ref{sec:comb}. A
key lemma on the decomposition of the typical message of a player in the
\SC{} problem is presented in Section~\ref{sec:decomp}. The main recurrence on
the evolution of Fourier coefficients throughout the stream is obtained in
Section~\ref{sec:evolution}, and solved in Section~\ref{sec:main-lemma}, where
we prove of our main lemma (Lemma~\ref{lem:finalH}).  We use this to prove
Theorem~\ref{thm:cycleslb} in Section~\ref{sec:cycleslb}. Finally, the proof of
Theorem~\ref{thm:compestlb} is presented in Section~\ref{sec:main-thm}, and the
proof of Theorem~\ref{thm:rwlb} is presented in Section~\ref{sec:rwlb}.

\subsection{Notation}\label{sec:notation}

We use the notation $\brac*{a} = \lbrace 1, 2, \dots, n \rbrace$ for positive
integers $a$.  For strings $x, y$, we use $x \cdot y$ to denote the concatenation of $x$ and
$y$. For $x \in \bool^n$ and $s, t \in \brac*{n}$, we write $x_{\brac*{s:t}}$ for
the substring of $x$ starting at index $s$ and ending at index $t$ and $x_{\le
t}$ for $x_{[1:t]}$.

We will use standard permutation notation, including $S_n$ to denote the set of
permutations of $\lbrack n \rbrack$. For any set $A$ and permutations $\pi,
\phi : A \rightarrow A$ we will use $\pi\phi : A \rightarrow A$ to denote their
composition. For any such permutation $\pi$ and a tuple $(u,v) \in A \times A$
or set $B \subset A$, we will write 
\begin{align*}
\pi((u,v)) &= (\pi(u), \pi(v))\\
\pi(B) &= \lbrace \pi(B) : a \in A \rbrace
\end{align*}
For any set $A$, we will use $\Uc(A)$ to denote the uniform distribution on
$A$.   For any predicate $p$, we will use $\1bb(p)$ to denote the variable that is $1$
if $p$ is true and $0$ otherwise.

\subsection{The \SC{} problem}\label{sec:sc}
The \SC($n, \ell$) problem is an $n$-player one way communication problem
defined as follows. Informally, in this game every player is presented with an
edge $e=(u, v)\in \binom{V}{2}$ together with $x_u$ (here $x \in \{0,
1\}^V$), with the promise that the union of the edges can be partitioned into
$n/\ell$ disjoint cycles of length $\ell$. The $n\nth$ player must output a
cycle $C$ in the graph together with the parity of the cycle, namely
$\sum_{u\in C} x_u$. 

We consider a distributional version of this problem in which the bit string
$x$ is chosen uniformly at random, the edges are assigned to the players
uniformly at random, and the partitioning of the edge set into cycles is not
known to the players in advance (otherwise the problem becomes trivial).  The
edges given to the players are public input (they are posted on a board visible
to all players), whereas the bits $x_u$ are private inputs. For every $t\in
[n]$ the $t\nth$ player receives a message of $c$ bits from the $(t-1)\nth$
player and sends a message of $c$ bits to the $(t+1)\nth$ player. The message
from the $0\nth$ player is the zero string of  length $c$.  We define the
problem formally in what follows.

\paragraph{Underlying graph $G$.} Let the \emph{underlying graph} $G = (V,E)$ with $V = \lbrack n \rbrack$ be defined as follows. 
We first define $\next : V \rightarrow V$ as \[
\next(j) = \begin{cases}
j+1 & \mbox{if $j \not\equiv 0 \bmod \ell$}\\
j - \ell + 1 & \mbox{otherwise.}
\end{cases}
\]
Then for each vertex $v \in V$ define an edge 
\begin{equation}\label{eq:e-v}
e_v = (v, \next(v))
\end{equation}
and set \[
E = \lbrace e_v \rbrace_{v \in V}
\]
The graph $G$ will be revealed to the players over $n$ time steps. We will use $T = \lbrack n\rbrack$ to denote the
set of these time steps. The labels of the vertices of $G$ are permuted before they are presented to the players. We define the permutation now.

\paragraph{Permuting labels of vertices.} Choose two permutations $\bpi : T \rightarrow V$, $\bsigma : V \rightarrow V$
uniformly at random. These two permutations will determine, for each $t \in T$,
the edge that is written on the board at time $t$ (we will also say the edge
``arrives'' at time $t$, and that it is ``present'' at all $t' \ge t$). Specifically, the edge arriving at time $t$ is 
$$
\bb_t = \bsigma(e_{\bpi(t)}).
$$

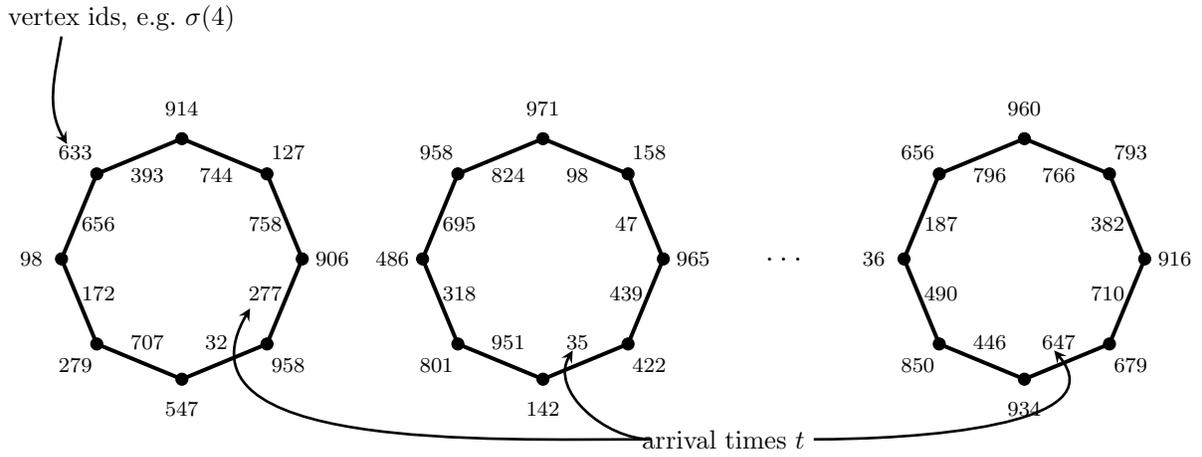
\begin{figure}
	\begin{center}
		\begin{tikzpicture}[scale=0.8]
		
		\draw (-1, 3) node {\large underlying graph $G$};
		
		\begin{scope}[shift={(-7, 0)}]
			\draw[line width=1.5pt] (0*45: 2) -- (1*45:2);
			\draw[line width=1.5pt] (1*45: 2) -- (2*45:2);
			\draw[line width=1.5pt] (2*45: 2) -- (3*45:2);
			\draw[line width=1.5pt] (3*45: 2) -- (4*45:2);
			\draw[line width=1.5pt] (4*45: 2) -- (5*45:2);
			\draw[line width=1.5pt] (5*45: 2) -- (6*45:2);
			\draw[line width=1.5pt] (6*45: 2) -- (7*45:2);
			\draw[line width=1.5pt] (7*45: 2) -- (8*45:2);
			
			\draw[fill=black!100] (0*45: 2) circle (0.1);
			\draw[fill=black!100] (1*45: 2) circle (0.1);			
			\draw[fill=black!100] (2*45: 2) circle (0.1);
			\draw[fill=black!100] (3*45: 2) circle (0.1);			
			\draw[fill=black!100] (4*45: 2) circle (0.1);
			\draw[fill=black!100] (5*45: 2) circle (0.1);			
			\draw[fill=black!100] (6*45: 2) circle (0.1);
			\draw[fill=black!100] (7*45: 2) circle (0.1);			
			
			\draw (0*45: 2.5) node {\scriptsize $1$};
			\draw (1*45: 2.5) node {\scriptsize $2$};
			\draw (2*45: 2.5) node {\scriptsize $3$};
			\draw (3*45: 2.5) node {\scriptsize $4$};
			\draw (4*45: 2.5) node {\scriptsize $5$};
			\draw (5*45: 2.5) node {\scriptsize $6$};
			\draw (6*45: 2.5) node {\scriptsize $7$};
			\draw (7*45: 2.5) node {\scriptsize $8$};
		\end{scope}
		
		\begin{scope}[shift={(-1, 0)}]
			\draw[line width=1.5pt] (0*45: 2) -- (1*45:2);
			\draw[line width=1.5pt] (1*45: 2) -- (2*45:2);
			\draw[line width=1.5pt] (2*45: 2) -- (3*45:2);
			\draw[line width=1.5pt] (3*45: 2) -- (4*45:2);
			\draw[line width=1.5pt] (4*45: 2) -- (5*45:2);
			\draw[line width=1.5pt] (5*45: 2) -- (6*45:2);
			\draw[line width=1.5pt] (6*45: 2) -- (7*45:2);
			\draw[line width=1.5pt] (7*45: 2) -- (8*45:2);
			
			\draw[fill=black!100] (0*45: 2) circle (0.1);
			\draw[fill=black!100] (1*45: 2) circle (0.1);			
			\draw[fill=black!100] (2*45: 2) circle (0.1);
			\draw[fill=black!100] (3*45: 2) circle (0.1);			
			\draw[fill=black!100] (4*45: 2) circle (0.1);
			\draw[fill=black!100] (5*45: 2) circle (0.1);			
			\draw[fill=black!100] (6*45: 2) circle (0.1);
			\draw[fill=black!100] (7*45: 2) circle (0.1);						
			
			\draw (0*45: 2.5) node {\scriptsize $9$};
			\draw (1*45: 2.5) node {\scriptsize $10$};
			\draw (2*45: 2.5) node {\scriptsize $11$};
			\draw (3*45: 2.5) node {\scriptsize $12$};
			\draw (4*45: 2.5) node {\scriptsize $13$};
			\draw (5*45: 2.5) node {\scriptsize $14$};
			\draw (6*45: 2.5) node {\scriptsize $15$};
			\draw (7*45: 2.5) node {\scriptsize $16$};
		\end{scope}
		
		\draw (+3, 0) node {\large $\ldots$};

		\begin{scope}[shift={(+7, 0)}]
			\draw[line width=1.5pt] (0*45: 2) -- (1*45:2);
			\draw[line width=1.5pt] (1*45: 2) -- (2*45:2);
			\draw[line width=1.5pt] (2*45: 2) -- (3*45:2);
			\draw[line width=1.5pt] (3*45: 2) -- (4*45:2);
			\draw[line width=1.5pt] (4*45: 2) -- (5*45:2);
			\draw[line width=1.5pt] (5*45: 2) -- (6*45:2);
			\draw[line width=1.5pt] (6*45: 2) -- (7*45:2);
			\draw[line width=1.5pt] (7*45: 2) -- (8*45:2);
			
			\draw[fill=black!100] (0*45: 2) circle (0.1);
			\draw[fill=black!100] (1*45: 2) circle (0.1);			
			\draw[fill=black!100] (2*45: 2) circle (0.1);
			\draw[fill=black!100] (3*45: 2) circle (0.1);			
			\draw[fill=black!100] (4*45: 2) circle (0.1);
			\draw[fill=black!100] (5*45: 2) circle (0.1);			
			\draw[fill=black!100] (6*45: 2) circle (0.1);
			\draw[fill=black!100] (7*45: 2) circle (0.1);						
			
			\draw (0*45: 2.55) node {\scriptsize $n-7$};
			\draw (1*45: 2.55) node {\scriptsize $n-6$};
			\draw (2*45: 2.55) node {\scriptsize $n-5$};
			\draw (3*45: 2.55) node {\scriptsize $n-4$};
			\draw (4*45: 2.55) node {\scriptsize $n-3$};
			\draw (5*45: 2.55) node {\scriptsize $n-2$};
			\draw (6*45: 2.55) node {\scriptsize $n-1$};
			\draw (7*45: 2.55) node {\scriptsize $n$};
		\end{scope}
		

		\draw (-7+6, 5-10) node {\large observed graph $\Ob_n$ with arrival times at the end of the stream};

 		\begin{scope}[shift={(0, -10)}]

			\draw (2, -3) node {\small arrival times $t$};

		\begin{scope}[shift={(-7, 0)}]
			\draw (-1, 4) node {\small vertex ids, e.g.\ $\sigma(4)$};
			\draw[line width=1pt, ->, >=stealth] (-2, 3.7) to [out=-100, in=120]  (3*45:2.7);
			
			\draw[line width=1pt, ->, >=stealth] (2+5.8, -3)  to [out=-180, in=-120]  (7.2*45:1.4);

			\draw[line width=1.5pt] (0*45: 2) -- (1*45:2);
			\draw[line width=1.5pt] (1*45: 2) -- (2*45:2);
			\draw[line width=1.5pt] (2*45: 2) -- (3*45:2);
			\draw[line width=1.5pt] (3*45: 2) -- (4*45:2);
			\draw[line width=1.5pt] (4*45: 2) -- (5*45:2);
			\draw[line width=1.5pt] (5*45: 2) -- (6*45:2);
			\draw[line width=1.5pt] (6*45: 2) -- (7*45:2);
			\draw[line width=1.5pt] (7*45: 2) -- (8*45:2);
			
			\draw[fill=black!100] (0*45: 2) circle (0.1);
			\draw[fill=black!100] (1*45: 2) circle (0.1);			
			\draw[fill=black!100] (2*45: 2) circle (0.1);
			\draw[fill=black!100] (3*45: 2) circle (0.1);			
			\draw[fill=black!100] (4*45: 2) circle (0.1);
			\draw[fill=black!100] (5*45: 2) circle (0.1);			
			\draw[fill=black!100] (6*45: 2) circle (0.1);
			\draw[fill=black!100] (7*45: 2) circle (0.1);			
			
			\draw (0*45: 2.5) node {\scriptsize $906$};
			\draw (1*45: 2.5) node {\scriptsize $127$};
			\draw (2*45: 2.5) node {\scriptsize $914$};
			\draw (3*45: 2.5) node {\scriptsize $633$};
			\draw (4*45: 2.5) node {\scriptsize $98$};
			\draw (5*45: 2.5) node {\scriptsize $279$};
			\draw (6*45: 2.5) node {\scriptsize $547$};
			\draw (7*45: 2.5) node {\scriptsize $958$};
			
			\draw (0.5*45: 1.5) node {\scriptsize $758$};
			\draw (1.5*45: 1.5) node {\scriptsize $744$};
			\draw (2.5*45: 1.5) node {\scriptsize $393$};
			\draw (3.5*45: 1.5) node {\scriptsize $656$};
			\draw (4.5*45: 1.5) node {\scriptsize $172$};
			\draw (5.5*45: 1.5) node {\scriptsize $707$};
			\draw (6.5*45: 1.5) node {\scriptsize $32$};
			\draw (7.5*45: 1.5) node {\scriptsize $277$};
		\end{scope}
		
		\begin{scope}[shift={(-1, 0)}]		
		
			\draw[line width=1pt, ->, >=stealth] (2+5.8-6, -3)  to [out=-180, in=-120]  (6.4*45:1.6);		
			
			\draw[line width=1.5pt] (0*45: 2) -- (1*45:2);
			\draw[line width=1.5pt] (1*45: 2) -- (2*45:2);
			\draw[line width=1.5pt] (2*45: 2) -- (3*45:2);
			\draw[line width=1.5pt] (3*45: 2) -- (4*45:2);
			\draw[line width=1.5pt] (4*45: 2) -- (5*45:2);
			\draw[line width=1.5pt] (5*45: 2) -- (6*45:2);
			\draw[line width=1.5pt] (6*45: 2) -- (7*45:2);
			\draw[line width=1.5pt] (7*45: 2) -- (8*45:2);
			
			\draw[fill=black!100] (0*45: 2) circle (0.1);
			\draw[fill=black!100] (1*45: 2) circle (0.1);			
			\draw[fill=black!100] (2*45: 2) circle (0.1);
			\draw[fill=black!100] (3*45: 2) circle (0.1);			
			\draw[fill=black!100] (4*45: 2) circle (0.1);
			\draw[fill=black!100] (5*45: 2) circle (0.1);			
			\draw[fill=black!100] (6*45: 2) circle (0.1);
			\draw[fill=black!100] (7*45: 2) circle (0.1);						
			
			\draw (0*45: 2.5) node {\scriptsize $965$};
			\draw (1*45: 2.5) node {\scriptsize $158$};
			\draw (2*45: 2.5) node {\scriptsize $971$};
			\draw (3*45: 2.5) node {\scriptsize $958$};
			\draw (4*45: 2.5) node {\scriptsize $486$};
			\draw (5*45: 2.5) node {\scriptsize $801$};
			\draw (6*45: 2.5) node {\scriptsize $142$};
			\draw (7*45: 2.5) node {\scriptsize $422$};
			
			\draw (0.5*45: 1.5) node {\scriptsize $47$};
			\draw (1.5*45: 1.5) node {\scriptsize $98$};
			\draw (2.5*45: 1.5) node {\scriptsize $824$};
			\draw (3.5*45: 1.5) node {\scriptsize $695$};
			\draw (4.5*45: 1.5) node {\scriptsize $318$};
			\draw (5.5*45: 1.5) node {\scriptsize $951$};
			\draw (6.5*45: 1.5) node {\scriptsize $35$};
			\draw (7.5*45: 1.5) node {\scriptsize $439$};
			
		\end{scope}
		
		\draw (+3, 0) node {\large $\ldots$};

		\begin{scope}[shift={(+7, 0)}]
		
			\draw[line width=1pt, ->, >=stealth] (2-5.5, -3)  to [out=0, in=-50]  (6.4*45:1.6);	
					
			\draw[line width=1.5pt] (0*45: 2) -- (1*45:2);
			\draw[line width=1.5pt] (1*45: 2) -- (2*45:2);
			\draw[line width=1.5pt] (2*45: 2) -- (3*45:2);
			\draw[line width=1.5pt] (3*45: 2) -- (4*45:2);
			\draw[line width=1.5pt] (4*45: 2) -- (5*45:2);
			\draw[line width=1.5pt] (5*45: 2) -- (6*45:2);
			\draw[line width=1.5pt] (6*45: 2) -- (7*45:2);
			\draw[line width=1.5pt] (7*45: 2) -- (8*45:2);
			
			\draw[fill=black!100] (0*45: 2) circle (0.1);
			\draw[fill=black!100] (1*45: 2) circle (0.1);			
			\draw[fill=black!100] (2*45: 2) circle (0.1);
			\draw[fill=black!100] (3*45: 2) circle (0.1);			
			\draw[fill=black!100] (4*45: 2) circle (0.1);
			\draw[fill=black!100] (5*45: 2) circle (0.1);			
			\draw[fill=black!100] (6*45: 2) circle (0.1);
			\draw[fill=black!100] (7*45: 2) circle (0.1);						
			
			\draw (0*45: 2.5) node {\scriptsize $916$};
			\draw (1*45: 2.5) node {\scriptsize $793$};
			\draw (2*45: 2.5) node {\scriptsize $960$};
			\draw (3*45: 2.5) node {\scriptsize $656$};
			\draw (4*45: 2.5) node {\scriptsize $36$};
			\draw (5*45: 2.5) node {\scriptsize $850$};
			\draw (6*45: 2.5) node {\scriptsize $934$};
			\draw (7*45: 2.5) node {\scriptsize $679$};
			
			\draw (0.5*45: 1.5) node {\scriptsize $382$};
			\draw (1.5*45: 1.5) node {\scriptsize $766$};
			\draw (2.5*45: 1.5) node {\scriptsize $796$};
			\draw (3.5*45: 1.5) node {\scriptsize $187$};
			\draw (4.5*45: 1.5) node {\scriptsize $490$};
			\draw (5.5*45: 1.5) node {\scriptsize $446$};
			\draw (6.5*45: 1.5) node {\scriptsize $647$};
			\draw (7.5*45: 1.5) node {\scriptsize $710$};
			
		\end{scope}
		\end{scope}
		
		\end{tikzpicture}
		\caption{Illustration of the underlying graph $G$ (top) and the observed graph (bottom). In this example $\pi^{-1}(1)=758, \pi^{-1}(2)=744, \pi^{-1}(3)=393$ etc.}	\label{fig:g-o}
	\end{center}
	
\end{figure}

\begin{definition}[Board state $\Bb_t$]\label{def:bt}
We define the ``board'' $\Bb_t$ to be the sequence of all edge arrivals
up to time $t$, i.e.\ \[
\Bb_t = (\bb_s)_{s = 1}^t \text{.}
\]
\end{definition}
We will occasionally use $\Bb_{s:t}$ to refer to $(\bb_i)_{i = s}^t$. This then
defines the \emph{observed graph} $\Ob_t = (V, \lbrace \bb_s \rbrace_{s =
1}^t)$ at time $t$. Note that $\Ob_t$ is isomorphic to $(V, \lbrace
e_v\rbrace_{\bpi(v) \le t})$, with $\bsigma$ giving an isomorphism.

Each subset $\Sb$ of $\lbrace \bb_s \rbrace_{s = 1}^t$ (or equivalently, subgraph
of $\Ob_t$) can be associated with a $t$-bit binary string $x^\Sb$ given by \[
x_s = \begin{cases}
1 & \mbox{if $\bb_s \in \Sb$}\\
0 & \mbox{otherwise.}
\end{cases}
\]
In cases where it is unambiguous, we will use $\Sb$ to refer to $x^\Sb$
directly.  Conversely, for each $x \in \bool^t$, we will write $\Ob^{x}_t$ for
the subgraph of $\Ob_t$ such that $x^{\Ob^x_t} = x$.  

Note that as $\Ob_t$ is isomorphic to a subgraph of a union of length-$\ell$
cycles, each of its components is either a length-$\ell$ cycle of a path of
length $<\ell$, and the same holds for all subgraphs $S$.

\paragraph{Player input.} Choose $\Xb\sim \Uc(\bool^T)$. At each time step $t \in T$, the $t\nth$ player
receives three pieces of input: {\bf (1)} the bit $\Xb_t\in \bool$, {\bf (2)} a
message $\Mb_{t-1}\in \bool^c$ from player $t-1$ (we let $\Mb_{0}=\mathbf{0}^c$ for convenience), and
{\bf (3)} the contents of the board $\Bb_t$. Then, if $t < n$, player $t$ sends
$\Mb_t$ to player $t+1$.

\paragraph{Objective of the game.} Player $n$, after timestep $n$, must output $v
\in V$ and $\Xb \cdot \Cb_v$, where $\Cb_v$ is
the component of $\Ob_n$ containing $v$ (which, as $G$ is a union of
length-$\ell$ cycles, will necessarily be a length-$\ell$ cycle).

As the input distribution of the problem is fixed, we will by Yao's principle
assume that the players are deterministic from now on. 
\subsection{The Underlying and the Observable Graph}\label{sec:oaugraph}
In this section we introduce terminology and some necessary lemmas for
understanding the relationship between the observed and the underlying graph.
\begin{definition}
\label{dfn:Pit}
For each $t \in \lbrack n \rbrack$, $\bPi_t$ is the set of possibilities for
$\bpi$ that are compatible with $\Bb_t$ (see Definition~\ref{def:bt}), i.e.\ \[
\bPi_t = \left\lbrace \pi \in S_n :  \exists \sigma \in S_n,
(\sigma(e_{\pi(s)}))_{s = 1}^t = \Bb_t\right\rbrace
\]
\end{definition}
\begin{lemma}
\label{lem:pidist}
For every $t\in [n]$, conditioned on $\Bb_t$, one has $\bpi \sim \mathcal{U}(\bPi_t)$.
\end{lemma}
\begin{proof}
Unconditionally, $(\bpi,\bsigma)$ are uniformly distributed on $S_n \times S_n$,
so conditioned on $\Bb_t$, they are distributed uniformly on \[
\lbrace (\pi, \sigma) \in S_n \times S_n : (\sigma(e_{\pi(s)}))_{s = 1}^t =
\Bb_t \rbrace
\]
and so it will suffice to prove that for every $\pi \in \bPi_t$, the number of
$\sigma \in S_n$ such that \[
(\sigma(e_{\pi(s)}))_{s = 1}^t = \Bb_t
\]
is the same. For any pair $\pi_1, \pi_2\in \bPi_t$ 
We will give an injection from $\sigma'_1$ such that \[
(\sigma'_1(e_{\pi_1(s)}))_{s = 1}^t = \Bb_t
\]
to $\sigma'_2$ such that \[
(\sigma'_2(e_{\pi_2(s)}))_{s = 1}^t = \Bb_t
\]
implying that there at least as many such $\sigma'_2$ as there are such
$\sigma'_1$, and so by symmetry there are the same number of each. 

Since $\sigma_1, \sigma_2\in \Pi_t$ by assumption, there exist $\sigma_1, \sigma_2$ such that 
\[
(\sigma_1(e_{\pi_1(s)}))_{s = 1}^t = \Bb_t =
(\sigma_2(e_{\pi_2(s)}))_{s = 1}^t.
\]
Fix a choice of such $\sigma_1$ and $\sigma_2$. The injection is defined by setting $\sigma_2' = \sigma_1'\sigma_1^{-1}\sigma_2$. For any fixing of
$\sigma_1, \sigma_2$, this is an injective function of $\sigma_1'$, as
$\sigma_1^{-1}\sigma_2$ is a permutation, and
\begin{align*}
(\sigma'_2(e_{\pi_2(s)}))_{s = 1}^t &=
(\sigma_1'\sigma_1^{-1}\sigma_2(e_{\pi_2(s)}))_{s = 1}^t\\
&= (\sigma_1'\sigma_1^{-1}\sigma_1(e_{\pi_1(s)}))_{s = 1}^t\\
&= (\sigma_1'(e_{\pi_1(s)}))_{s = 1}^t\\
&= \Bb_s
\end{align*}
completing the proof.
\end{proof}

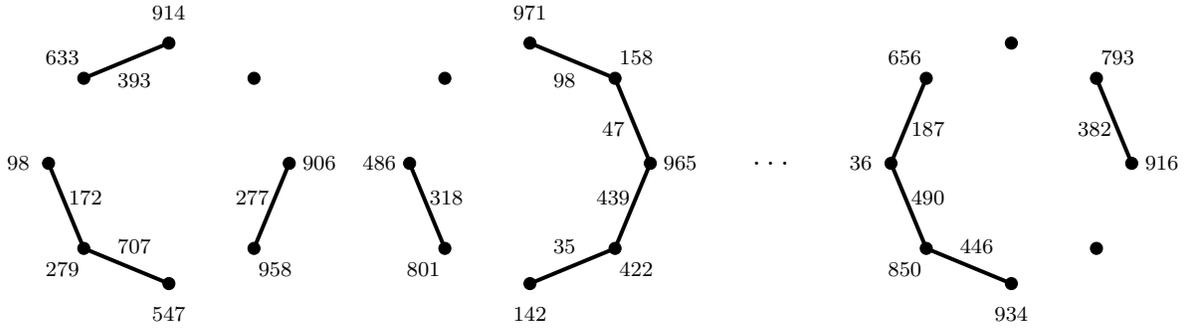
\begin{figure}[H]
	\begin{center}
		\begin{tikzpicture}[scale=0.8]
		

		\draw (-7+6, 3.5-10) node {\large observed graph $\Ob_t$ with arrival times at time $t=500$};

 		\begin{scope}[shift={(0, -10)}]

		\begin{scope}[shift={(-7, 0)}]

			\draw[line width=1.5pt] (2*45: 2) -- (3*45:2);
			\draw[line width=1.5pt] (4*45: 2) -- (5*45:2);
			\draw[line width=1.5pt] (5*45: 2) -- (6*45:2);
			\draw[line width=1.5pt] (7*45: 2) -- (8*45:2);
			
			\draw[fill=black!100] (0*45: 2) circle (0.1);
			\draw[fill=black!100] (1*45: 2) circle (0.1);			
			\draw[fill=black!100] (2*45: 2) circle (0.1);
			\draw[fill=black!100] (3*45: 2) circle (0.1);			
			\draw[fill=black!100] (4*45: 2) circle (0.1);
			\draw[fill=black!100] (5*45: 2) circle (0.1);			
			\draw[fill=black!100] (6*45: 2) circle (0.1);
			\draw[fill=black!100] (7*45: 2) circle (0.1);			
			
			\draw (0*45: 2.5) node {\scriptsize $906$};
			\draw (2*45: 2.5) node {\scriptsize $914$};
			\draw (3*45: 2.5) node {\scriptsize $633$};
			\draw (4*45: 2.5) node {\scriptsize $98$};
			\draw (5*45: 2.5) node {\scriptsize $279$};
			\draw (6*45: 2.5) node {\scriptsize $547$};
			\draw (7*45: 2.5) node {\scriptsize $958$};
			

			\draw (2.5*45: 1.5) node {\scriptsize $393$};
			\draw (4.5*45: 1.5) node {\scriptsize $172$};
			\draw (5.5*45: 1.5) node {\scriptsize $707$};
			\draw (7.5*45: 1.5) node {\scriptsize $277$};
		\end{scope}
		
		\begin{scope}[shift={(-1, 0)}]

			\draw[line width=1.5pt] (0*45: 2) -- (1*45:2);
			\draw[line width=1.5pt] (1*45: 2) -- (2*45:2);
			\draw[line width=1.5pt] (4*45: 2) -- (5*45:2);
			\draw[line width=1.5pt] (6*45: 2) -- (7*45:2);
			\draw[line width=1.5pt] (7*45: 2) -- (8*45:2);
			
			\draw[fill=black!100] (0*45: 2) circle (0.1);
			\draw[fill=black!100] (1*45: 2) circle (0.1);			
			\draw[fill=black!100] (2*45: 2) circle (0.1);
			\draw[fill=black!100] (3*45: 2) circle (0.1);			
			\draw[fill=black!100] (4*45: 2) circle (0.1);
			\draw[fill=black!100] (5*45: 2) circle (0.1);			
			\draw[fill=black!100] (6*45: 2) circle (0.1);
			\draw[fill=black!100] (7*45: 2) circle (0.1);						
			
			\draw (0*45: 2.5) node {\scriptsize $965$};
			\draw (1*45: 2.5) node {\scriptsize $158$};
			\draw (2*45: 2.5) node {\scriptsize $971$};
			\draw (4*45: 2.5) node {\scriptsize $486$};
			\draw (5*45: 2.5) node {\scriptsize $801$};
			\draw (6*45: 2.5) node {\scriptsize $142$};
			\draw (7*45: 2.5) node {\scriptsize $422$};
			
			\draw (0.5*45: 1.5) node {\scriptsize $47$};
			\draw (1.5*45: 1.5) node {\scriptsize $98$};
			\draw (4.5*45: 1.5) node {\scriptsize $318$};
			\draw (6.5*45: 1.5) node {\scriptsize $35$};
			\draw (7.5*45: 1.5) node {\scriptsize $439$};
			
		\end{scope}
		
		\draw (+3, 0) node {\large $\ldots$};

		\begin{scope}[shift={(+7, 0)}]

			\draw[line width=1.5pt] (0*45: 2) -- (1*45:2);
			\draw[line width=1.5pt] (3*45: 2) -- (4*45:2);
			\draw[line width=1.5pt] (4*45: 2) -- (5*45:2);
			\draw[line width=1.5pt] (5*45: 2) -- (6*45:2);
			
			\draw[fill=black!100] (0*45: 2) circle (0.1);
			\draw[fill=black!100] (1*45: 2) circle (0.1);			
			\draw[fill=black!100] (2*45: 2) circle (0.1);
			\draw[fill=black!100] (3*45: 2) circle (0.1);			
			\draw[fill=black!100] (4*45: 2) circle (0.1);
			\draw[fill=black!100] (5*45: 2) circle (0.1);			
			\draw[fill=black!100] (6*45: 2) circle (0.1);
			\draw[fill=black!100] (7*45: 2) circle (0.1);						
			
			\draw (0*45: 2.5) node {\scriptsize $916$};
			\draw (1*45: 2.5) node {\scriptsize $793$};
			\draw (3*45: 2.5) node {\scriptsize $656$};
			\draw (4*45: 2.5) node {\scriptsize $36$};
			\draw (5*45: 2.5) node {\scriptsize $850$};
			\draw (6*45: 2.5) node {\scriptsize $934$};
			
			\draw (0.5*45: 1.5) node {\scriptsize $382$};
			\draw (3.5*45: 1.5) node {\scriptsize $187$};
			\draw (4.5*45: 1.5) node {\scriptsize $490$};
			\draw (5.5*45: 1.5) node {\scriptsize $446$};
			
		\end{scope}
		\end{scope}
		
		\end{tikzpicture}
		\caption{Illustration of the state of the board (i.e., the observed graph annotated with arrival times) at an intermediate point $t=500$ in the stream.}	\label{fig:o-intermediate}
	\end{center}
	
\end{figure}

\subsection{Basic properties of collections of components}\label{sec:basic-defs}
We write $\zpp$ to denote the set of multisets of non-negative integers.
We will refer to each such multiset as a \emph{collection type}. For such a
multiset $\alpha$ we will write $\alpha\lbrack i \rbrack$ for the number of
times $i$ appears in $\alpha$.

\begin{definition}\label{def:sim}
  For $t\in T$, $z\in \{0, 1\}^t$ and $\alpha\in \zpp$ we write
  $z\bsim_t \alpha$ if $\Ob_t^z$ is a union of components in $\Ob_t$,
  and for each $i \in \lbrack \ell \rbrack$, the number of $i$-edge
  components in $\Ob_t^z$ is $\alpha\lbrack i \rbrack$.
\end{definition}

\begin{definition}[Weight and size of $\alpha\in \zpp$]\label{def:weight-and-size}
For $\alpha\in \zpp$ we let $|\alpha|_*:=\sum_{i\in \alpha} i$ denote the {\em weight} of $\alpha$ and let $|\alpha|$ denote the number of elements in $\alpha$, which we refer to as the {\em size} of $\alpha$.
\end{definition}

\begin{definition}[Extensions and Merges]
For any $\alpha \in \zpp$ we define an extension or merge as follows:
\begin{description}
\item[{[Extension]}] Increment one of the elements in $\alpha$ by $1$;
\item[{[Merge]}] Replace two elements $a, b\in \alpha$ with $a+b+1$.
\end{description}
Similarly, for any $t \in T$, $z \in \bool^t$, we say $\Ob_t^z$ experiences an
extension or merge at time $t+1$ if:
\begin{description}
\item[{[Extension]}] $\bb_{t+1}$ is incident to exactly one component in $\Ob_t^z$.
\item[{[Merge]}] $\bb_{t+1}$ connects two components in $\Ob_t^z$.
\end{description}
\end{definition}
Note that $\Ob_t^z$ experiences an extension or merge at time $t+1$ iff $\Ob_t
\bsim_t \alpha$, $\Ob_{t+1} \bsim_{t+1} \beta$ for $\beta$ obtained by an
extension or merge on $\alpha$, respectively. 

\begin{definition}[Down set of $\beta\in \zpp$]\label{def:downset}
For $\alpha, \beta\in \zpp$ we write $\alpha\in \beta-1$ if $\beta$ can be
obtained from $\alpha$ by either an {\em extension} or a {\em merge}  followed
by possibly adding an arbitrary number of $1$'s to $\alpha$.

For $\beta\in \zpp$ and $\alpha\in \beta-1$ we write $|\beta-\alpha|$ to denote the number of ones that need to be added to $\alpha$ after a merge or extension to obtain $\beta$.
\end{definition}

\begin{definition}[Growth event $\grow(z, s, t)$]
For every $t\in T$, every $s \in \lbrack t \rbrack$ and every $z\in \bool^t$,
let $\grow(z, s, t)=1$ if $\bb_{s+1}$ extends one or merges two components of
$\Ob_s^{z_{\le s}}$ and only single edge components of $\Ob_t^z$ arrive between
$s+1$ and $t$, and let $\grow(z, s, t)=0$ otherwise. 
\end{definition}

\begin{lemma}
\label{lem:isogrowth}
For every $t\in T$, $B_t$ in the support of $\Bb_t$, every $r, r'\in
\bool^t$ such that $r\bsim_t \alpha$ and $r' \bsim_t \alpha$ for some $\alpha\in \zpp$ when
$\Bb_t = B_t$, then for all $\beta \in \zpp$ \[
\Pb[\Bb_{t+1}]{ r \cdot 1 \bsim_{t+1} \beta \middle| \Bb_t = B_t} =
\Pb[\Bb_{t+1}]{ r' \cdot 1 \bsim_{t+1} \beta | \Bb_t = B_t} \text{.}
\]
\end{lemma}
\begin{proof}
Recall that the sequence of edges $\Bb_t$ is given by \[
(\bsigma(e_{\bpi(s)}))_{s = 1}^t
\]
where $\bsigma$ and $\bpi$ are uniformly chosen permutations from $\lbrack
n\rbrack$ to $\lbrack n\rbrack$, while the edge that arrives at time $t+1$ is
given by \[
\bsigma(e_{\bpi(t+1)})
\] 
Recall also that $\Bb_t$ does \emph{not} include the identity of $\bsigma, \bpi$
themselves. Now, as $r \bsim_t \alpha$ and $r' \bsim_t \alpha$, we can
construct a graph automorphism $\phi : V \rightarrow V$ of $\Ob_t$, depending
only on $\Bb_t$, such that $\phi$ swaps $\Ob_t^r$ with $\Ob_t^{r'}$ and is the
identity everywhere else. Then, as $\Bb_t$ includes the arrival time of the
edges, we can construct a permutation $\psi : T \rightarrow T$, also depending
only on $\Bb_t$, such that for all $t \in T$ \[
\phi(\bb_{\psi(t)}) = \bb_t
\]
and $\psi$ is the identity everywhere except the support of $r$ and $r'$.

Recalling that $\bb_s = \bsigma(e_{\bpi(s)})$, this means that, as the
distribution of $\bsigma, \bpi$ conditional on $\Bb_t$ is uniform on \[
\lbrace (\sigma, \pi) \in S_n \times S_n : (\sigma(e_{\pi(s)}))_{s = 1}^t =
(\bb_s)_{s=1}^t \rbrace
\]
we have, for each $\sigma, \pi \in S_n \times S_n$, \[
\Pb[\bsigma, \bpi]{ (\bsigma, \bpi) = (\sigma, \pi) } =
\Pb[\bsigma, \bpi]{ (\bsigma, \bpi) = (\phi\sigma, \psi\pi) }\text{.}
\]
Now, as $\phi$ is a graph automorphism swapping $\Ob_{t}^z$ and $\Ob_{t}^{z'}$,
and $\psi(t+1) = t+1$,  $\bb_{t+1}$ will extend or merge $\Ob_{t}^z$ iff
$\phi(\bb_{\psi(t+1)})$ causes the same extension or merge in $\Ob_{t}^{z'}$,
and so the above equation gives us the result.
\end{proof}

This allows us to define the probability of growth events using only the
\emph{collection type} of a collection of paths.

\begin{definition}
For $s \in [n]$, a pair of collection types $\alpha, \beta\in
\zpp$ such that $\alpha\lbrack 1\rbrack = \beta\lbrack 1\rbrack$ and an element
$B_s$ of the support of $\Bb_s$ we write
\[
p_s(\alpha, \beta, B_s) = \Pb[\Bb_{s+1}]{ z\cdot 1 \bsim_{s+1} \beta | \Bb_s =
B_s}
\] for any $z \in \lbrace 0, 1\rbrace^s$ such that $z \bsim_s \alpha$. Furthermore,
if $\beta \lbrack 1\rbrack > \alpha \lbrack 1 \rbrack$, we write \[
p_s(\alpha, \beta, B_s) = p(\alpha, \beta', B_s)
\]
where $\beta'$ is $\beta$ with $\beta\lbrack 1 \rbrack - \alpha\lbrack 1
\rbrack$ copies of 1 removed, whereas if $\beta \lbrack 1\rbrack < \alpha
\lbrack 1 \rbrack$, we set $p(\alpha, \beta, B_s) = 0$.
\end{definition}

\subsection{Extension and Merge Probabilities} \label{sec:comb}
In this section, we will prove the following lemma, bounding the probability
that a merge or extension event will occur at a given time.
\begin{lemma}\label{lm:ext-prob}
For some absolute constant $C > 0$, for each $t \le n - C 2^{2\ell} n^{5/6}
\log n$, there is an event $\mathcal{E}_t$ over $\Bb_t$ such that
$\Pb[\Bb_t]{ \mathcal{E}_t } \ge 1 - 1/n^{\ell + 1}$, and for any
$B_t \in \mathcal{E}_t$ and for every $\alpha, \beta$ such that $|\alpha|_* \le
\ell - 3$ (recall Definition~\ref{def:weight-and-size}) \[ 
p_t(\alpha, \beta, B_t) \le \begin{cases}
\frac{O(\alpha\lbrack a\rbrack)}{n} &\mbox{if $\alpha\to \beta$ is an extension
of a path of size $a$}\\
\frac{O(\alpha \lbrack a\rbrack \cdot \alpha \lbrack b\rbrack)}{(n - t)^2}
&\mbox{if $\alpha\to \beta$ is a merge of paths of size $a$ and $b$.}
\end{cases}
\]
\end{lemma}
We will prove four cases of the lemma, depending on $t$ and whether we are
considering an extension or a merge. 

We start by proving the lemma when $t \le n/8$ and $\beta$ is reached from
$\alpha$ by an extension. 
\begin{claim}
For each $t \le n/8$, every $B_t$ in the support of $\Bb_t$, and for
every $\alpha, \beta$ such that $\alpha\to \beta$ is an extension of a path of
length $a$,
 \[
p_{t}(\alpha, \beta, B_t) \le \frac{3\alpha\lbrack
a\rbrack}{n} \text{.}
\]
\end{claim}
\begin{proof}  
Let $O_t$ be the graph corresponding to the edges in the sequence $B_t$. Let $H$
be any subgraph of $B_t$ with component sizes corresponding to $\alpha$.  It
will suffice to bound the probability that a length-$a$ component of $H$ is
extended at time $t + 1$ when $\Bb_t = B_t$.

There are $2\alpha\lbrack a \rbrack$ distinct end vertices of components of
length $a$ in $H$. Call this set $S$. One of these paths is extended iff
$\bsigma(\bpi(s+1)) \in S$ or $\bsigma(\next(\bpi(s+1))) \in S$. For any
$\sigma$, suppose $\bsigma = \sigma$. Now for each element of $S$, there is
exactly one possibility for $\bpi(s+1)$ that will cause one of
$\bsigma(\bpi(s+1)) \in S$ or $\bsigma(\next(\bpi(s+1))) \in S$ to hold, and
conditioned on $\Bb_t = B_t$, $\bsigma = \sigma$, $\bpi(t+1)$ is uniformly
distributed on a set of size $n - t \ge 7n/8$ (as fixing $\Bb_t$ and $\bsigma$
fixes $(\bpi(s))_{s=1}^t$). So the result follows by taking a union bound.
\end{proof}

Next, we prove it for merges when $t \le n/8$. To do this, we will first
introduce a new concept, the swap graph $\Gbc_t^z$.  

\noindent\paragraph{The swap graph $\Gbc_t^z$.} For any $t \in T$, $z \in
\bool^t$, the vertex set of this (undirected) graph will be $\bPi_t$ as defined
in Definition~\ref{dfn:Pit}, the set of permutations in the support of $\bpi$
that are consistent with the observed board $\Bb_t$.

We now define the edge set of $\Gbc_t^z$. If $\Ob^z_t$ is \emph{not} a single
component of size at most $\ell - 2$ in $\Ob_t$, $\Gbc_t^z$ is the complete
graph for convenience.

Now suppose that  $\Ob^z_t$ is a single component of size at most $\ell-2$. We now define edges incident on a permutation $\pi \in \bPi_t$ in $\Gbc_t^z$. To define these edges, first note that the set of vertices
$$
\pi(\lbrace s \in \lbrack
t \rbrack : z_s = 1 \rbrace)
$$ induces a path $P$ in the underlying graph $G$. Let $k$ denote the length of this path. Write 
\begin{equation}\label{eq:p}
P=(u_i)_{i=1}^k, u_i\in V,
\end{equation} where for every $i\in [k-1]$ one has $u_{i+1}=\next(u_i)$. In particular, we have, as per~\eqref{eq:e-v}, that $e_{u_i} = (u_i, u_{i+1})$ for each $i \in \brac*{k - 1}$.  
For every path 
\begin{equation}\label{eq:p-prime}
P'=(u'_i)_{i=1}^k, u'_i\in V
\end{equation}
in $G$ of the same length first 
define 
\begin{equation}\label{eq:psi}
\psi : V \rightarrow V
\end{equation} to swap $u_i$ and $u'_i$ for each $i \in
\brac*{k-1}$ while being the identity everywhere else, and then add an edge $(\pi, \pi')$, where 
\begin{equation}\label{eq:pi-prime}
\pi'=\psi\pi,
\end{equation}
to $\Gbc_t^z$ if $\pi' \in \bPi_t$.  We have that $\pi'$ is consistent
with $\Bb_t$ on every edge except, possibly, the edges immediately
before and after $P'$ or P.  In particular, $\pi' \in \bPi_t$ if none
of those edges arrive in $\Bb_t$.

\begin{lemma}
\label{lm:swapgraph}
If $t \le n/8$, the minimum degree of $\Gbc_t^z$ is at least $n/8$.
\end{lemma}
\begin{proof}
The result is trivial if $\Ob^z_t$ is not a single component of size at most
$\ell - 2$ in $\Ob_t$, so we will consider only the case where it is. 

We will to show that, for any $\pi \in \bPi_t$, there are at least $n/12 - 1$
ways of choosing $\pi'$ by the process defined above such that $\pi'$ will
still be in $\bPi_t$. Let $P, P', \psi, \pi'=\psi\pi$ be as in~\eqref{eq:p},~\eqref{eq:p-prime},~\eqref{eq:psi} and~\eqref{eq:pi-prime} above. Note that $\pi'$ will be in $\bPi_t$ if there is some choice of
permutation $\sigma' : V \rightarrow V$ such that $(\sigma'(e_{\pi'(s)}))_{s =
1}^t = \Bb_t$.  As there is a permutation $\sigma$ such that
$(\sigma(e_{\pi(s)}))_{s = 1}^t = \Bb_t$, choosing 
$$
\sigma'=\psi \sigma
$$ will guarantee that $\sigma'(e_{\pi'(s)}) = \bb_s$ for each $s \in \brac*{t}$. So
$(\sigma'(e_{\pi'(s)}))_{s = 1}^t = \Bb_t$ will hold provided
\begin{equation}\label{eq:2983yt8348g}
\sigma'(e_{\pi'(s)}) = \sigma(e_{\pi'(s)}) = \sigma(e_{\pi(s)})
\end{equation}
for all $s$ such that $z_s = 0$. The condition in~\eqref{eq:2983yt8348g} will
be satisfied when there is no $s \in \brac*{t}$ such that $z_s = 0$ and
$e_{\pi(s)}$ is incident on either $P$ or $P'$. This will happen iff there is
no $s \in \brac*{t}$ such that either $\next(\pi(s)) = u'_1$ or $\pi(s) =
u'_k$. 

Therefore, the number of valid choices for the new path, and therefore the
degree of $\pi$ in $\Gbc_t^z$, is given by the number of pairs $w_0, w_k$
in $V$ such that {\bf (a}) $\pi(w_0), \pi(w_k) > t$, and {\bf (b)} there is
a path $(w_i)_{i=0}^{k}$ such that $w_{i} = \next(w_{i-1})$ for each $i \in
\brac*{k}$.  We lower bound the number of such pairs now.

We will start by lower bounding the number of disjoint pairs $w_0, w_k$ in
$V$ that satisfy {\bf (b)}. We may assume without loss of generality
that $k \le \ell/2$, as there is a one-to-one correspondence between paths
$(w_0)_{i=1}^k$ and paths $(w_i')_{i=0}^{\ell - k}$ such that $w'_0 =
w_k$ and $w'_{\ell - k} = w_0$. 

We will consider two cases to lower bound this number of disjoint pairs:
\begin{description}
\item{$k \le \ell/4$:} We may divide each cycle in the underlying graph $G$ into
$\floor{\frac{\ell}{2k}}$ disjoint blocks of $2k$ consecutive vertices.
In each such block we may fit $k$ disjoint pairs satisfying {\bf (b)}.
Therefore each cycle contains \[
\floor{\frac{\ell}{2k}} \cdot k \ge \left(\frac{\ell}{2k}- 1\right)\cdot k \ge \ell/4
\]
 disjoint pairs satisfying {\bf (b)}, and so $G$ contains at least $n/4$ of
 them.  
 \item{$k > \ell/4$} For any cycle in $G$, let $v$ be a vertex in the cycle.
 Then for each $i < k$, $(\next^i(v), \next^{i+k}(v))$ is a pair satisfying
 {\bf (b)}, and these pairs are all disjoint. So there are at least $\min(k,
 \ell - k) > \ell/4$ such disjoint pairs in the cycle, and therefore at least
 $n/4$ in $G$.
\end{description}
Therefore, as $t \le n/8$, there are at most $n/8$  pairs $(w_0, w_k)$ that do
not satisfy {\bf (a)}. Combining with the bounds above, we get that there are
at least $n/8$ pairs $(w_0, w_k)$ that satisfy both {\bf (a)} and {\bf
(b)}.  This completes the proof.
\end{proof}

We are now ready to prove the result for merges when $t \le n/8$.
\begin{claim}
For each $t \le n/8$, every $B_t$ in the support of $\Bb_t$, and for
every $\alpha, \beta$ such that $\alpha\to \beta$ is a merge of paths of
length $a,b$,
 \[
 p_{t}(\alpha, \beta, B_s) \le \frac{14\alpha\brac*{a}
 \cdot \alpha\brac*{b}}{n^2} \text{.}
 \]
 \end{claim}
\begin{proof}
Fix $\Bb_t = B_t$ and let $O_t$ be the corresponding value of $\Ob_t$. For any
pair of components $P_1, P_2$ of lengths $a$, $b$ in $O_t$, they will merge at
time $t+1$ iff two criteria are satisfied:
\begin{enumerate}
\item The two paths $\Qb_1 = \bsigma^{-1}(P_1)$, $\Qb_2 = \bsigma^{-1}(P_2)$ in the
underlying graph $G$ need only one edge to be connected.
\item That edge arrives at time $t+1$.
\end{enumerate}
If the first criterion is satisfied, the second will be with probability
$\frac{1}{n-t} \le \frac{8}{7n}$, so we will seek to bound the first. Let $\Pi$ be the value of $\bPi_t$ (recall Definition~\ref{dfn:Pit}) given $\Bb_t = B_t$, and let $\Pi^*$ be
the set of $\pi \in \Pi$ that would lead to the first criterion being satisfied
if $\bpi = \pi$.  We will prove that $|\Pi^*| \le \frac{8}{n}|\Pi|$, which by
Lemma~\ref{lem:pidist} suffices to prove that the probability of $\bpi$
satisfying the first criterion is at most $8/n$.

Let $z \in \bool^t$ be given by $z_s = 1$ if an edge in $P_1$ arrived at time
$s$, and $z_s = 0$ otherwise. Let $\Gc^z_t$ be the swap graph $\Gbc_t^z$ when
$\Bb_t = B_t$. Each vertex of this graph corresponds to a choice of $\pi$. 
Note
that, given $\Bb_t = B_t$, the values of $\bsigma^{-1}$ on non-isolated vertices
of $\Ob_t$ and the values of $\bpi$ on $\brac*{t}$ each uniquely determine the
other, so in particular such a vertex determines a choice of $\bsigma^{-1}$. Thus, the choice of 
a vertex in $\Gbc_t^z$ uniquely determines a choice $Q_1$, $Q_2$ of the length $a$, $b$ paths $\Qb_1$,
$\Qb_2$. Moving along an edge of $\Gc^z_t$ corresponds to choosing a different
$Q_1$ while leaving $Q_2$ the same. 

We observe that for each $\pi \in \Pi$,
there are at most $2$ neighbors of $\pi$ in $\Pi^*$, as when fixing $Q_2$ there
is at most two choices of $Q_1$ that lead to them being separated by exactly
one edge. By Lemma~\ref{lm:swapgraph}, the minimum degree of $\Gc^z_t$ is at
least $n/8$. Therefore, each vertex in $\Pi^*$ is a neighbor of at least $n/8$ 
vertices of $\Pi$, each of which neighbors at most $2$ vertices of $\Pi^*$, and
so\[
|\Pi| \ge \frac{n}{16} |\Pi^*|.
\]
Therefore the probability that $\bpi$ satisfied the first criterion given
$\Bb_t = B_t$ is at most $n/16$. We have that the probability of any
given pair of components of length $a, b$ merging is at most\[
\frac{16}{n} \cdot \frac{7}{8n} = \frac{14}{n^2}
\]
and so the result follows.
\end{proof}

For the remaining cases we will use a new random variable $\Ub_s$, that gives the
final board state, up to re-orderings of the edge arrivals from time $s+1$ to
$n$. As conditioning on $\Ub_s$ means that Lemma~\ref{lem:isogrowth} no longer
holds, we will introduce new notation for the \emph{average} growth
probability.
\begin{definition}[Average growth probabilities]\label{def:ut}
For each $s \in \lbrack n \rbrack$, the random variable $\Ub_s$ is given by $\Bb_s$
and $(\bpi(i))_{i=1}^s$. For each $\alpha, \beta \in \zpp$ such that
$\alpha\lbrack 1 \rbrack = \beta\lbrack 1 \rbrack$, \[
\upsilon_{s}(\alpha, \beta, \Ub_s) = \frac{1}{|z \in \lbrace 0, 1\rbrace^s : z
\bsim_s \alpha|}\sum_{\substack{z \in \lbrace 0, 1\rbrace^s\\ z \bsim_s \alpha}}
\Pb[\Bb_{s+1}] { z \cdot 1 \bsim_{s+1} \beta | \Ub_s }
\]
while if $\beta \lbrack 1\rbrack > \alpha \lbrack 1 \rbrack$, we write \[
\upsilon_s(\alpha, \beta, B_s) = \upsilon_s(\alpha, \beta', B_s)
\]
where $\beta'$ is $\beta$ with $\beta\lbrack 1 \rbrack - \alpha\lbrack 1
\rbrack$ copies of 1 removed, whereas if $\beta \lbrack 1\rbrack < \alpha
\lbrack 1 \rbrack$, we set $\upsilon(\alpha, \beta, B_s) = 0$.

\end{definition}

For extensions and merges we will define high probability 
events $\mathcal{E}^e_s$, $\mathcal{E}^m_s$ respectively, over $\Ub_s$, such that conditioned on these events the quantity $\upsilon(\alpha, \beta, U_t^e)$ satisfies the natural analog of bound in Lemma~\ref{lm:ext-prob}.  Lemma~\ref{lm:u2b} below shows that such events imply Lemma~\ref{lm:ext-prob}. The rest of this section is devoted to defining such events.

\begin{lemma}\label{lm:u2b}
For any $s \in \lbrack n\rbrack$, let $\mathcal{E}_s^e$, $\mathcal{E}_s^m$ be
events such that for all $(U_s^e,U_s^m) \in \mathcal{E}_s^e \times
\mathcal{E}_s^m$, and any $\alpha, \beta$ such that $\abs{\alpha}_* \le \ell
-3$,
\begin{align*}
\upsilon(\alpha, \beta, U_s^e) &\le \frac{O(\alpha\lbrack a\rbrack)}{n}
&\mbox{if $\alpha\to \beta$ is an extension of a path of size $a$}\\
\upsilon(\alpha, \beta, U_s^m) &\le  \frac{O(\alpha \lbrack a\rbrack \cdot
\alpha \lbrack b\rbrack)}{(n - s)^2} &\mbox{if $\alpha\to \beta$ is a merge of
paths of size $a$ and $b$}
\end{align*}
and each event occurs with probability at least $1 - 1/(2n^{\ell + 3})$ over
$\Ub_s$. Then there is an event $\mathcal{E}_s$ such that for any $B_s \in
\mathcal{E}_s$ and for every $\alpha, \beta$ such that $\abs{\alpha}_* \le \ell
- 3$ \[
p_{s}(\alpha, \beta, B_s) \le \begin{cases}
\frac{O(\alpha\lbrack a\rbrack)}{n} &\mbox{if $\alpha\to \beta$ is an extension
of a path of size $a$}\\
\frac{O(\alpha \lbrack a\rbrack \cdot \alpha \lbrack b\rbrack)}{(n -
s)^2} &\mbox{if $\alpha\to \beta$ is a merge of paths of size $a$ and $b$.}
\end{cases}
\]
and $\Pb[\Bb_s]{ \mathcal{E}_s } \ge 1 - 1/n^{\ell + 1}$.
\end{lemma}
\begin{proof}
Define $\mathcal{E}_s$ to be the event that $\Pb[\Ub_s]{ \mathcal{E}_s^e
\cap \mathcal{E}_s^m \middle| \Bb_s } \ge 1 - 1/n^2$. Then 
\begin{align*}
1 - \Pb[\Ub_s]{\mathcal{E}_s^e \cap \mathcal{E}_s^m} &=
\E[B_s]{1 - \Pb[\Ub_s]{\mathcal{E}_s^e \cap \mathcal{E}_s^m |
\Bb_s}} \\
&\ge \frac{1}{n^2}( 1 - \Pb[\Bb_s]{\Ec_s})
\end{align*}
and so \[
\Pb[\Bb_s]{ \mathcal{E}_s } \ge 1 - 1/n^{\ell + 1}.
\]
For any $B_s \in \mathcal{E}_s$, $z \bsim_s \alpha$, by applying Lemma~\ref{lem:isogrowth},  we get
\begin{align*}
\Pb[\Bb_{s+1}]{z \cdot 1 \bsim_{s+1} \beta| \Bb_s = B_s} &= \frac{1}{|z
\in \lbrace 0, 1\rbrace^s : z \bsim_s \alpha|}\sum_{\substack{z \in \lbrace 0,
1\rbrace^s\\ z \bsim_s \alpha}} \Pb[\Bb_{s+1}] { z \cdot 1 \bsim_{s+1}
\beta | \Bb_s = B_s }\\
&=  \E[\Ub_s]{\frac{1}{|z \in \lbrace 0, 1\rbrace^s : z
\bsim_s \alpha|}\sum_{\substack{z \in \lbrace 0, 1\rbrace^s\\ z \bsim_s \alpha}}
\Pb[\Bb_{s+1}]{ z \cdot 1 \bsim_{s+1} \beta | \Ub_s } | \Bb_s = B_s}\\
&= \E[\Ub_s]{\upsilon(\alpha, \beta, \Ub_s) | \Bb_s = B_s}\\
&\le \E[\Ub_s]{\upsilon(\alpha, \beta, \Ub_s) | \Bb_s = B_s, \Ec^e_s
\cap \Ec^m_s} + \frac{1}{n^2}\\ &\le \begin{cases}
\frac{O(\alpha\lbrack a\rbrack)}{n} + \frac{1}{n^2} &\mbox{if $\alpha\to \beta$
is an extension of a path of size $a$}\\ \frac{O(\alpha \lbrack
a\rbrack \cdot \alpha \lbrack b\rbrack)}{(n - s)^2} + \frac{1}{n^2} &\mbox{if
$\alpha\to \beta$ is a merge of paths of size $a$ and $b$}
\end{cases}
\end{align*}
This gives the result.

\end{proof}

$\Ub_s$ tells us how edges visible in $\Bb_s$ correspond to edges in the underlying
graph. We will use the following lemma to obtain tight bounds on how often
$\Ub_s$ contains certain ``patterns'' of edges. Specifically, we will be interested in the following types of patterns:
\begin{description}
\item[(1)] a path $P$ with $a$ edges in the underlying graph $G$ such that all of the edges of $P$ have arrived by time $t$, but neither of the adjacent edges have (see Fig.~\ref{fig:1}, (a))
\item[(2)] a path $P$ with $a$ edges in the underlying graph $G$ such that all of the edges of $P$ have arrived by time $t$, neither of the adjacent edges have arrived, and the {\bf next} vertex in the canonical order  (defined by the $\next$ function) does not have any incident edges at time $t$ (see Fig.~\ref{fig:1}, (b))
\item[(3)] a path $P$ with $a$ edges in the underlying graph $G$ such that all of the edges of $P$ have arrived by time $t$, neither of the adjacent edges have arrived, and the {\bf previous} vertex in the canonical order (defined by the $\next$ function) does not have any incident edges at time $t$ (see Fig.~\ref{fig:1}, (c))
\item[(4)] a path $P$ with $a+b$ edges in the underlying graph $G$ such that the first $a$ and last $b$ of the edges of $P$ have arrived by time $t$, the $(a+1)\nth$ edge has not arrived, and neither of the edges adjacent to $P$ have arrived (see Fig.~\ref{fig:2}). 
\end{description} 

\begin{figure}[H]
	\begin{center}
		\begin{tikzpicture}[scale=0.9]

 		\begin{scope}[shift={(0, -10)}]

		\begin{scope}[shift={(-7, 0)}]
			
			\draw[line width=1.5pt] (0*45: 2) -- (1*45:2);
			\draw[line width=1.5pt] (1*45: 2) -- (2*45:2);
			\draw[line width=1.5pt] (2*45: 2) -- (3*45:2);
			\draw[line width=1.5pt, dashed] (3*45: 2) -- (4*45:2);
			\draw[line width=0.5pt, color=black!50] (4*45: 2) -- (5*45:2);
			\draw[line width=0.5pt, color=black!50] (5*45: 2) -- (6*45:2);
			\draw[line width=1.5pt, dashed] (6*45: 2) -- (7*45:2);
			\draw[line width=1.5pt] (7*45: 2) -- (8*45:2);
			
			\draw[fill=black!100] (0*45: 2) circle (0.1);
			\draw[fill=black!100] (1*45: 2) circle (0.1);			
			\draw[fill=black!100] (2*45: 2) circle (0.1);
			\draw[fill=black!100] (3*45: 2) circle (0.1);			
			\draw[fill=black!50] (4*45: 2) circle (0.1);
			\draw[fill=black!50] (5*45: 2) circle (0.1);			
			\draw[fill=black!50] (6*45: 2) circle (0.1);
			\draw[fill=black!100] (7*45: 2) circle (0.1);			
			
			\draw (0*45: 2.5) node {\scriptsize $906$};
			\draw (1*45: 2.5) node {\scriptsize $127$};
			\draw (2*45: 2.5) node {\scriptsize $914$};
			\draw (3*45: 2.5) node {\scriptsize $633$};
			\draw (7*45: 2.5) node {\scriptsize $958$};
			
			\draw (0.5*45: 1.5) node {\scriptsize $758$};
			\draw (1.5*45: 1.5) node {\scriptsize $744$};
			\draw (2.5*45: 1.5) node {\scriptsize $393$};
			\draw (7.5*45: 1.5) node {\scriptsize $277$};
			
			\draw (0, -3) node {(a)};			
		\end{scope}
		
		\begin{scope}[shift={(-1, 0)}]

			\draw[line width=1.5pt] (0*45: 2) -- (1*45:2);
			\draw[line width=1.5pt] (1*45: 2) -- (2*45:2);
			\draw[line width=1.5pt] (2*45: 2) -- (3*45:2);
			\draw[line width=1.5pt] (3*45: 2) -- (4*45:2);
			\draw[line width=1.5pt, dashed] (4*45: 2) -- (5*45:2);
			\draw[line width=1.5pt, dashed] (5*45: 2) -- (6*45:2);
			\draw[line width=0.5pt] (6*45: 2) -- (7*45:2);
			\draw[line width=1.5pt, dashed] (7*45: 2) -- (8*45:2);
			
			\draw[fill=black!100] (0*45: 2) circle (0.1);
			\draw[fill=black!100] (1*45: 2) circle (0.1);			
			\draw[fill=black!100] (2*45: 2) circle (0.1);
			\draw[fill=black!100] (3*45: 2) circle (0.1);			
			\draw[fill=black!100] (4*45: 2) circle (0.1);
			\draw[fill=black!50] (5*45: 2) circle (0.1);			
			\draw[fill=black!50] (6*45: 2) circle (0.1);
			\draw[fill=black!50] (7*45: 2) circle (0.1);						
			
			\draw (0*45: 2.5) node {\scriptsize $965$};
			\draw (1*45: 2.5) node {\scriptsize $158$};
			\draw (2*45: 2.5) node {\scriptsize $971$};
			\draw (3*45: 2.5) node {\scriptsize $958$};
			\draw (4*45: 2.5) node {\scriptsize $486$};
		
			\draw (0.5*45: 1.5) node {\scriptsize $47$};
			\draw (1.5*45: 1.5) node {\scriptsize $98$};
			\draw (2.5*45: 1.5) node {\scriptsize $824$};
			\draw (3.5*45: 1.5) node {\scriptsize $695$};
			
			\draw (0, -3) node {(b)};			
			
		\end{scope}

		\begin{scope}[shift={(+5, 0)}]
							
			\draw[line width=1.5pt, dashed] (0*45: 2) -- (1*45:2);
			\draw[line width=1.5pt] (1*45: 2) -- (2*45:2);
			\draw[line width=1.5pt] (2*45: 2) -- (3*45:2);
			\draw[line width=1.5pt] (3*45: 2) -- (4*45:2);
			\draw[line width=0.5pt] (4*45: 2) -- (5*45:2);
			\draw[line width=0.5pt] (5*45: 2) -- (6*45:2);
			\draw[line width=0.5pt] (6*45: 2) -- (7*45:2);
			\draw[line width=1.5pt, dashed] (7*45: 2) -- (8*45:2);
			
			\draw[fill=black!50] (0*45: 2) circle (0.1);
			\draw[fill=black!100] (1*45: 2) circle (0.1);			
			\draw[fill=black!100] (2*45: 2) circle (0.1);
			\draw[fill=black!100] (3*45: 2) circle (0.1);			
			\draw[fill=black!100] (4*45: 2) circle (0.1);
			\draw[fill=black!100] (5*45: 2) circle (0.1);			
			\draw[fill=black!50] (6*45: 2) circle (0.1);
			\draw[fill=black!50] (7*45: 2) circle (0.1);						
			
			\draw (1*45: 2.5) node {\scriptsize $793$};
			\draw (2*45: 2.5) node {\scriptsize $960$};
			\draw (3*45: 2.5) node {\scriptsize $656$};
			\draw (4*45: 2.5) node {\scriptsize $36$};
			\draw (5*45: 2.5) node {\scriptsize $850$};
			
			\draw (1.5*45: 1.5) node {\scriptsize $766$};
			\draw (2.5*45: 1.5) node {\scriptsize $796$};
			\draw (3.5*45: 1.5) node {\scriptsize $187$};
			\draw (4.5*45: 1.5) node {\scriptsize $490$};
			
			\draw (0, -3) node {(c)};
			
		\end{scope}
		\end{scope}
		
		\end{tikzpicture}
		\caption{Illustration of patterns {\bf (1), (2)} and {\bf (3)}. Solid edges must be present at time $t$, dashed edges must not be present at time $t$,}	\label{fig:1}
	\end{center}
	
\end{figure}
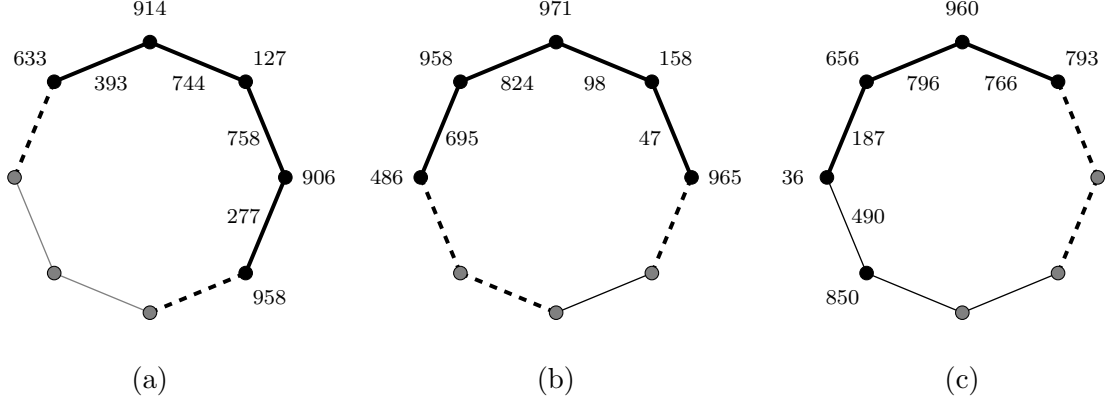

\begin{figure}[H]
	\begin{center}
		\begin{tikzpicture}[scale=0.9]

 		\begin{scope}[shift={(0, -10)}]
		\begin{scope}[shift={(-1, 0)}]

			\draw[line width=1.5pt] (0*45: 2) -- (1*45:2);
			\draw[line width=1.5pt, dashed] (1*45: 2) -- (2*45:2);
			\draw[line width=1.5pt] (2*45: 2) -- (3*45:2);
			\draw[line width=1.5pt] (3*45: 2) -- (4*45:2);
			\draw[line width=1.5pt, dashed] (4*45: 2) -- (5*45:2);
			\draw[line width=0.5pt] (5*45: 2) -- (6*45:2);
			\draw[line width=0.5pt] (6*45: 2) -- (7*45:2);
			\draw[line width=1.5pt, dashed] (7*45: 2) -- (8*45:2);
			
			\draw[fill=black!100] (0*45: 2) circle (0.1);
			\draw[fill=black!100] (1*45: 2) circle (0.1);			
			\draw[fill=black!100] (2*45: 2) circle (0.1);
			\draw[fill=black!100] (3*45: 2) circle (0.1);			
			\draw[fill=black!100] (4*45: 2) circle (0.1);
			\draw[fill=black!50] (5*45: 2) circle (0.1);			
			\draw[fill=black!50] (6*45: 2) circle (0.1);
			\draw[fill=black!50] (7*45: 2) circle (0.1);						
			
			\draw (0*45: 2.5) node {\scriptsize $755$};
			\draw (1*45: 2.5) node {\scriptsize $277$};
			\draw (2*45: 2.5) node {\scriptsize $680$};
			\draw (3*45: 2.5) node {\scriptsize $656$};
			\draw (4*45: 2.5) node {\scriptsize $163$};
			
			\draw (0.5*45: 1.5) node {\scriptsize $341$};
			\draw (2.5*45: 1.5) node {\scriptsize $224$};
			\draw (3.5*45: 1.5) node {\scriptsize $752$};
			
		\end{scope}
		\end{scope}
		
		\end{tikzpicture}
		\caption{Illustration of pattern {\bf (4)}. Solid edges must be present at time $t$, dashed edges must not be present at time $t$,}	\label{fig:2}
	\end{center}
	
\end{figure}
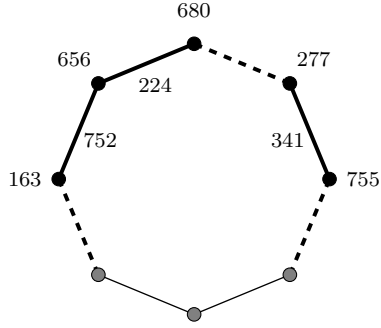

Bounds on the number of occurrences of the first three types of patterns are then used to bound the average extension probability in Claim~\ref{cl:avg-ext} below. The first and the forth are then used to bound the average merge probability in Claim~\ref{cl:avg-merge} below.

The following lemma bounds the number of ``patterns'' of the above types simultaneously.  In order to cover all patterns above, we prove a more general lemma that applies to all possible patterns as opposed to just the ones above. For an integer $k\geq 1$ we encode such patterns by a binary string $y$ of length $k$, where for $i\in [k]$ we have $y_i=1$ if the $i\nth$ edge in a consecutive segment of $k$ edges in the underlying graph must be present, and $y_i=0$ if the $i\nth$ edge must be absent.

\begin{lemma}
\label{lem:patternfreq}
For any $t \in T$, let $y \in \lbrace 0, 1\rbrace^k$ for some $k \in \lbrack
\ell \rbrack$. Suppose $|\overline{y}| \in \set*{2,3}$. For each $v \in V$, let
$(v_i)_{i=1}^{k+1}$ be the path in the underlying graph $G$ such that $v_1 = v$
and $v_{i+1} = \next(v_i)$ for each $i \in \brac*{k}$. Let $\Yb_v = 1$ if for
all $i \in \brac*{k}$, $\1bb(\bpi^{-1}(v_i) \le t) = y_i$. Let $\Yb = \sum_{v\in
V} \Yb_v$. 

Then, for every constant $C > 0$ there is a $D > 0$
depending only on $C$ such that for every \[
t \in \brac*{n/8, n - D 2^{2\ell} n^{5/6} \log n}
\]
one has \[
\Yb = \Theta\paren*{n \cdot \paren*{\frac{t}{n}}^{|y|} \cdot \paren*{\frac{n -
t}{n} }^{|\overline{y}|}}
\]
with probability at least $1 - n^{-C \ell}$ over $\Ub_t$.
\end{lemma}

\begin{proof}
Let $\bbeta : V \rightarrow \bool$ be given by $\bbeta(v) = \1bb(\bpi^{-1}(v)
\le t)$. Then for all $v \in V$, $\E[\Ub_t]{\bbeta(v)} = t/n$, but the
values $\bbeta(V)$ are not independent, as $\bbeta$ is fixed to have exactly $t$
values equal to 1, and $n-t$ equal to 0.

Now, let $\bbeta'$ be defined as follows:
\begin{enumerate}
\item Draw $\tb' \sim \bi(n, t/n)$.
\item ``Correct'' $\bbeta$ to have exactly $\tb'$ values equal to 1, by either
flipping $\tb' - t$ randomly chosen zeroes to 1, or flipping $t - \tb'$
randomly chosen ones to 0, as necessary.
\end{enumerate}
Now, $\bbeta'$ will have $\tb'$ ones, and which values of $\bbeta'$ are 1 will
be chosen uniformly at random, so we have both $\E[\Ub_t,
\tb']{\bbeta'(v)} = t/n$ for all $v \in V$ and the values $\bbeta'(V)$
being independent.

Now, for each $v \in V$, let $(v_i)_{i=1}^{k+1}$ be as defined in the lemma
statement, and let $\Yb'_v$ be 1 if for all $i \in \brac*{k}$, $\1bb(\bbeta'(v))
= y_i$. Let $\Yb' = \sum_{v \in V}\Yb'_v$. Now, since for every $v\in V$ the value of $\bbeta(v)$ (resp. $\bbeta'(v)$) affects at most $k\leq \ell$ values of $\Yb(u), u\in V$ (resp. $\Yb'(u), u\in V$), we have
 \[
|\Yb' - \Yb| \le \left|\{v\in V: \bbeta(v) \not= \bbeta'(v\}\right\|\leq \ell|\tb' - t|
\]
and so by applying the Chernoff bounds to $\tb'$, with probability $1 -
n^{-C\ell}/2$ over $\tb'$,
\begin{align*}
|\Yb' - \Yb| &= O(\sqrt{\ell \cdot \E{\tb'}\log n})\\
&= O(\sqrt{\ell t \log n})\\
&= O(\sqrt{\ell n  \log n})
\end{align*}
and so we may proceed by bounding the distribution of $\Yb'$.

The events $\set*{\Yb'_v : v \in V}$ are not independent. However for each $v$,
$\Yb'_v$ is independent of the set of events that do not depend on $\bbeta'(v_i)$
for any $i \in \brac*{k}$. This is all but $2k - 2 < 2\ell - 1$ of the other
random variables in $\set*{\Yb'_v : v \in V}$. So we may construct a $2\ell$
set partition of $V$, denoted by $\brac*{V_i : i \in \brac*{2\ell}}$, such that, for each $i
\in \brac*{2\ell}$, the events $\set*{\Yb'_v : v \in V_i}$ are independent. 

Now, for each $v \in V$, $\Yb'_v = 1$ with probability
$\paren*{\frac{t}{n}}^{|y|}\paren*{\frac{n - t}{n}}^{|\overline{y}|}$, and $0$
otherwise. Now, for each $i \in \brac*{2\ell}$, let $\Yb'^{(i)} = \sum_{v \in
V_i}\Yb'_v$. By the Chernoff bounds, with probability $1 - n^{-C\ell}/4\ell$ over
$\Ub_t$ and $\tb'$, \[ 
\abs{\Yb'^{(i)} - \E[\Ub_t, \tb']{\Yb'^{(i)}}} = O\paren*{\sqrt{\ell \cdot
\E{\Yb'^{(i)}} \log n \log \ell}} 
\]
and so by taking a union bound over $\set*{\Yb'^{(i)} : i \in \brac*{2\ell}}$, with
probability $1 - n^{-C\ell}/2$ over $\Ub_t$ and $\tb'$, \[
\abs{\Yb'- \E[\Ub_t, \tb']{\Yb'}} = O\paren*{\ell^{3/2}
\sqrt{\E{\Yb'} \log n \log \ell}}.
\]
Combining this with our bound on $\abs{\Yb' -
\Yb}$ above, with probability $1 - n^{-C\ell}$ over $\Ub_t$ (as $\Yb$ is independent of
$\tb'$), \[
\abs{\Yb - \E[\Ub_t, \tb']{\Yb'}} = O\paren*{\sqrt{\ell n  \log n} +
\ell^{3/2} \sqrt{\E{\Yb'} \log n \log \ell}}.
\]
Since \[
\E[U_t, \tb']{\Yb'}  = n\paren*{\frac{t}{n}}^{|y|}\paren*{\frac{n -
t}{n}}^{|\overline{y}|}
\]
we have 
\begin{align*}
\sqrt{\ell n \log n} +\ell^{3/2} \sqrt{\E{\Yb'} \log n \log \ell} &= \sqrt{\ell n \log n} + \ell^{3/2} \sqrt{n \log n \log \ell } \paren*{\frac{t}{n}}^{|y|/2}\paren*{\frac{n -
t}{n}}^{|\overline{y}/2|}\\
&\le 2\ell^{3/2}\sqrt{n\log n \log \ell}
\end{align*}
and so, using the fact that $t \ge n/8$,
\begin{equation}\label{eq:y-wg8932gthg}
\Yb = n\paren*{\frac{t}{n}}^{|y|}\paren*{\frac{n -t}{n}}^{|\overline{y}|}\cdot \paren*{1 + \gamma},\\
\end{equation}
where
$$
\gamma=O\paren*{\ell^{3/2}\sqrt{\log n \log \ell} \paren*{\frac{t}{n}}^{-|y|}\paren*{\frac{n -t}{n}}^{-|\overline{y}|}}/\sqrt{n}.
$$
We now note that
\begin{align*}
|\gamma|&=O\paren*{\ell^{3/2}\sqrt{\log n \log \ell} 2^{3\ell}\paren*{\frac{n}{n-t}}^3}/\sqrt{n}\text{~~~~~~~(since $\paren*{\frac{t}{n}}^{-|y|}\leq 2^{3\ell}$ since $t\geq n/8$)}\\
&\le O\paren*{\sqrt{\log n } 2^{5\ell} \paren*{\frac{n}{n-t}}^3}/\sqrt{n}\text{~~~~~~~~~~~~~~~~~~~~~~(since $\ell^{3/2}\leq 2^{2\ell}$ for $\ell\geq 2$)}\\
&\le O\paren*{\paren*{\frac{2^{2\ell}n^{5/6} \log n}{n-t}}^3}\\
&\leq O(1/D^3)\text{~~~~~~~~~~~~~~~~~~~~~~~~~~~~~~~~~~~~~~~~~~~~~~~~~~~~~~~~~~~~(since $t \geq n - D2^{2\ell} n^{5/6}\log n$)}\\
&\leq 1/2\\
\end{align*}
as long as $D$ is larger than an absolute constant. Thus, we get that $|\gamma|\leq 1/2$ and the lemma follows by~\eqref{eq:y-wg8932gthg}.
\end{proof} 
We now use this result to bound $\upsilon(\alpha, \beta, U_t^e)$ in the case
where $t \ge n/8$ and $\alpha \rightarrow \beta$ is an extension.
\begin{claim}\label{cl:avg-ext} For some absolute constant $C > 0$, for each $t \in \brac*{n/8,
n - C 2^{2\ell} n^{5/6} \log n }$, there is an event
$\mathcal{E}_t^e$ such that for all $U_t^e \in \mathcal{E}_t^e $, and any
$\alpha, \beta$ such that $\abs{\alpha}_* \le \ell -3$ ($|\alpha|_*$ is the weight of $\alpha$ as per Definition~\ref{def:weight-and-size}) and $\alpha \rightarrow
\beta$ is an extension of a length-$a$ path, one has \[
\upsilon(\alpha, \beta, U_t^e) \le \frac{O(\alpha\lbrack a\rbrack)}{n}.
\]
The event $\mathcal{E}_t^e$ occurs with probability at least $1 - 1/2n^{\ell + 3}$ over
$\Ub_t$.
\end{claim}
\begin{proof}
Given $\Ub_t$, we naturally associate with each $a$-edge component in $\Ob_t$ an $(a+2)$-edge path
$(v_i)_{i=1}^{a+3}$ that consists of the component present at time $t$ together with the two adjacent edges in the underlying graph $G$ that have not arrived yet (see Fig.~\ref{fig:extension-prob}).  Specifically, we have $v_{i+1} =
\next(v_i)$ for each $i \in \brac*{a+2}$, and $\bpi^{-1}(v_1) > t$,
$\bpi^{-1}(\set*{v_i : 1 < i \le a + 1}) \subseteq \brac*{t}$,
$\bpi^{-1}(v_{a+2}) > t$ (pattern of type {\bf (1)} above, see Fig.~\ref{fig:1}). Call the number of such paths $\Yb_1$. 

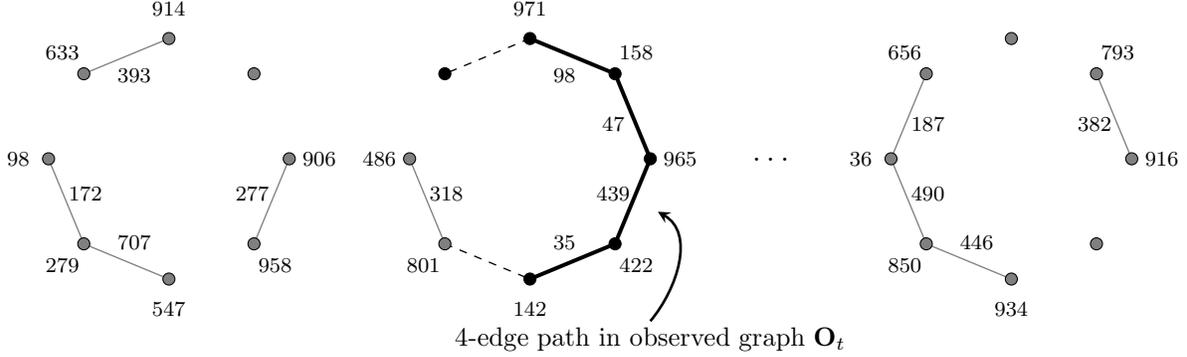
\begin{figure}[H]
	\begin{center}
		\begin{tikzpicture}[scale=0.8]
		

 		\begin{scope}[shift={(0, -10)}]

		\begin{scope}[shift={(-7, 0)}]

			\draw[line width=0.5pt, color=black!50] (2*45: 2) -- (3*45:2);
			\draw[line width=0.5pt, color=black!50] (4*45: 2) -- (5*45:2);
			\draw[line width=0.5pt, color=black!50] (5*45: 2) -- (6*45:2);
			\draw[line width=0.5pt, color=black!50] (7*45: 2) -- (8*45:2);
			
			\draw[fill=black!50] (0*45: 2) circle (0.1);
			\draw[fill=black!50] (1*45: 2) circle (0.1);			
			\draw[fill=black!50] (2*45: 2) circle (0.1);
			\draw[fill=black!50] (3*45: 2) circle (0.1);			
			\draw[fill=black!50] (4*45: 2) circle (0.1);
			\draw[fill=black!50] (5*45: 2) circle (0.1);			
			\draw[fill=black!50] (6*45: 2) circle (0.1);
			\draw[fill=black!50] (7*45: 2) circle (0.1);			
			
			\draw (0*45: 2.5) node {\scriptsize $906$};
			\draw (2*45: 2.5) node {\scriptsize $914$};
			\draw (3*45: 2.5) node {\scriptsize $633$};
			\draw (4*45: 2.5) node {\scriptsize $98$};
			\draw (5*45: 2.5) node {\scriptsize $279$};
			\draw (6*45: 2.5) node {\scriptsize $547$};
			\draw (7*45: 2.5) node {\scriptsize $958$};
			

			\draw (2.5*45: 1.5) node {\scriptsize $393$};
			\draw (4.5*45: 1.5) node {\scriptsize $172$};
			\draw (5.5*45: 1.5) node {\scriptsize $707$};
			\draw (7.5*45: 1.5) node {\scriptsize $277$};
		\end{scope}
		
		\begin{scope}[shift={(-1, 0)}]

			\draw[line width=1.5pt] (0*45: 2) -- (1*45:2);
			\draw[line width=1.5pt] (1*45: 2) -- (2*45:2);
			\draw[line width=0.5pt, dashed] (2*45: 2) -- (3*45:2);
			\draw[line width=0.5pt, color=black!50] (4*45: 2) -- (5*45:2);
			\draw[line width=0.5pt, dashed] (5*45: 2) -- (6*45:2);
			\draw[line width=1.5pt] (6*45: 2) -- (7*45:2);
			\draw[line width=1.5pt] (7*45: 2) -- (8*45:2);
			
			\draw[fill=black!100] (0*45: 2) circle (0.1);
			\draw[fill=black!100] (1*45: 2) circle (0.1);			
			\draw[fill=black!100] (2*45: 2) circle (0.1);
			\draw[fill=black!100] (3*45: 2) circle (0.1);			
			\draw[fill=black!50] (4*45: 2) circle (0.1);
			\draw[fill=black!50] (5*45: 2) circle (0.1);			
			\draw[fill=black!100] (6*45: 2) circle (0.1);
			\draw[fill=black!100] (7*45: 2) circle (0.1);						
			
			\draw (0*45: 2.5) node {\scriptsize $965$};
			\draw (1*45: 2.5) node {\scriptsize $158$};
			\draw (2*45: 2.5) node {\scriptsize $971$};
			\draw (4*45: 2.5) node {\scriptsize $486$};
			\draw (5*45: 2.5) node {\scriptsize $801$};
			\draw (6*45: 2.5) node {\scriptsize $142$};
			\draw (7*45: 2.5) node {\scriptsize $422$};
			
			\draw (0.5*45: 1.5) node {\scriptsize $47$};
			\draw (1.5*45: 1.5) node {\scriptsize $98$};
			\draw (4.5*45: 1.5) node {\scriptsize $318$};
			\draw (6.5*45: 1.5) node {\scriptsize $35$};
			\draw (7.5*45: 1.5) node {\scriptsize $439$};

			\draw (2, -3) node {\small $4$-edge path in observed graph $\Ob_t$};

			\draw[line width=1pt, ->, >=stealth] (2, -2.7)  to [out=50, in=-30]  (7.5*45:2.3);

		\end{scope}
		
		\draw (+3, 0) node {\large $\ldots$};

		\begin{scope}[shift={(+7, 0)}]

			\draw[line width=0.5pt, color=black!50] (0*45: 2) -- (1*45:2);
			\draw[line width=0.5pt, color=black!50] (3*45: 2) -- (4*45:2);
			\draw[line width=0.5pt, color=black!50] (4*45: 2) -- (5*45:2);
			\draw[line width=0.5pt, color=black!50] (5*45: 2) -- (6*45:2);
			
			\draw[fill=black!50] (0*45: 2) circle (0.1);
			\draw[fill=black!50] (1*45: 2) circle (0.1);			
			\draw[fill=black!50] (2*45: 2) circle (0.1);
			\draw[fill=black!50] (3*45: 2) circle (0.1);			
			\draw[fill=black!50] (4*45: 2) circle (0.1);
			\draw[fill=black!50] (5*45: 2) circle (0.1);			
			\draw[fill=black!50] (6*45: 2) circle (0.1);
			\draw[fill=black!50] (7*45: 2) circle (0.1);						
			
			\draw (0*45: 2.5) node {\scriptsize $916$};
			\draw (1*45: 2.5) node {\scriptsize $793$};
			\draw (3*45: 2.5) node {\scriptsize $656$};
			\draw (4*45: 2.5) node {\scriptsize $36$};
			\draw (5*45: 2.5) node {\scriptsize $850$};
			\draw (6*45: 2.5) node {\scriptsize $934$};
			
			\draw (0.5*45: 1.5) node {\scriptsize $382$};
			\draw (3.5*45: 1.5) node {\scriptsize $187$};
			\draw (4.5*45: 1.5) node {\scriptsize $490$};
			\draw (5.5*45: 1.5) node {\scriptsize $446$};
			
		\end{scope}
		\end{scope}
		
		\end{tikzpicture}
		\caption{Path with $a=4$ edges in the observed graph and the corresponding $(a+2)$-edge path in the underlying graph (incident edges that have not arrived yet are shown as dashed)}	\label{fig:extension-prob}
	\end{center}
	
\end{figure}

Such a component will be extended at time $t+1$ iff the edge that arrives at
this time is incident to it and not to any other component, that is iff either:
\begin{description}
\item[(1)] $\bpi(t+1) = v_1$ and $\bpi^{-1}(\next^{-1}(v_1)) > t$ (i.e., the other endpoint of the arriving edge does not have any incident edges in the observed graph $\Ob_t$)
\end{description}
or
\begin{description}
\item[(2)] $\bpi(t+1) = v_{a+2}$ and $\bpi^{-1}(\next(v_{a+2})) > t$ (i.e., the other endpoint of the arriving edge does not have any incident edges in the observed graph $\Ob_t$).
\end{description}
For each of these, the second criterion is determined by $\Ub_t$ (recall Definition~\ref{def:ut}), and the first
will be satisfied with probability $\frac{1}{n-t}$ over $\bpi(t+1)$. So
conditioned on $\Ub_t$, the average probability that an $a$-edge component in
$\Ob_t$ is extended is $\frac{1}{n - t} \cdot \frac{\Yb_2 + \Yb_3}{\Yb_1}$,
where $\Yb_2$ is the number of such paths where $\bpi^{-1}(\next^{-1}(v_1)) >
t$ and $\Yb_3$ is the number where $\bpi^{-1}(\next(v_{a+2})) > t$ (patterns of type {\bf (2)} and {\bf (3)} above, see Fig.~\ref{fig:1}).

By applying Lemma~\ref{lem:patternfreq} with $C = 3$ (as $\ell$ must be at
least $3$ and so $n$ must be as well, and thus $n^{C\ell} > 6n^{\ell + 3}$) and
taking a union bound,
\begin{align*}
\Yb_2 + \Yb_3  &= \Theta\paren*{n\paren*{\frac{t}{n}}^{a}
\paren*{\frac{n-t}{n}}^3}\\
\Yb_1 &= \Theta\paren*{n\paren*{\frac{t}{n}}^{a}
\paren*{\frac{n-t}{n}}^2}
\end{align*}
with probability $1 - 1/2n^{\ell + 3}$ over $\Ub_t$, provided \[
n/8 \le t \le n - D 2^{2\ell} n^{5/6} \log n
\]
for a sufficiently large constant $C$. So let $\mathcal{E}_t^e$ be the event
that this holds, then $\Pb[\Ub_t]{\mathcal{E}_t^e} \ge 1 - 1/2n^{\ell + 3}$ and
for all $U_t \in \mathcal{E}_s^e$,
\begin{align*}
\upsilon(\alpha, \beta, U_t) &= \frac{\alpha\lbrack a\rbrack}{n - t} \cdot
\Theta\left(\frac{n\left(\frac{t}{n}\right)^{a}
\left(\frac{n-t}{n}\right)^3}{n\left(\frac{t}{n}\right)^{a}
\left(\frac{n-t}{n}\right)^2}\right)\\
&= \frac{O(\alpha\lbrack a \rbrack)}{n}
\end{align*}
completing the proof.
\end{proof}
Finally, we consider the case where $t \ge n/8$ and $\alpha \rightarrow \beta$ is a
merge.
\begin{claim}\label{cl:avg-merge} For some absolute constant $C > 0$, for each $t \in \brac*{n/8, n
- C 2^{2\ell} n^{5/6} \log n}$, there is an event $\mathcal{E}_t^e$ such that
for all $U_t^e \in \mathcal{E}_t^e $, and any $\alpha, \beta$ such that
$\abs{\alpha}_* \le \ell -3$ and $\alpha \rightarrow \beta$ is an merge of a
length-$a$ and length-$b$ path, \[
\upsilon(\alpha, \beta, U_t^e) \le \frac{O(\alpha\lbrack a\rbrack \cdot
\alpha\lbrack b \rbrack)}{(n-t)^2}
\]
$\mathcal{E}_t^e$ occurs with probability at least $1 - 1/2n^{\ell + 3}$ over
$\Ub_t$.
\end{claim}
\begin{proof}
Given $\Ub_t$, we naturally associate with each pair of an $a$-edge and a $b$-edge component in $\Ob_t$
an $(a+2)$-edge path $(u_i)_{i=1}^{a+3}$ and a $(b+2)$-edge path $(v_i)_{i=1}^{b+3}$ in the underlying
graph $G$ that consist of the edges of the two components together with the not yet arrived adjacent edges (see Fig.~\ref{fig:merge-prob}). Specifically, $u_{i+1} = \next(u_i)$ for each $i \in \brac*{a+2}$,
$v_{i+1} = \next(v_i)$ for each $i \in \brac*{b+2}$, and $ \bpi^{-1}(u_1) > t,
\bpi^{-1}(v_1) > t$, \[
\bpi^{-1}(\set*{u_i : 1 < i \le a +1} \cup \set*{v_i : 1 < i \le b + 1})
\] is contained in $\brac*{t}$, $\bpi^{-1}(u_{a+2}) > t, \bpi^{-1}(v_{b+2})
> t$. Both correspond to patterns of type {\bf (1)} above, see Fig.~\ref{fig:1}. Call the number of such paths $\Yb_1$, $\Yb_2$, respectively.

\begin{figure}[H]
	\begin{center}
		\begin{tikzpicture}[scale=0.8]
		

 		\begin{scope}[shift={(0, -10)}]

		\begin{scope}[shift={(-7, 0)}]

			\draw[line width=0.5pt, color=black!50] (2*45: 2) -- (3*45:2);
			\draw[line width=0.5pt, dashed] (3*45: 2) -- (4*45:2);
			\draw[line width=1.5pt, color=black!100] (4*45: 2) -- (5*45:2);
			\draw[line width=1.5pt, color=black!100] (5*45: 2) -- (6*45:2);
			\draw[line width=0.5pt, dashed] (6*45: 2) -- (7*45:2);
			\draw[line width=0.5pt, color=black!50] (7*45: 2) -- (8*45:2);
			
			\draw[fill=black!50] (0*45: 2) circle (0.1);
			\draw[fill=black!50] (1*45: 2) circle (0.1);			
			\draw[fill=black!50] (2*45: 2) circle (0.1);
			\draw[fill=black!50] (3*45: 2) circle (0.1);			
			\draw[fill=black!100] (4*45: 2) circle (0.1);
			\draw[fill=black!100] (5*45: 2) circle (0.1);			
			\draw[fill=black!100] (6*45: 2) circle (0.1);
			\draw[fill=black!50] (7*45: 2) circle (0.1);			
			
			\draw (0*45: 2.5) node {\scriptsize $906$};
			\draw (2*45: 2.5) node {\scriptsize $914$};
			\draw (3*45: 2.5) node {\scriptsize $633$};
			\draw (4*45: 2.5) node {\scriptsize $98$};
			\draw (5*45: 2.5) node {\scriptsize $279$};
			\draw (6*45: 2.5) node {\scriptsize $547$};
			\draw (7*45: 2.5) node {\scriptsize $958$};
			

			\draw (2.5*45: 1.5) node {\scriptsize $393$};
			\draw (4.5*45: 1.5) node {\scriptsize $172$};
			\draw (5.5*45: 1.5) node {\scriptsize $707$};
			\draw (7.5*45: 1.5) node {\scriptsize $277$};
			
			\draw[line width=1pt, ->, >=stealth] (3.4, -3)  to [out=180, in=-120]  (5.5*45:2.3);
						
		\end{scope}
		
		\begin{scope}[shift={(-1, 0)}]

			\draw[line width=0.5pt] (0*45: 2) -- (1*45:2);
			\draw[line width=0.5pt] (1*45: 2) -- (2*45:2);
			\draw[line width=0.5pt, color=black!50] (4*45: 2) -- (5*45:2);
			\draw[line width=0.5pt] (6*45: 2) -- (7*45:2);
			\draw[line width=0.5pt] (7*45: 2) -- (8*45:2);
			
			\draw[fill=black!50] (0*45: 2) circle (0.1);
			\draw[fill=black!50] (1*45: 2) circle (0.1);			
			\draw[fill=black!50] (2*45: 2) circle (0.1);
			\draw[fill=black!50] (3*45: 2) circle (0.1);			
			\draw[fill=black!50] (4*45: 2) circle (0.1);
			\draw[fill=black!50] (5*45: 2) circle (0.1);			
			\draw[fill=black!50] (6*45: 2) circle (0.1);
			\draw[fill=black!50] (7*45: 2) circle (0.1);
			
			\draw (0*45: 2.5) node {\scriptsize $965$};
			\draw (1*45: 2.5) node {\scriptsize $158$};
			\draw (2*45: 2.5) node {\scriptsize $971$};
			\draw (4*45: 2.5) node {\scriptsize $486$};
			\draw (5*45: 2.5) node {\scriptsize $801$};
			\draw (6*45: 2.5) node {\scriptsize $142$};
			\draw (7*45: 2.5) node {\scriptsize $422$};
			
			\draw (0.5*45: 1.5) node {\scriptsize $47$};
			\draw (1.5*45: 1.5) node {\scriptsize $98$};
			\draw (4.5*45: 1.5) node {\scriptsize $318$};
			\draw (6.5*45: 1.5) node {\scriptsize $35$};
			\draw (7.5*45: 1.5) node {\scriptsize $439$};

			\draw (2, -3) node {\small $2$-edge path and a $3$-edge path in observed graph $\Ob_t$};

		\end{scope}
		
		\draw (+3, 0) node {\large $\ldots$};

		\begin{scope}[shift={(+7, 0)}]

			\draw[line width=0.5pt, color=black!50] (0*45: 2) -- (1*45:2);
			\draw[line width=0.5pt, dashed] (2*45: 2) -- (3*45:2);
			\draw[line width=1.5pt, color=black!100] (3*45: 2) -- (4*45:2);
			\draw[line width=1.5pt, color=black!100] (4*45: 2) -- (5*45:2);
			\draw[line width=1.5pt, color=black!100] (5*45: 2) -- (6*45:2);
			\draw[line width=0.5pt, dashed] (6*45: 2) -- (7*45:2);
			
			\draw[fill=black!50] (0*45: 2) circle (0.1);
			\draw[fill=black!50] (1*45: 2) circle (0.1);			
			\draw[fill=black!50] (2*45: 2) circle (0.1);
			\draw[fill=black!100] (3*45: 2) circle (0.1);			
			\draw[fill=black!100] (4*45: 2) circle (0.1);
			\draw[fill=black!100] (5*45: 2) circle (0.1);			
			\draw[fill=black!100] (6*45: 2) circle (0.1);
			\draw[fill=black!50] (7*45: 2) circle (0.1);						
			
			\draw (0*45: 2.5) node {\scriptsize $916$};
			\draw (1*45: 2.5) node {\scriptsize $793$};
			\draw (3*45: 2.5) node {\scriptsize $656$};
			\draw (4*45: 2.5) node {\scriptsize $36$};
			\draw (5*45: 2.5) node {\scriptsize $850$};
			\draw (6*45: 2.5) node {\scriptsize $934$};
			
			\draw (0.5*45: 1.5) node {\scriptsize $382$};
			\draw (3.5*45: 1.5) node {\scriptsize $187$};
			\draw (4.5*45: 1.5) node {\scriptsize $490$};
			\draw (5.5*45: 1.5) node {\scriptsize $446$};
			
			\draw[line width=1pt, ->, >=stealth] (-6, -2.8)  to [out=70, in=180]  (4.5*45:2.3);			
			
		\end{scope}
		\end{scope}
		
		\end{tikzpicture}
		\caption{A pair of paths with $a=2$ and $b=3$ edges in the observed graph respectively, together with the corresponding $(a+2)$-edge and $(b+2)$-edge paths in the underlying graph (incident edges that have not arrived yet are shown as dashed)}	\label{fig:merge-prob}
	\end{center}
	
\end{figure}
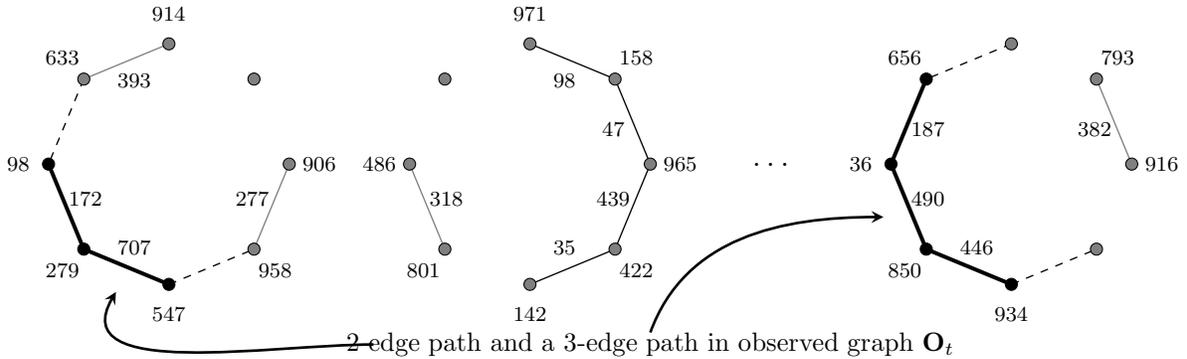

Such a component will be merged at time $t+1$ iff the edge that arrives at
this time is incident to both of them, that is iff either:
\begin{enumerate}
\item $\bpi(t+1) = v_1$ and $u_{a+2} = v_1$.
\item $\bpi(t+1) = u_1$ and $v_{b+2} = u_1$.
\end{enumerate}
So conditioned on $\Ub_t$, the average probability that a pair of an $a$-edge
and a $b$-edge component in $\Ob_t$ is $\frac{1}{n-t}$ times the fraction of
those pairs such that either $u_{a+2} = v_1$ or $v_{b+2} = u_1$. 

Let $\Yb_3$ be the number of paths $(w_i)_{i=1}^{a+b+4}$ such that $w_{i+1} =
\next(w_i)$ for each $i \in \brac*{a+b+3}$, $\bpi^{-1}(w_1) > t$,
$\bpi^{-1}(w_{a+2}) >t$, $\bpi^{-1}(w_{a+b+3}) > t$, and \[
\bpi^{-1}(\set*{w_i : 1 < i < a + 2 \vee a + 2 < i < a + b + 3})
\] is contained in $\brac*{t}$. Let $\Yb_4$ be
the number of paths $(w_i)_{i=1}^{a+b+4}$ such that $w_{i+1} = \next(w_i)$ for
each $i \in \brac*{a+b+3}$, $\bpi^{-1}(w_1) > t$, $\bpi^{-1}(w_{b+2}) >t$,
$\bpi^{-1}(w_{a+b+3}) > t$, and \[ \bpi^{-1}(\set*{w_i : 1 < i < b + 2 \vee b +
2 < i < a + b + 3})\] is contained in $\brac*{t}$. Both correspond to patterns
of type {\bf (4)} above, see Fig.~\ref{fig:2}.

Then the fraction of these pairs such that $u_{a+2} = v_1$ or $v_{b+2} = u_1$ is either $\frac{\Yb_3 + \Yb_4}{\Yb_1 \Yb_2}$ (if $a \not= b$) or $\frac{\Yb_3}{\binom{\Yb_1}{2}}$ (if $a = b$). If $a = b$ and $\Yb_1 \le 1$ the lemma holds trivially, so we may assume $\Yb_1 > 1$ and thus both of these are bounded by \[
2 \frac{\Yb_3 + \Yb_4}{\Yb_1 \Yb_2}
\]
and by applying
Lemma~\ref{lem:patternfreq} with $C = 3$ (as $\ell$ must be at least $3$ and so
$n$ must be as well, and thus $n^{C\ell} > 8n^{\ell + 3}$) and taking a union
bound,
\begin{align*}
\Yb_1 + \Yb_2  &= \Theta\paren*{n\paren*{\frac{t}{n}}^{a + b}
\paren*{\frac{n-t}{n}}^3}\\
\Yb_3 &= \Theta\paren*{n\paren*{\frac{t}{n}}^{a}
\paren*{\frac{n-t}{n}}^2}\\
\Yb_4 &= \Theta\paren*{n\paren*{\frac{t}{n}}^{b}
\paren*{\frac{n-t}{n}}^2}
\end{align*}
with probability $1 - 1/2n^{\ell + 3}$ over $\Ub_t$, provided \[
n/8 \le t \le n - C 2^{10\ell} n^{5/6} \log n
\]
for a sufficiently large constant $C$. So let $\mathcal{E}_t^m$ be the event
that this holds, then $\Pb[\Ub_t]{\mathcal{E}_t^m} \ge 1 - 1/2n^{\ell +
3}$ and for all $U_t \in \mathcal{E}_s^m$,
\begin{align*}
\upsilon(\alpha, \beta, U_t) &= \Theta\paren*{\frac{n\paren*{\frac{t}{n}}^{a +
b} \paren*{\frac{n-t}{n}}^3}{n\paren*{\frac{t}{n}}^a
\paren*{\frac{n-t}{n}}^2 n\paren*{\frac{t}{n}}^b
\paren*{\frac{n-t}{n}}^2}} \cdot \frac{\alpha\lbrack a \rbrack \cdot
\alpha\lbrack b \rbrack}{n-t}\\ 
&= \frac{O\paren*{\alpha\lbrack a \rbrack \cdot \alpha\lbrack b
\rbrack}}{(n-t)^2}
\end{align*}
completing the proof.
\end{proof}

By combining these four claims, Lemma~\ref{lm:ext-prob} is proved.

\subsection{A decomposition of a typical message}\label{sec:decomp}
The main result of this section is Lemma~\ref{lm:decomposition}, a central tool in our analysis. For every $t\in T$ and $1\leq s<t$ the lemma gives a decomposition of a typical message sent by player $t$ (equivalently, the state of the streaming algorithm after processing the $t\nth$ part of the input) in terms of the message sent by player $s$ and an auxiliary function that represents the mapping from the message of player $t$ to the message of player $t$.

As the players are deterministic, we may assume that the $t\nth$ player sends
$\Mb_t:=\gb_t(\Xb_t; \Mb_{t-1})$ for some function \begin{equation}\label{eq:g-def}
\gb_t(\cdot , \cdot) : \{0, 1\}\times \{0, 1\}^c \to \{0, 1\}^c,
\end{equation}
depending only on $\Bb_t$, where $c$ is the size of the messages that the players exchange. We let
$\Fb_t:\bool^t \to \bool$ denote the \emph{indicator function of the message}
$\Mb_t$ sent by the $t\nth$ player, defined as \[
\Fb_t(x)= \begin{cases}
1& \mbox{if, conditioned on $\Bb_t$, $\Xb_t = x$ would cause player $t$ to send $\Mb_t$}\\
0& \mbox{otherwise.}
\end{cases}
\]
For convenience we let $\Fb_0$ denote the trivial function from $\emptyset$ to $\{0, 1\}$, and $\Mb_0=0^c$. Note that $\Fb_t, \Mb_t, \gb_t$ above are random variables depending on $\Xb_{\leq t}$ and $\Bb_t$; we let $F_t$ denote a realization of $\Fb_t$, and let $m_t$ denote a realization of $\Mb_t$ (as usual, we use boldface for random variables and regular font to represent their realizations). We note that indicator functions $F_t$ are in one to one correspondence with
messages $m_t$ . 

Let $\Fbc_t$ denote the collection of possible indicator
functions at time $t$ given $\Bb_t$, and note that for every $t$ such messages
partition $\bool^t$: \[
\sum_{F_t\in \Fbc_t} F_t(x)=1.
\]
Since $F_t$'s are in one to one correspondence with possible messages of the
$t\nth$ player, the function $\gb_t$ from~\eqref{eq:g-def} naturally extends to
map $\bool\times \Fbc_{t-1}$ to $\Fbc_t$:
\begin{equation}
\gb_t(\cdot , \cdot) : \{0, 1\}\times \Fc_{t-1} \to \Fc_t.
\end{equation}

The proof of Lemma~\ref{lm:decomposition} relies on the observation that for every
$t\in T$, $F_t \in \Fbc_t$, and every $x_{\leq t}$ one has
\begin{equation}\label{eq:decomp}
\begin{split}
F_t(x_{\leq t})&=\sum_{F_s\in \Fc_s} F_s(x_{\leq s})\cdot \rb_t(x_{[s+1:t]}; F_s, F_t),
\end{split}
\end{equation}
where $\rb_t(x_{[s+1:t]}; F_s, F_t)$ is the indicator function for
$x_{[s+1:t]}$ taking $F_s$ to $F_t$, i.e.\
\begin{equation}\label{eq:r-def}
\rb(x_{[s+1:t]}; F_s, F_t)=\sum_{(F_{s+1}, \ldots, F_{t-1})\in \Fc_{s+1}\times
\ldots \times \Fc_{t-1}} \mathbf{I}\left[\bigwedge_{j=s+1}^t \gb_j(F_{j-1},
x_j)=F_j\right].
\end{equation}
Note that $\rb$ is a random function depending only on $\Bb_t$. The function $\rb$ satisfies
\begin{claim}\label{cl:prop-r}
For every $1\leq s<t\leq n$, every $F_s\in \Fc_s$ one has for every
$x_{[s+1:t]}$ that \[
\rb(x_{[s+1:t]}; F_s, F_t)\in \bool
\] for every $F_t\in \Fc_t$
and
\[
\sum_{F_t\in \Fc_t} \rb(x_{[s+1:t]}; F_s, F_t)=1.
\]
\end{claim}

We now consider the random variable $\Xb \in \Uc(\bool^n)$, and
corresponding random variables $\Fb_t:\bool^t \to \bool$ for each
$t \in [n]$ to be the indicator of the message that results from
$\Xb_{\leq t}$.  Because $\Xb$ is uniform and the messages are
deterministic, $(\Xb \mid \Fb_t=F_t)$ is uniform over the support of
$F_t$.  Therefore for each $t \in [n]$ and $F_t$ if we define
\[
\wt{F}_t(z) := \frac{2^s}{||F_t||_1} \wh{F}_t(z).
\]
(with $\wt{F}_0$ being the trivial function again) this satisfies
\[
 \wt{F}_t(z) = \E[\Xb ]{(-1)^{z \cdot \Xb_{\le t}}\middle| \Fb_t=F_t, \Bb_t}.
\]

Now consider the distribution of $(\Xb \mid \Fb_s, \Fb_t, \Bb_t)$ for any $s <
t$.  $\Xb = x$ is consistent with $(\Fb_s, \Fb_t, \Bb_t)$ if and only if
$\Fb_s(x_{\leq s}) = 1$ and $\rb(x_{[s+1:t]}; \Fb_s, \Fb_t) = 1$.  Since $\Xb$
is uniform, $(\Xb \mid \Fb_s, \Fb_t, \Bb_t)$ is uniform over those $x$
consistent with $(\Fb_s, \Fb_t, \Bb_t)$; and since the $\Fb_s$ and $\rb$
constraints are on disjoint sets of coordinates, this means that $\Xb_{\leq s}$
and $\Xb_{[s+1:t]}$ are independent conditioned on $(\Fb_s, \Fb_t, \Bb_t)$,
with $(\Xb_{\leq s} \mid \Fb_s, \Fb_t, \Bb_t) = (\Xb_{\leq s} \mid \Fb_s,
\Bb_t)$ and $(\Xb_{[s+1:t]} \mid \Fb_s, \Fb_t, \Bb_t)$ being uniform over the
support of $\rb(\cdot; \Fb_s, \Fb_t)$.  Therefore if we define
\[
\wt{\rb}(z; F_s, F_t):= \frac{2^{t-s}}{||\rb(\cdot ; F_s, F_t)||_1} \wh{\rb}(z;
F_s, F_t)
\]
where the Fourier transform $\wh{\rb}$ is with respect to its first argument,
this satisfies
\[
\wt{\rb}(z; F_s, F_t)=  \E[\Xb]{(-1)^{z \cdot
\Xb_{[s+1:t]}} \middle| \Fb_s=F_s, \Fb_t=F_t, \Bb_t}.
\] 

These definitions lead to the following:
\lmdecomposition*

\begin{proof}
  As discussed above, $\Xb_{\leq s}$ and $\Xb_{[s+1:t]}$ are independent
  conditioned on $\Fb_s$, $\Fb_t$, $\Bb_t$.  Therefore
  \begin{align*}
    \E[\Xb]{(-1)^{z \cdot \Xb} \middle| \Fb_t} &= \E[\Fb_s]{\E[\Xb
    ]{(-1)^{z \cdot \Xb} \middle|\Fb_s, \Fb_t, \Bb_t} \middle|\Fb_t, \Bb_t}\\
    &= \E[\Fb_s]{\E[\Xb_{\leq s}, \Xb_{[s+1:t]}]{(-1)^{z_{\leq s}
    \cdot \Xb_{\leq s}} (-1)^{z_{[s+1:t]} \cdot \Xb_{[s+1:t]}} \middle|\Fb_s,
    \Fb_t, \Bb_t} \middle|\Fb_t, \Bb_t}\\
    &= \E[\Fb_s]{ \E[\Xb_{\leq s}]{(-1)^{z_{\leq s} \cdot \Xb_{\leq
    s}} \middle|\Fb_s, \Fb_t, \Bb_t}\E[\Xb_{[s+1:t]}]{ (-1)^{z_{[s+1:t]}
    \cdot \Xb_{[s+1:t]}} \middle|\Fb_s, \Fb_t, \Bb_t} \middle|\Fb_t, \Bb_t}\\
    &= \E[\Fb_s]{\wt{\Fb}_s(z_{\leq s}) \cdot \wt{\rb}(z_{[s+1:t]}; \Fb_s,
    \Fb_t)  \middle|\Fb_t, \Bb_t}.
  \end{align*}
\end{proof}

In the next section we combine Lemma~\ref{lm:decomposition} with bounds on the probability of extensions and merges from Section~\ref{sec:comb}  to obtain Lemma~\ref{lm:evolution}, our key lemma on the expected evolution of Fourier coefficients as a function of $t$.

\subsection{Evolution of Fourier coefficients (Lemma~\ref{lm:evolution})}\label{sec:evolution}

Define, for every $t\in [n]$ and $\beta\in \zpp$:
\begin{equation}\label{eq:fu-def}
\Hb_\beta^t=\sum_{z\in \{0, 1\}^t, z\bsim_t \beta} \wt{\Fb}_t(z)^2
\end{equation}

The main result of this section is Lemma~\ref{lm:evolution}, restated here for convenience of the reader: 
\lmevolution*

Before proving Lemma~\ref{lm:evolution}, we derive several useful bounds on Fourier coefficients of Boolean functions, which we later apply to $\Fb_t$.
To bound sums of Fourier coefficients of $\Fb_t$, we need a corollary of the KKL
Lemma of~\cite{KKL88}, an extended version of that of~\cite{GavinskyKKRW07}.

\begin{lemma}[\cite{KKL88}]
\label{lem:kkl}
Let $f$ be a function $f : \lbrace 0, 1\rbrace^n \rightarrow \lbrace -1, 0, 1
\rbrace$. Let $A = \lbrace x | f(x) \not= 0 \rbrace$. Then for every $\delta
\in \brac*{0,1}$ we have \[
\sum_{s \in \bool^n} \delta^{|s|} \widehat{f}(s)^2 \le
\paren*{\frac{|A|}{2^n}}^\frac{2}{1 + \delta}\text{.}
\]
\end{lemma}

We will use 

\begin{lemma}[Lemma F.3 of~\cite{KapralovK19}]\label{xorKKL}
Let $f$ be a function $f : \lbrace 0, 1\rbrace^n \rightarrow \lbrace 0, 1
\rbrace$. Let $A = \lbrace x | f(x) \not= 0 \rbrace$, and let $|A| \ge 2^{n -
c}$.
Then for every $y\in \{0,1\}^n$ and every $q\leq c$ one has 
\[
\sum_{\substack{x\in \{0,1\}^n \\ |x\oplus y|=q}} \wt{f}^2(x) \leq \left(\frac{4c}{q}\right)^{q}.
\]
\end{lemma}

The following lemma is an immediate consequence of Lemma~\ref{xorKKL}:
\begin{lemma}\label{lm:fouriercoeffs}
Let $f$ be a function $f : \lbrace 0, 1\rbrace^n \rightarrow \lbrace 0, 1
\rbrace$. Let $A = \lbrace x | f(x) \not= 0 \rbrace$, and let $|A| \ge 2^{n -
c}$. Then for every $k \in \lbrace 1, \dots, c \rbrace$, we have 
\[
\sum_{v \in \{0, 1\}^{n-1}, |v|=k} \wt{f}(1\cdot v)^2 \le \left(\frac{4c}{k}\right)^k\text{.}
\]
\end{lemma}

\begin{lemma}\label{lm:fouriercoeffs-paths}
Let $f$ be a function $f : \lbrace 0, 1\rbrace^n \rightarrow \lbrace 0, 1
\rbrace$. Let $A = \lbrace x | f(x) \not= 0 \rbrace$, and let $|A| \ge 2^{n -
c}$. Then for every $k \in \lbrace 1, \dots, c \rbrace$, any partition $P$ of
$\lbrack n\rbrack$, and any set $S$ consisting of size-$k$ unions of elements
of $P$, we have \[
\sum_{v \in S} \wt{f}(v)^2 \le \left(\frac{4c}{k}\right)^k\text{.}
\]
\end{lemma}
\begin{proof}
For each $s \in P$, write $j_s$ for the first index in $s$. For each $z \in \lbrace 0, 1\rbrace^n$, let $z'$ be defined by: \[
z_j' = \begin{cases}
z \cdot s \bmod 2 & \mbox{if $j = j_s$ for some $s \in P$}\\
z_j & \mbox{otherwise.}
\end{cases}
\]
Then, define: \[
A' = \lbrace z' : z \in A \rbrace
\]
The transformation $z \rightarrow z'$ is bijective, as it is given by a
triangular matrix with ones along the diagonal. Therefore, $|A'| = |A|$, and so
$|A'|/2^n = 2^{-c}$. Now, let $f'$ be the indicator of $A'$. By the definition
of $z \rightarrow z'$, for all $V \subseteq P$,
\begin{align*}
\wt{f}\left(\bigcup V\right) &= \frac{1}{|A|}\sum_{z \in A} \prod_{s \in V}
(-1)^{z \cdot s}\\
&= \frac{1}{|A|}\sum_{z \in A} \prod_{s \in V} (-1)^{z'_{j_s}}\\
&= \frac{1}{|A'|}\sum_{z \in A'} \prod_{s \in V} (-1)^{z \cdot \lbrace j_s
\rbrace}\\
&= \wt{f'}\left(\bigcup_{s \in V}\lbrace j_s\rbrace \right)
\end{align*}
and so by applying Lemma~\ref{lem:kkl} with $\delta = k/c$,
\begin{align*}
\sum_{v \in S} \wt{f}(v)^2 &\le \sum_{v \in \bool^n : |v| = k} \wt{f'}(v)^2\\
&= \frac{2^{2n}}{|A'|^2} \sum_{v \in \bool^n : |v| = k} \widehat{f}(v)^2\\
&\le \delta^{-k} \frac{2^{2n}}{|A'|^2} \sum_{v \in \bool^n} \delta^{|v|}
\widehat{f}(v)^2\\
&\le \delta^{-k} \frac{2^{2n}}{|A'|^2} \paren*{\frac{|A'|}{2^n}}^\frac{2}{1 + \delta}\\
&= \delta^{-k} 2^\frac{2c\delta}{1 + \delta}\\
&\le \paren*{\frac{4c}{k}}^k\text{.}
\end{align*}
\end{proof}

We will use this with Parseval's identity to obtain a bound on coefficients of
$r$. Here $\norm{.}_2$ is the functional $\ell_2$-norm given by $\norm{f}_2^2
= 2^{-n}\sum_{x \in \bool^n} \abs{f(x)}^2$.
\begin{lemma}[Parseval]
For every function $f : \bool^n \rightarrow \mathbb{R}$, \[
\sum_{v \in \bool^n}\widehat{f}(v)^2 = \norm{f}_2^2\text{.}
\]
\end{lemma}
Note that if, as in the previous lemma statements, $f$ is the indicator
function of a set $A$ such that $|A| = 2^{n-c}$, this implies that \[
\sum_{v \in \bool^n}\wt{f}(v)^2 = 2^c
\]
as $\wt{f} = 2^{n-c}\widehat{f}$.
\begin{lemma}
\label{lem:rbound}
For any $s < t \in T$, and any
$k \in \lbrack t- s-1\rbrack$, \[
\sum_{\substack{z \in \lbrace 0, 1\rbrace^{t-s-1} \\ \text{$|z|=k$}}}
\E[\Xb_{\lbrack s+1:t \rbrack}]{ \wt{\rb}(1\cdot z; \Fb_s, \Fb_t)^2 |
\Xb_{\le s}, \Bb_t } \le q(k)
\] 
where $q(k) = \begin{cases}
\left(\frac{8c}{k}\right)^k & \mbox{if $k \le c$}\\
2^c & \mbox{otherwise.}
\end{cases}$
\end{lemma}
\begin{proof}
Condition on $\Bb_t = B_t$, $\Xb_{\le s} = y$, for any $B_t, y$ in the support of
$\Bb_t, \Xb_{\le s}$. Let $F_s$ be the resulting value of $\Fb_s$. Conditioned on
these, each $x \in \bool^{t-s}$ results in a fixed value of $\Fb_t$ when
$\Xb_{\brac*{s+1 : t}} = x$.  As $\Fb_t$ is supported on at most $2^c$ elements,
this gives a size at most $2^c$ partition of $\bool^{t-s}$. Call this partition
$\Ac$, and for each $F_t$ in the support of $\Fb_t$ conditioned on $\Bb_t = B_t$,
$\Xb_{\le s} = y$, call the (unique) corresponding element of the partition
$A(F_t)$.

For each $A \in \Ac$, write $c_A$ for $n - \log |A|$, and $p_A$ for
$\Pb{\Fb_t = F_t | \Xb_{\le s} = y, \Bb_t = B_t}$, where $A(F_t) = A$.  As
$\Fb_t = F_t$ iff $\Xb_{\brac*{s+1:t}} \in A$, $p_A = |A|2^{-n} = 2^{-c_A}$.  Note that
by Lemma~\ref{lm:fouriercoeffs} and Parseval's equality, for each $F_t$ in the
support of $\Fb_t$ conditioned on $\Bb_t = B_t$, $\Xb_{\le s} = y$, \[
\sum_{\substack{z \in \lbrace 0, 1\rbrace^{t-s-1} \\ \text{$|z|=k$}}} \wt{\rb}(1\cdot z; F_s, F_t)^2 \le q_{A(F_t)}(k)
\]
where for each $A \in \Ac$, \[
q_A(k) = \begin{cases}
\min\left(2^{c_A}, \left(\frac{4c_A}{k}\right)^k\right) & \mbox{if $k \le
c_A$}\\ 2^{c_A} & \mbox{otherwise}
\end{cases}
\] (note the different constants from $q$ in the first line). We will
consider two cases, based on the value of $k = |z|$.

\paragraph{Case 1: $c < 2k$} For each $A \in \mathcal{A}$, $q_A(k) \le 2^{c_A}$
and so $p_Aq_A(k) \le 1$. Therefore,
\begin{align*}
\sum_{\substack{z \in \lbrace 0, 1\rbrace^{t-s-1} \\ \text{$|z|=k$}}} \E[\Xb_{\lbrack s+1:t \rbrack}] { \wt{\rb}(1\cdot z;
\Fb_s, \Fb_t)^2| \Xb_{\le s} = y, \Bb_t = B_t} &\le \sum_{A \in \mathcal{A}}
p_Aq_A(k)\\
&\le |\mathcal{A}|\\
&\le 2^{c}\\
&\le q(k)
\end{align*}
as if $k \ge c$, $q(k) = 2^c$, and if $k \in (c/2, c\rbrack$, $q(k) \ge
8^{c/2} \ge 2^c$.

\paragraph{Case 2: $c \ge 2k$} Let $\mathcal{A} = \mathcal{A}^+ \cup \mathcal{A}^-$, with
$\mathcal{A}^+$ containing all $A$ in $\mathcal{A}$ such that $c_A > c$ and $\mathcal{A}^-$
containing everything else. Then $q_A(k) \le \left(\frac{4c}{k}\right)^k$ for
each $A$ in $\mathcal{A}^-$ and so \[
\sum_{a \in A^-} p_Aq_A(k) \le \left(\frac{4c}{k}\right)^k.
\]
At the same time for each $A \in A^+$, $c_A > c \ge 2k$, so
\begin{align*}
p_Aq_A &= \left(\frac{4c_A}{k}\right)^k 2^{-c_A}\\
&\le \left(\frac{4c}{k}\right)^k2^{-c}
\end{align*}
as for $x \ge 2k$, $\frac{d}{dx}(x^k2^{-x}) = (k - x)x^{k - 1}2^{-x} < 0$.

Therefore,
\begin{align*}
\sum_{\substack{z \in \lbrace 0, 1\rbrace^{t-s-1} \\ \text{$|z|=k$}}} \E[\Xb_{\lbrack s+1:t \rbrack}] { \wt{\rb}(1\cdot z; \Fb_s, \Fb_t)^2 |
\Xb_{\le s} = y, \Bb_t = B_t } &\le \sum_{A \in \mathcal{A}^+} p_Aq_A(k) +
\sum_{A \in \mathcal{A}^-} p_Aq_A(k)\\ &\le \left(\frac{4c}{k}\right)^k +
|\mathcal{A}|\left(\frac{4c}{k}\right)^k2^{-c}\\
&\le 2\left(\frac{4c}{k}\right)^k\\
&\le q(k)
\end{align*}
and so the result follows.
\end{proof}

\begin{lemma}
\label{lem:rbound-paths}
For any $s \le t \in T$, and any
$k \in \lbrack t- s\rbrack$, \[
\sum_{\substack{z \in \lbrace 0, 1\rbrace^{t-s} \\ \text{$z$ a union of $k$
paths}}}
\E[\Xb_{\lbrack s+1:t \rbrack}]{ \wt{\rb}(z; \Fb_s, \Fb_t)^2 |
\Xb_{\le s}, \Bb_t } \le q(k)
\] 
where $q(k) = \begin{cases}
\left(\frac{8c}{k}\right)^k & \mbox{if $k \le c$}\\
2^c & \mbox{otherwise.}
\end{cases}$
\end{lemma}
\begin{proof}
The proof follows the proof of Lemma~\ref{lem:rbound} except for its reliance on Lemma~\ref{lm:fouriercoeffs-paths} as opposed to Lemma~\ref{lm:fouriercoeffs}.
Condition on $\Bb_t = B_t$, $\Xb_{\le s} = y$, for any $B_t, y$ in the support of
$\Bb_t, \Xb_{\le s}$. Let $F_s$ be the resulting value of $\Fb_s$. Conditioned on
these, each $x \in \bool^{t-s}$ results in a fixed value of $\Fb_t$ when
$\Xb_{\brac*{s+1 : t}} = x$.  As $\Fb_t$ is supported on at most $2^c$ elements,
this gives a size at most $2^c$ partition of $\bool^{t-s}$. Call this partition
$\Ac$, and for each $F_t$ in the support of $\Fb_t$ conditioned on $\Bb_t = B_t$,
$\Xb_{\le s} = y$, call the (unique) corresponding element of the partition
$A(F_t)$.

For each $A \in \Ac$, write $c_A$ for $n - \log |A|$, and $p_A$ for
$\Pb{\Fb_t = F_t | \Xb_{\le s} = y, \Bb_t = B_t}$, where $A(F_t) = A$.  As
$\Fb_t = F_t$ iff $\Xb_{\brac*{s+1:t}} \in A$, $p_A = |A|2^{-n} = 2^{-c_A}$.  Note that
by Lemma~\ref{lm:fouriercoeffs-paths} and Parseval's equality, for each $F_t$ in the
support of $\Fb_t$ conditioned on $\Bb_t = B_t$, $\Xb_{\le s} = y$, \[
\sum_{\substack{z \in \lbrace 0, 1\rbrace^{t-s} \\ \text{$z$ a union of $k$
paths}}} \wt{\rb}(z; F_s, F_t)^2 \le q_{A(F_t)}(k)
\]
where for each $A \in \Ac$, \[
q_A(k) = \begin{cases}
\min\left(2^{c_A}, \left(\frac{4c_A}{k}\right)^k\right) & \mbox{if $k \le
c_A$}\\ 2^{c_A} & \mbox{otherwise}
\end{cases}
\] (note the different constants from $q$ in the first line). We will
consider two cases, based on the value of $k = |z|$.

\paragraph{Case 1: $c < 2k$} For each $A \in \mathcal{A}$, $q_A(k) \le 2^{c_A}$
and so $p_Aq_A(k) \le 1$. Therefore,
\begin{align*}
\sum_{\substack{z \in \lbrace 0, 1\rbrace^{t-s} \\ \text{$z$ a union of $k$
paths}}} \E[\Xb_{\lbrack s+1:t \rbrack}] { \wt{\rb}(z;
\Fb_s, \Fb_t)^2| \Xb_{\le s} = y, \Bb_t = B_t} &\le \sum_{A \in \mathcal{A}}
p_Aq_A(k)\\
&\le |\mathcal{A}|\\
&\le 2^{c}\\
&\le q(k)
\end{align*}
as if $k \ge c$, $q(k) = 2^c$, and if $k \in (c/2, c\rbrack$, $q(k) \ge
8^{c/2} \ge 2^c$.

\paragraph{Case 2: $c \ge 2k$} Let $\mathcal{A} = \mathcal{A}^+ \cup \mathcal{A}^-$, with
$\mathcal{A}^+$ containing all $A$ in $\mathcal{A}$ such that $c_A > c$ and $\mathcal{A}^-$
containing everything else. Then $q_A(k) \le \left(\frac{4c}{k}\right)^k$ for
each $A$ in $\mathcal{A}^-$ and so \[
\sum_{a \in A^-} p_Aq_A(k) \le \left(\frac{4c}{k}\right)^k.
\]
At the same time for each $A \in A^+$, $c_A > c \ge 2k$, so
\begin{align*}
p_Aq_A &= \left(\frac{4c_A}{k}\right)^k 2^{-c_A}\\
&\le \left(\frac{4c}{k}\right)^k2^{-c}
\end{align*}
as for $x \ge 2k$, $\frac{d}{dx}(x^k2^{-x}) = (k - x)x^{k - 1}2^{-x} < 0$.

Therefore,
\begin{align*}
\sum_{\substack{z \in \lbrace 0, 1\rbrace^{t-s} \\ \text{$z$ a union of $k$
paths}}} \E[\Xb_{\lbrack s+1:t \rbrack}] { \wt{\rb}(z; \Fb_s, \Fb_t)^2 |
\Xb_{\le s} = y, \Bb_t = B_t } &\le \sum_{A \in \mathcal{A}^+} p_Aq_A(k) +
\sum_{A \in \mathcal{A}^-} p_Aq_A(k)\\ &\le \left(\frac{4c}{k}\right)^k +
|\mathcal{A}|\left(\frac{4c}{k}\right)^k2^{-c}\\
&\le 2\left(\frac{4c}{k}\right)^k\\
&\le q(k)
\end{align*}
and so the result follows.
\end{proof}

We now give a proof of Lemma~\ref{lm:evolution}. 

\begin{proofof}{Lemma~\ref{lm:evolution}}
By Lemma~\ref{lm:decomposition}, for every $t\in T$, every $s<t$, and every
$z\in \bool^t$, we have
\[
\wt{\Fb}_t(z)=\E[\Fb_s]{\wt{\Fb}_s(z_{\leq s}) \cdot
\wt{\rb}(z_{[s+1:t]}; \Fb_s, \Fb_t)\middle|\Fb_t, \Bb_t}.
\]
 As exactly
 one extension event $\grow(z, s, t)$ will hold, we have
\begin{equation*}
\begin{split}
\wt{\Fb}_t(z)=\sum_{s = 1}^{t-1}    \grow(z, s, t)\cdot\E[\Fb_s]
{\wt{\Fb}_s(z_{\leq s}) \cdot \wt{\rb}(z_{[s+1:t]}; \Fb_s, \Fb_t)\middle|\Fb_t,
\Bb_t},
\end{split}
\end{equation*}
where we condition on $\Bb_t$ on both sides, so $\grow(z, s, t)$ is well-defined.
Since the above sum contains only one nonzero term, we have by Jensen's
inequality
\begin{equation}\label{eq:8h238fg23D}
\begin{split}
\wt{\Fb}_t(z)^2\leq \sum_{s = 1}^{t-1}  \grow(z, s, t)\cdot  \E[\Fb_s]
{\wt{\Fb}_s(z_{\leq s})^2 \cdot \wt{\rb}(z_{[s+1:t]}; \Fb_s, \Fb_t)^2|\Fb_t,
\Bb_t}.
\end{split}
\end{equation}

Summing~\eqref{eq:8h238fg23D} over all $z$ such that $z\bsim_t \beta$, taking
expectations over $\Xb$ and noting that for every $s\in T$ such that $\grow(z,
s, t) = 1$ one has $z_{\leq s}\bsim_s \alpha$ for some $\alpha\in \beta-1$ (recall Definition~\ref{def:sim})), we
get
\begin{equation}\label{eq:9324hg923g}
\begin{split}
\E[\Xb]{ \Hb_\beta^t | \Bb_t } &= \E[\Xb] { \sum_{z: z\bsim_t \beta}
\wt{\Fb}_t(z)^2 \middle| \Bb_t  }\\
&= \sum_{s = 1}^{t-1} \sum_{z: z\bsim_t \beta} \grow(z, s, t)\cdot
\E[\Fb_t] { \E[\Fb_s]{\wt{\Fb}_s(z_{\leq s})^2 \cdot
\wt{\rb}(z_{[s+1:t]}; \Fb_s, \Fb_t)^2 \middle| \Fb_t, \Bb_t}\middle| \Bb_t}\\
&= \sum_{s = 1}^{t-1} \sum_{z: z\bsim_t \beta} \grow(z, s, t)\cdot
\E[\Fb_s, \Fb_t] {\wt{\Fb}_s(z_{\leq s})^2 \cdot
\wt{\rb}(z_{[s+1:t]}; \Fb_s, \Fb_t)^2 \middle| \Bb_t}\\
&= \sum_{s = 1}^{t-1} \sum_{z: z\bsim_t \beta} \grow(z, s, t)\cdot
\E[\Fb_s] {\wt{\Fb}_s(z_{\leq s})^2 \cdot \E[\Fb_t]
{\wt{\rb}(z_{[s+1:t]}; \Fb_s, \Fb_t)^2 \middle| \Fb_s, \Bb_t} \middle| \Bb_t}\\
\end{split}
\end{equation}

For every $\beta\in \zpp$ and $\alpha\in \beta-1$ we let $\beta^{-\alpha}$ be given by setting $\beta\lbrack 1 \rbrack =
\alpha\lbrack 1 \rbrack$. Splitting each $z
\bsim_t \beta$ into the co-ordinates before $\grow(z,s,t)$ holds, namely $z_{\leq s}$, and those after, namely $z_{[s+1:t]}$ (note that the latter all corresponding to isolated edges in $\Ob_t$), we get

\begin{equation}\label{eq:0923ht902jgsfadfd}
\begin{split}
&\sum_{z: z\bsim_t \beta} \grow(z, s, t)\cdot \E[\Fb_s] {\wt{\Fb}_s(z_{\leq
s})^2 \cdot \E[\Fb_t] {\wt{\rb}(z_{[s+1:t]}; \Fb_s, \Fb_t)^2 \middle| \Fb_s,
\Bb_t} \middle| \Bb_t}\\
&\le \sum_{a=0}^{|\beta|} \sum_{\substack{\alpha\in \beta-1,\\
|\beta-\alpha|=a}} \sum_{\substack{r\in \{0, 1\}^s:\\ r\bsim_s \alpha\\
r\cdot 1 \bsim_{s+1} \beta^{-\alpha}}} \E[\Fb_s] {\wt{\Fb}_s(r)^2
\sum_{\substack{p\in \{0, 1\}^{t-s-1}, |p|=a,\\ r \cdot 1 \cdot p \bsim_t
\beta}}\E[\Xb]{\wt{\rb}(1\cdot p; \Fb_s, \Fb_t)^2 \middle|\Fb_s, \Bb_t } \middle|
\Bb_t }\\
\end{split}
\end{equation}

Using Lemma~\ref{lem:rbound} we upper bound the sum in the expectation in~\eqref{eq:0923ht902jgsfadfd} by
\begin{equation}\label{eq:r-kkl}
\begin{split}
\sum_{\substack{p\in \bool^{t-s-1}, |p|=a,\\ r\cdot 1 \cdot p
\bsim_t\beta}}\E[\Xb]{\wt{\rb}(1\cdot p; \Fb_s, \Fb_t)^2|\Fb_s, \Bb_t}\leq
\sum_{\substack{p\in \{0, 1\}^{[s+2:t]} \\|p|=a}}\E[\Xb]{\wt{\rb}(1\cdot p;
\Fb_s, \Fb_t)^2 \middle|\Fb_s, \Bb_t}\leq q(a).
\end{split}
\end{equation}
Indeed, note that the bound of Lemma~\ref{lem:rbound} holds conditioned on \emph{any} value
of $X_{\le s}$, and thus certainly holds conditioned on $\Fb_s$. Substituting
the above into~\eqref{eq:9324hg923g}, we get 
\begin{equation}\label{eq:9ihfifhFIHC3g}
\begin{split}
\E[\Xb]{\Hb_\beta^t \middle| \Bb_t}
&\leq \sum_{s = 1}^{t-1}  \sum_{a=0}^{|\beta|} q(a) \sum_{\substack{\alpha\in \beta-1,\\
|\beta-\alpha|=a}} \sum_{\substack{r\in \{0, 1\}^s:\\ r\bsim_s \alpha\\
r\cdot 1 \bsim_{s+1} \beta^{-\alpha}}}
\E[\Fb_s] {\wt{\Fb}_s(r)^2 \middle| \Bb_t }\\
&=  \sum_{s = 1}^{t-1}  \sum_{a=0}^{|\beta|} q(a) \sum_{\substack{\alpha\in
\beta-1,\\ |\beta-\alpha|=a}} \sum_{\substack{r\in \{0, 1\}^s:\\ r\bsim_s
\alpha\\ r\cdot 1 \bsim_{s+1} \beta^{-\alpha}}} \E[\Xb]
{\wt{\Fb}_s(r)^2 \middle| \Bb_s}
\end{split}
\end{equation}
as $\Fb_s$ is independent of $\Bb_{s+1:t}$.  Note that the condition $r\bsim_s \alpha$ and $r\cdot 1 \bsim_{s+1} \beta^{-\alpha}$ above makes sense since we condition on $\Bb_t$ on both sides.   Taking expectations with
respect to $\Bb_t$ of both sides of~\eqref{eq:9ihfifhFIHC3g}, we get
\begin{equation*}
\begin{split}
\E[\Xb, \Bb_t]{\Hb_\beta^t}&\leq \sum_{s = 1}^{t-1} \sum_{a=0}^{|\beta|} q(a)
\sum_{\substack{\alpha\in \beta-1,\\ |\beta-\alpha|=a}}  \E[\Bb_s]
{ \E[\Bb_{s+1:t}]{ \sum_{\substack{r\in \{0,
1\}^s:\\ r\bsim_s \alpha\\ r\cdot 1 \bsim_{s+1} \beta^{-\alpha}}}\E[\Xb]
{ \wt{\Fb}_s(z_{\leq s})^2 \middle| \Bb_s} \middle| \Bb_s
} }\\
&\leq \sum_{s = 1}^{t-1} \sum_{a=0}^{|\beta|} q(a) \sum_{\substack{\alpha\in
\beta-1,\\ |\beta-\alpha|=a}}  \E[\Bb_s, \Fb_s] {
\sum_{\substack{r\in \{0, 1\}^s:\\ r\bsim_s \alpha}}\Pb[\bb_{s+1}]{
r\cdot 1 \bsim_{s+1} \beta^{-\alpha} | \Bb_s } \wt{\Fb}_s(z_{\leq s})^2 \middle|
\Bb_s }\\
&= \sum_{s = 1}^{t-1} \sum_{\substack{\alpha\in \beta-1}}  q(|\beta-\alpha|)
\sum_{z: z\bsim_s \alpha} \E[\Xb, \Bb_s] {\wt{\Fb}_s(z_{\leq s})^2
\cdot p_s(\alpha, \beta, \Bb_s)}\\
&= \sum_{s = 1}^{t-1} \sum_{\substack{\alpha\in \beta-1}}  q(|\beta-\alpha|)
\cdot \E[\Xb, \Bb_s] {\Hb_\alpha^s \cdot p_s(\alpha, \beta, \Bb_s)}
\end{split}
\end{equation*}
as required. 

For the base case where $\beta$ is all size 1 components, and as $\Fb_0$ is
always the trivial function on $\emptyset$, we apply
Lemma~\ref{lm:decomposition} to get
\begin{align*}
\Hb_\beta^t &\le \sum_{z : z \bsim_t \beta} \E[\Xb, \Bb_t]{ \wt{\rb}(z, \Fb_0,
\Fb_t)^2 }\\
&\le q(|\beta|)
\end{align*}
by Lemma~\ref{lem:rbound}.
\end{proofof}

\subsection{Main lemma (Lemma~\ref{lem:finalH})}\label{sec:main-lemma}

The main result of this section is Lemma~\ref{lem:finalH}, restated here for convenience of the reader:
\lemfinalH*

Before proving the lemma, we introduce a useful definition of potential function $p$ below, prove a useful property (Claim~\ref{cl:s-prop}), as well as establish two technical lemmas. The first lemma,  namely Lemma~\ref{lem:hbetabound} below, bounds the expected evolution of $\Hb^t_\beta$ until almost the end of the stream. Lemma~\ref{lem:almostclosed}  provides a useful combinatorial characterization of the board towards the end of the stream.
\begin{definition}
For every $\beta\in \zpp$ define $\nu(\beta) = \sum_{i \in \beta} \left\lceil \frac{i}{2} \right\rceil$.
\end{definition}

We will need the following properties of $\nu(\beta)$:
\begin{claim}\label{cl:s-prop}
For every $\beta\in \zpp$ and every $\alpha\in \beta-1$ one has $\nu(\alpha)\leq \nu(\beta)$.
\end{claim}
\begin{proof}
We first consider {\bf extensions}, i.e.\ suppose that $\beta$ can be obtained from $\alpha$ by incrementing one of the elements in $\alpha$ by $1$, followed by possibly adding an arbitrary number of $1$'s to $\alpha$. It suffices to note that every added $1$ contributes $0$ or $1$ to $\nu(\beta)-\nu(\alpha)$, and the increment similarly contributes either $1$ or $0$.

We now consider {\bf merges}, i.e.\ suppose that $\beta$ is obtained from $\alpha$ by replacing two elements $a, b\in \alpha$ with $a+b+1$, followed by possibly adding an arbitrary number of $1$'s to $\alpha$. As before, every added $1$ contributes $0$ or $1$ to $\nu(\beta)-\nu(\alpha)$. We now verify that $\left\lceil \frac{a+b+1}{2}\right\rceil \geq \left\lceil {\frac{a}{2}}\right\rceil+ \left\lceil {\frac{b}{2}}\right\rceil$ for all non-negative integers $a, b$. If one of $a$ and $b$ is even (suppose it is $a$), we get 
$$
\left\lceil \frac{a+b+1}{2}\right\rceil = \left\lceil\frac{a}{2}\right\rceil+\left\lceil \frac{b+1}{2}\right\rceil \geq \left\lceil {\frac{a}{2}}\right\rceil+ \left\lceil {\frac{b}{2}}\right\rceil.
$$
If both are odd, then 
$$
\left\lceil \frac{a+b+1}{2}\right\rceil=\left\lceil \frac{a+b+2}{2}\right\rceil =\left\lceil {\frac{a+1}{2}}\right\rceil+ \left\lceil {\frac{b+1}{2}}\right\rceil \geq \left\lceil {\frac{a}{2}}\right\rceil+ \left\lceil {\frac{b}{2}}\right\rceil,
$$
as required.
\end{proof}

\begin{lemma}
\label{lem:hbetabound}
For every $\beta\in \zpp$ such that $\abs{\beta}_*
< \ell$, $t\in \lbrack n - 2^\ell\sqrt{n\cdot c} - Q n^{5/6}
2^{2\ell} \log n\rbrack$ one has $$
\E{\Hb^t_\beta} \leq \left(\prod_{j\in \beta} \frac{1}{j!}\right) \cdot
Q^{\abs{\beta}_*} \cdot \left(\frac{t}{n}\right)^{\abs{\beta}_* -
\nu(\beta)}\cdot c^{|\beta|}  + Qn^{\abs{\beta}_* - \ell}
$$
for some absolute constant $Q>1$.
\end{lemma}
\begin{proof}
We will proceed by induction on $\abs{\beta}_*$. Recall that as per Definition~\ref{def:weight-and-size} we denote the weight of $\beta$ by 
$$
|\beta|_*:=\sum_{i\in \beta} i
$$
and the size of $\beta$ (i.e., the number of elements in $\beta$) by $|\beta|$. Applying
Lemma~\ref{lm:evolution}. For our base case, note that if $\abs{\beta}_* = 1$
then $\beta = \lbrace 1\rbrace$, and so, \[
\Hb_\beta^t \le q(|\beta|) \le (8c)^{|\beta|}
\]
for all $t$. Now suppose $\abs{\beta}_* > 1$ and the result holds for all
$\alpha$ with $\abs{\alpha}_* < \abs{\beta}_*$. If every element of $\beta$ is
$1$, we again have \[
\Hb_\beta^t \le q(|\beta|) \le (8c)^{|\beta|}
\]
for all $t$. Otherwise, we have 
\begin{equation}\label{eq:0283t382gf}
\E[\Xb, \Bb_t] {\Hb^t_\beta} \leq \sum_{s = 1}^t \sum_{\alpha\in
\beta-1} \E[\Xb, \Bb_s]{ \Hb_\alpha^s\cdot q(|\beta-\alpha|)\cdot
p(\alpha, \beta, \Bb_s)}
\end{equation}
with every term of the outer sum being zero if $\beta$ cannot be made from
$\alpha$ by either an extension or a merge (with possible addition of ones in
either case). Let $A^e$ represent the set of $\alpha$ that can become $\beta$
through extension (and adding ones) and $A^m$ represent the set of $\alpha$
that can become $\beta$ through merges (and adding ones). Let $\mathcal{E}_s$ denote the event from Lemma~\ref{lm:ext-prob} such that 

\[
p_{s}(\alpha, \beta, B_t) \le \begin{cases}
\frac{O(\alpha\lbrack a\rbrack)}{n} &\mbox{if $\alpha\to \beta$ is an extension
of a path of size $a$}\\
\frac{O(\alpha \lbrack a\rbrack \cdot \alpha \lbrack b\rbrack)}{(n - s)^2}
&\mbox{if $\alpha\to \beta$ is a merge of paths of size $a$ and $b$.}
\end{cases}
\]
for each $B_s \in \mathcal{E}_s$. Substituting this bound in~\eqref{eq:0283t382gf} and using the fact that $\Hb_\alpha^s$ is always non-negative, we get
\begin{align*}
\E{ \Hb_\beta^t } &\le \sum_{\alpha \in A^e}\sum_{s = 0}^{t}
q(|\beta -\alpha|)
\E{ \Hb_\alpha^s } p_s(\alpha, \beta)
+  \sum_{\alpha \in A^m}\sum_{s = 0}^{t}q(|\beta - \alpha|)
\E{ \Hb_\alpha^s } p_s(\alpha, \beta)\\
&+ \sum_{\alpha \in \beta - 1}\sum_{s = 0}^{t} q(|\beta -
\alpha|) \E{ \Hb_\alpha^s \middle| \overline{\mathcal{E}_s} } \Pb
{ \overline{\mathcal{E}_s}} 
\stepcounter{equation}\tag{\theequation}\label{eq:threecontribs}
\end{align*}
where 
\begin{equation}\label{eq:p-bounds}
p_s(\alpha, \beta) = \begin{cases}
\frac{O(\alpha\lbrack a\rbrack)}{n} &\mbox{if $\alpha\to \beta$ is an extension
of a path of size $a$}\\
\frac{O(\alpha \lbrack a\rbrack \cdot \alpha \lbrack b\rbrack)}{(n -
s)^2} &\mbox{if $\alpha\to \beta$ is a merge of paths of size $a$ and $b$.}
\end{cases}
\end{equation}
We will proceed to bound each of these three terms---the contribution from
extensions, merges, and $\overline{\mathcal{E}_t}$---in turn.
\paragraph{Bounding the contribution of extensions.} Let $o = \beta\lbrack
1\rbrack$. The $\alpha$ that can be made into $\beta$ by an extension are given
by removing up to $o$ ones from $\beta$ and then choosing one non-one element
of $\beta$ to decrement. For each $x \in \lbrace 0, \dots, o\rbrace$ and $y \in
\lbrack n\rbrack \setminus \lbrace 1\rbrace$ such that $\beta\lbrack j\rbrack >
0$, let $\beta'$ be $\beta$ with $x$ ones removed and one $y$ replaced with $y
- 1$.   

By our inductive hypothesis,
\begin{equation*}
\begin{split}
\E{\Hb_{\beta'}^s} &\le \left(\prod_{j\in \beta'} \frac{1}{j!}\right) \cdot
Q^{\abs{\beta'}_*} \cdot \left(\frac{s}{n}\right)^{\abs{\beta'}_* -
\nu(\beta')}\cdot c^{|\beta'|}  + Qn^{\abs{\beta'}_* - \ell}\\
 &\le \frac{y}{Q^{x+1}c^x}\left(\prod_{j\in \beta} \frac{1}{j!}\right) \cdot
 Q^{\abs{\beta}_*} \cdot \left(\frac{s}{n}\right)^{\abs{\beta}_* - \nu(\beta)
 - 1}\cdot c^{|\beta|}  + n^{-x - 1}Qn^{\abs{\beta}_* - \ell},
\end{split}
\end{equation*}
where we used the fact that $|\beta'|_*=|\beta|_*-x-1$ and the fact that $\nu(\beta) \ge \nu(\beta')$ by Claim~\ref{cl:s-prop}.
Using~\eqref{eq:p-bounds}, we can therefore bound the contribution of extensions that take $\beta'$ to
$\beta$ to
the sum by
\begin{align}\label{eq:iwebgewg24saffaf3ef3}
\sum_{s=1}^{t-1} q(x)\E{\Hb_{\beta'}^s} p_s(\beta', \beta) &\le
\frac{O(q(x)y\alpha\lbrack y\rbrack)}{Qc^x}\left(\prod_{j\in \beta}
\frac{1}{j!}\right) \cdot Q^{\abs{\beta}_*} \cdot c^{|\beta|} \cdot
\frac{1}{n}\sum_{s=1}^{t-1} \left(\frac{s}{n}\right)^{\abs{\beta}_* - \nu(\beta)
- 1} \\
&+ O(n^{-x-1})q(x)\alpha\lbrack y\rbrack Q n^{\abs{\beta}_* -
\ell}\frac{t}{n}\\
\end{align}
Since 
$$
\frac{1}{n}\sum_{s=1}^{t-1} \left(\frac{s}{n}\right)^{\abs{\beta}_* - \nu(\beta) - 1}\leq \frac{1}{n} \int_2^t
\left(\frac{s}{n}\right)^{\abs{\beta}_* - \nu(\beta) - 1}ds\leq \frac1{\abs{\beta}_* - \nu(\beta)} \left(\frac{s}{n}\right)^{\abs{\beta}_* - \nu(\beta)},
$$
we upper bound the lhs in~\eqref{eq:iwebgewg24saffaf3ef3} by
\begin{align*}
\sum_{s=1}^{t-1} q(x)\E{\Hb_{\beta'}^s} p_s(\beta', \beta)&\le
\frac{O(q(x)y\alpha\lbrack y\rbrack)}{Qc^x(\abs{\beta}_* - \nu(\beta))}
\left(\prod_{j\in \beta} \frac{1}{j!}\right) \cdot Q^{\abs{\beta}_*} \cdot
\left(\frac{s}{n}\right)^{\abs{\beta}_* - \nu(\beta)} \cdot c^{|\beta|} \\
&+
O(n^{-x-1})q(x)Q\ell n^{\abs{\beta}_* - \ell}.
\end{align*}
Summing the above over $x\in \set*{0, 1,\ldots, o}$ and $y > 1$ such that
$\beta[y]>0$, we get 
\begin{align}\label{eq:283g8ihfigwefgh09wfgdfsfd124}
\sum_{\alpha \in A^e}\sum_{s=1}^t \E{ \Hb_\alpha^s }
p_t(\alpha, \beta)
&\le  O(Q^{-1})\left(\prod_{j\in \beta}
\frac{1}{j!}\right) \cdot Q^{\abs{\beta}_*} \cdot
\left(\frac{s}{n}\right)^{\abs{\beta}_* - \nu(\beta)} \cdot c^{|\beta|}
\left(\sum_{x = 0}^o \frac{q(x)}{c^x}\right) \sum_{\substack{y \in
\beta :\\ y > 1}} \frac{y}{\abs{\beta}_* - \nu(\beta)}\\
\label{eq:extcontadditive}
&+ O(Q)n^{\abs{\beta}_* - \ell}\frac{|\beta| \ell
}{n}\sum_{x=0}^oq(x)n^{-x}\text{.}
\end{align}
Recall that $\beta$ is a multiset of integers, and  the above sum over $y\in
\beta$ goes over all elements of $\beta$, taking multiplicities into account.

We now bound the multiplicative terms
in~\eqref{eq:283g8ihfigwefgh09wfgdfsfd124}. For the first multiplicative terms
we have
\begin{align*}
\sum_{x=0}^oq(x)c^{-x} &\le \sum_{x = 0}^c \left(\frac{8x}{x}\right)^kc^{-x} +
\sum_{x=c}^\infty2^cc^{-x}\\
&\le \sum_{x=0}^\infty \left(\frac{O(1)}{x}\right)^x + O(1)\\
&= O(1). \stepcounter{equation}\tag{\theequation}\label{eq:sumcbound}
\end{align*}
For the second term we have
\begin{align*}
\sum_{\substack{y \in
\beta :\\ y > 1}} \frac{y}{\abs{\beta}_* - \nu(\beta)} &= \frac{\sum_{\substack{y \in
\beta :\\ y > 1}}y}{\sum_{y \in \beta} \left\lfloor \frac{y}{2} \right\rfloor }\\
&\le \frac{\sum_{\substack{y \in
\beta :\\ y > 1}}y}{\sum_{\substack{y \in \beta:\\ y>1}} \left\lfloor \frac{y}{2} \right\rfloor }\\
&\le 3.
\end{align*}
Finally, for the last additive term in~\eqref{eq:extcontadditive} we have
\begin{align*}
\sum_{x=0}^oq(x)n^{-x} &\le \sum_{x=0}^oq(x)c^{-x}\\
&= O(1) \stepcounter{equation}\tag{\theequation}\label{eq:sumnbound}
\end{align*}
by~\eqref{eq:sumcbound} and the fact that $c \le n$ whenever the interval of permitted $t$'s is non-empty, provided $Q \ge 1$. Therefore,
\begin{align*}
O(Q)n^{\abs{\beta}_* - \ell}\frac{|\beta|\ell}{n}\sum_{x=0}^o q(x)n^{-x} &\le
O(Q)n^{\abs{\beta}_* - \ell}\frac{2|\beta|\ell}{n}\\ 
&\le \frac{1}{3}Qn^{\abs{\beta}_* - \ell}
\end{align*}
as $|\beta|\ell < n/D$ for any constant $D > 0$ whenever the interval of permitted
$t$'s is non-empty, provided $Q$ is chosen to be sufficiently large. We therefore have \[
\sum_{\alpha \in A^e}\sum_{s = 1}^t \E{ \Hb_\alpha^s }
p_t(\alpha, \beta) \le \frac{1}{2} \left(\prod_{j\in \beta} \frac{1}{j!}\right) \cdot
Q^{\abs{\beta}_*} \cdot c^{|\beta|} \cdot
\left(\frac{s}{n}\right)^{\abs{\beta}_* - \nu(\beta)} + \frac{Q}{3}
n^{\abs{\beta}_* - \ell}
\]
provided $Q$ is chosen to be sufficiently large.

\paragraph{Bounding the contribution of merges.} Let $o = \beta\lbrack
1\rbrack$. The $\alpha$ that can be made into $\beta$ by a merge are given by
removing up to $o$ ones from $\beta$ and then choosing a $y > 2$ in $\beta$ to
replace with $a, b$, where $a + b + 1 = y$. 

For each $x \in \lbrace 0, \dots, o\rbrace$, $y > 2$ such that $\beta\lbrack
y\rbrack > 0$, and $a, b$ such that $a + b + 1 = y$, let $\beta_x^{
y\rightarrow a,b}$ be $\beta$ with $x$ ones removed and one $y$ replaced with
$a,b$. Note that $|\beta_x^{y\rightarrow a, b}|_*=|\beta|_*-1-x\leq |\beta|_*-1$, $|\beta_x^{y\rightarrow a, b}|=|\beta|-x+1$ and that $\nu(\beta) \ge \nu(\beta')$ by Claim~\ref{cl:s-prop}. By our
inductive hypothesis we thus have
\begin{align*}
\E{\Hb_{\beta_x^{y\rightarrow a,b}}^s} &\le c^{1 -
x}\frac{y!}{a!b!}\left(\prod_{j\in \beta} \frac{1}{j!}\right) \cdot
Q^{\abs{\beta}_* - 1} \cdot \left(\frac{s}{n}\right)^{\abs{\beta}_* - \nu(\beta)
- 1}\cdot c^{|\beta|}  + n^{-1-x}Qn^{\abs{\beta}_* - \ell}\\
&= c^{1 - x}(b+1){\binom{y}{a}}\left(\prod_{j\in \beta} \frac{1}{j!}\right)
\cdot Q^{\abs{\beta}_* - 1} \cdot \left(\frac{s}{n}\right)^{\abs{\beta}_* -
\nu(\beta) - 1}\cdot c^{|\beta|}  + n^{-1-x}Qn^{\abs{\beta}_* - \ell}
\end{align*}
and so we can bound the contribution of $a,b$ to $a + b + 1$ merges by
\begin{align*}
\sum_{s=1}^{t-1} q(x) \E{ \Hb_{\beta_x^{y\rightarrow a,b}}^s}
p_s(\beta_x^{y\rightarrow a,b}, \beta) &\le   O(c^{1 - x})q(x)b{\binom{y}{a}}
\left(\prod_{j\in \beta} \frac{1}{j!}\right) \cdot Q^{\abs{\beta}_* - 1} \cdot
c^{|\beta|} \cdot \frac{\alpha\lbrack a \rbrack \cdot \alpha\lbrack b
\rbrack}{(n-t)^2}&\\
&~~~~~\cdot\sum_{s=1}^{t-1}\left(\frac{s}{n}\right)^{\abs{\beta}_* - \nu(\beta)
- 1} + O(tn^{-1-x})q(x)Qn^{\abs{\beta}_* - \ell}\frac{\alpha\lbrack a \rbrack
\cdot \alpha\lbrack b \rbrack}{(n - t)^2}\\
&\le \frac{O(c \cdot n \cdot \ell^2)}{Q (n -t)^2} q(x)c^{- x}b\binom{y}{a}
\left(\prod_{j\in \beta} \frac{1}{j!}\right) Q^{\abs{\beta}_*}
c^{|\beta|} 
\left(\frac{t}{n}\right)^{\abs{\beta}_* - \nu(\beta)} \\
&~~~~~+ O(n^{-x})Qq(x) n^{\abs{\beta}_* - \ell}\frac{\ell^2}{\paren*{n-t}^2},
\end{align*}
where we used the bound 
\begin{equation*}
\begin{split}
\sum_{s=1}^{t-1}\left(\frac{s}{n}\right)^{\abs{\beta}_* - \nu(\beta)- 1}&\leq \int_2^t \left(\frac{s}{n}\right)^{\abs{\beta}_* - \nu(\beta)- 1} ds\\
&=\frac{n}{\abs{\beta}_*-\nu(\beta)}\left(\frac{t}{n}\right)^{\abs{\beta}_* - \nu(\beta)}\\
&\leq n\left(\frac{t}{n}\right)^{\abs{\beta}_* - \nu(\beta)},
\end{split}
\end{equation*}
since $\abs{\beta}_* - \nu(\beta)=\sum_{i\in \beta} (i-\lceil i/2\rceil)\geq 1$, as $\beta$ contains at least one component of size more than $1$ (the others are taken care of by the base case).

Now since $b \le \ell$ and $\sum_{a = 1}^y \binom{y}{a} \le 2^y \le 2^\ell$
and $|\beta| < \ell$, we can sum over $y \in \beta$ and $x = 0, \dots, o$ to
get 
\begin{align*}
\sum_{\alpha \in A^m}\sum_{s=1}^{t-1} q(|\beta - \alpha|) \E{
\Hb_\alpha^s } p_s(\alpha, \beta) &\le  \frac{c \cdot n\cdot \ell^4 2^\ell}{
Q(n-t)^2} \left(\prod_{j\in \beta} \frac{1}{j!}\right) \cdot Q^{\abs{\beta}_*}
\cdot c^{|\beta|} \cdot \left(\frac{t}{n}\right)^{\abs{\beta}_* - \nu(\beta)}
\left(\sum_{x = 0}^o q(x)c^{-x}\right)\\
&+ \frac{\ell^4 2^\ell}{\paren*{n-t}^2}Qn^{\abs{\beta}_* - \ell} \paren*{\sum_{x
= 0}^o q(x)n^{-x}}\\
&\le \frac{c\cdot n\cdot \ell^4 2^\ell}{
Q(n-t)^2} \left(\prod_{j\in \beta} \frac{1}{j!}\right) \cdot Q^{\abs{\beta}_*}
\cdot c^{|\beta|} \cdot \left(\frac{t}{n}\right)^{\abs{\beta}_* - \nu(\beta)}\\
&+ \frac{Q}{3}n^{\abs{\beta}_* - \ell} 
\end{align*}
by~\eqref{eq:sumcbound},~\eqref{eq:sumnbound}, and the fact that, by the bound
on $t$ in the lemma statement, $\frac{\ell^4 2^\ell}{(n-t)^2}$ can be bounded
above by any constant if $Q$ is chosen to be large enough. Finally,
as $n - t \ge 2^\ell\sqrt{c\cdot n}$, \[
\sum_{\alpha \in A^m}\sum_{s=1}^{t-1}q(|\beta - \alpha|) \E{
\Hb_\alpha^s } p_s(\alpha, \beta) \le \frac{1}{2} \left(\prod_{j\in \beta}
\frac{1}{j!}\right) \cdot Q^{\abs{\beta}_*} \cdot c^{|\beta|} \cdot
\left(\frac{s}{n}\right)^{\abs{\beta}_* - \nu(\beta)} +
\frac{Q}{3}n^{|\beta|_* - \ell}
\]
provided $Q$ is chosen to be large enough.

\paragraph{Bounding the contribution of the low probability event} By
Lemma~\ref{lm:ext-prob}, $\Pb{ \overline{\mathcal{E}_s}} \le
1/n^{\ell + 1}$ for each $s$. Then we have, using Lemma~\ref{lem:rbound-paths},
\begin{align*}
\sum_{\alpha \in \beta - 1}\sum_{s = 1}^{t-1} q(|\beta - \alpha|)\E{
\Hb_\alpha^s \middle| \overline{\mathcal{E}_s} } \Pb{
\overline{\mathcal{E}_s} } &\le \frac{1}{n^{\ell + 1}}\sum_{s = 1}^{t-1} \sum_{k =
1}^{|\beta|}c^{|\beta| - k} \sum_{\substack{z \in \lbrace 0, 1\rbrace^s\\\text{$z$
a collection of $k$ paths}}} \E[\Xb,\Bb_s]{ \wt{\Fb}_s^2(z) \middle|
\overline{\Ec_s}} \\
&= \frac{1}{n^{\ell + 1}}\sum_{s = 1}^{t-1} \sum_{k = 1}^{|\beta|}c^{|\beta| -
k} \sum_{\substack{z \in \lbrace 0, 1\rbrace^s\\\text{$z$ a collection of $k$
paths}}} \E[\Xb,\Bb_s]{ \wt{\rb}(z, \Fb_0, \Fb_s) \middle|
\overline{\Ec_s}}\\
&\le \frac{1}{n^{\ell + 1}}\sum_{s = 1}^{t-1} \sum_{k = 1}^{|\beta|}c^{|\beta|
- k} q(k)\\
&\le \frac{1}{n^{\ell}} \paren*{\sum_{k = 1}^{c}c^{|\beta| - k} (c/k)^k +
\sum_{k=c}^\infty c^{|\beta| - k} 2^c}\\
&\le \frac{1}{n^{\ell}} \paren*{\sum_{k = 1}^{\infty} k^{-k} +
\sum_{k=c}^\infty c^{c-k}}\\
&= \frac{O\paren*{c^{|\beta|}}}{n^{\ell}}\\
&\le O\paren*{n^{|\beta|_* - \ell}}\\
&\le \frac{Q}{3}n^{|\beta|_*-\ell}
\end{align*}
provided $Q$ is chosen to be large enough (which in particular implies $c \le
n$). Here we assumed $c$ is at least 2---if it is 1 the result follows from
$\sum_{s = 1}^{t-1} \sum_{k = 1}^{|\beta|}q(|\beta| - k) q(k) \le 4\ell^2$.

With these three bounds in hand, we return to equation~\eqref{eq:threecontribs}.
\begin{align*}
\E{ \Hb_\beta^t} &\le \sum_{\alpha \in A^e}\sum_{s = 0}^{t}
q(|\beta -\alpha|)
\E{ \Hb_\alpha^s } p_s(\alpha, \beta)
+  \sum_{\alpha \in A^m}\sum_{s = 0}^{t}q(|\beta - \alpha|)
\E{ \Hb_\alpha^s } p_s(\alpha, \beta)\\
&+ \sum_{\alpha \in \beta - 1}\sum_{s = 0}^{t} q(|\beta -
\alpha|) \E{ \Hb_\alpha^s \middle| \overline{\mathcal{E}_s} } \Pb{
 \overline{\mathcal{E}_s} }\\
&\le \frac{1}{2} \left(\prod_{j\in \beta} \frac{1}{j!}\right) \cdot
Q^{\abs{\beta}_*} \cdot c^{|\beta|} \cdot
\left(\frac{s}{n}\right)^{\abs{\beta}_* - \nu(\beta)} + \frac{Q}{3}
n^{\abs{\beta}_* - \ell}\\
&+ \frac{1}{2} \left(\prod_{j\in \beta}
\frac{1}{j!}\right) \cdot Q^{\abs{\beta}_*} \cdot c^{|\beta|} \cdot
\left(\frac{s}{n}\right)^{\abs{\beta}_* - \nu(\beta)} +
\frac{Q}{3}n^{|\beta|_* - \ell}\\
&~~~~~~~~~~~~~~~~~~~~~~~~~~~~~~~~~~~~~~~~~~~~~~~~~~~~~~~~~~~~~~~~~~~~~~ +
\frac{Q}{3}n^{|\beta|_*-\ell}\\
&= \paren*{\prod_{j\in \beta} \frac{1}{j!}} \cdot Q^{\abs{\beta}_*} \cdot
c^{|\beta|} \cdot \left(\frac{s}{n}\right)^{\abs{\beta}_* - \nu(\beta)} + Q
n^{\abs{\beta}_* - \ell}
\end{align*}

\end{proof}

\begin{lemma}
\label{lem:almostclosed}
For all $\varepsilon > 0$, with probability $1 - 1/n^2$ over $\Bb_n$, every cycle
has at least $\ell - 3/\varepsilon$ edges present at time $n -
n^{1-\varepsilon}/\ell$.
\end{lemma}
\begin{proof}
The probability of any edge \emph{not} being present at time $t$ is $(n - t)/n
= 1/\ell n^\varepsilon$. Moreover, these events are negatively associated, so
for any cycle the probability that at least $k$ edges are not present is at
most $\ell^k/(\ell n^\varepsilon)^k = n^{-\varepsilon k}$. So by setting $k =
3/\varepsilon$ and taking a union bound over all $n/\ell$ cycles the result
follows.
\end{proof}

\begin{proofof}{Lemma~\ref{lem:finalH}} We will assume $\varepsilon \le 1/24$ (as if not, it will suffice
to use the $Q$ that would be chosen for $\varepsilon = 1/24$). Let $\varepsilon'
= 2\varepsilon$, and let $D$ be chosen to be large enough that \[
2^\ell\sqrt{n\cdot c} + Q n^{5/6} 2^{2\ell} \log n < n^{1 - \varepsilon'}/\ell
\]
where $Q$ is the universal constant from Lemma~\ref{lem:hbetabound}.  Our
bounds on $c$ and $\ell$ allow this to be done with $D$ only depending on
$\varepsilon$.

Let $t' = n - n^{1 -\varepsilon'}/\ell$. Let $\Ec$ denote the event
(over $\Bb_n$) that at time $t'$ no cycle had more than $3/\varepsilon'$ edges
missing. By Lemma~\ref{lm:decomposition}, for
all $z \in \lbrace 0, 1\rbrace^n$, $t \in T$, \[ 
\wt{\Fb}_n(z)=\E[\Fb_{t}] {\wt{\Fb}_{t}(z_{\leq t}) \cdot
\wt{\rb}(z_{[t+1:n]}; \Fb_{t}, \Fb_n)  \middle| \Fb_n, \Bb_n}
\]
and so, choosing any $B_n \in \Ec$, \[
\E[\Xb]{ \wt{\Fb}_n(z)^2 \middle| \Bb_n = B_n}\rbrack \le
\E[\Xb]{ \wt{\Fb}_{t}(z)^2 \middle| \Bb_n = B_n} \text{.}
\]
Now, for each of the $n/\ell$ cycles present at time $n$, either at most $\ell
- 3$ of its edges are present at time $t'$ or there is a time $t < t'$ when the
$(\ell - 2)\nth$ edge of the cycle arrives. Furthermore, this implies
that there are fewer than $3/\varepsilon'$ different paths present in the cycle
at time $t'$. We may therefore write 
\begin{align*}
\E[\Xb]{ \Hb_{\lbrace \ell \rbrace} \middle| \Bb_n = B_n} &\le
\sum_{\substack{\alpha \in \zpp\\ \ell - 3/\varepsilon' \le \abs{\alpha}_* \le
\ell - 3\\ |\alpha| \le 3/\varepsilon' }}
\E[\Xb]{ \Hb_\alpha^{t'} | \Bb_n = B_n }\\
&+ \sum_{t = 1}^{t'} \sum_{\substack{\beta \in \zpp\\ \abs{\beta}_* = \ell -
2\\|\beta| \le 2}}\sum_{\alpha \in \beta - 1}\sum_{\substack{z \in \lbrace 0, 1\rbrace^t\\ z
\bsim_t \alpha\\ z \cdot 1 \bsim_{t+1} \beta}} \E[\Xb]{ \wt{\Fb}(z)^2 | \Bb_n =
B_n }.
\end{align*}

Now noting that $\Hb_{\ell} \le n/\ell$ with probability $1$ and $\Hb_\alpha \ge 0$ for all $\alpha$, we take 
expectation over $\Bb_n$, getting
\begin{align*}
\E[\Xb,\Bb_n]{ \Hb_{\lbrace \ell\rbrace } } &\le \sum_{\substack{\alpha \in
\zpp\\ \ell - 3/\varepsilon' \le \abs{\alpha}_* \le \ell - 3\\ |\alpha| \le
3/\varepsilon'}} \E[\Xb, \Bb_{t'}]{ \Hb_\alpha^{t'} \middle| \Ec} \cdot \Pb {
\Ec }\\
&+ \sum_{t = 1}^{t'} \sum_{\substack{\beta \in \zpp\\ \abs{\beta}_* = \ell -
2\\ |\beta| \le 2}}\sum_{\alpha \in \beta - 1} \E[\Xb, \Bb_t] { \Hb_\alpha^t \cdot
p(\alpha, \beta, \Bb_t) \middle |\Ec } \Pb{ \Ec }\\
&+ n\Pb{ \overline{\Ec} }\\
&\le \sum_{\substack{\alpha \in \zpp\\ \ell - 3/\varepsilon' \le \abs{\alpha}_*
\le \ell - 3\\ |\alpha| \le 3/\varepsilon'}} \E[\Xb, \Bb_{t'}] {
\Hb_\alpha^{t'}} + \sum_{t = 1}^{t'} \sum_{\substack{\beta \in \zpp\\
\abs{\beta}_* = \ell - 2\\ |\beta| \le 2}}\sum_{\alpha \in \beta - 1} \E[\Xb,
\Bb_t] { \Hb_\alpha^t \cdot p(\alpha, \beta, \Bb_t) } + 1/n\text{.}
\end{align*}
We now proceed to bound the first term in this sum. By
Lemma~\ref{lem:hbetabound}, for all $\alpha$, and for some universal constant $Q$, \[
\E{\Hb^{t'}_\alpha} \leq \left(\prod_{j\in \alpha} \frac{1}{j!}\right) \cdot
{Q}^{\abs{\alpha}_*} \cdot \left(\frac{t}{n}\right)^{\abs{\alpha}_* -
\nu(\alpha)}\cdot c^{|\alpha|}  + {Q}n^{\abs{\alpha}_* - \ell}
\]
and for $\alpha$ with $\ell - 3/\varepsilon' \le \abs{\alpha}_* \le \ell - 3$,
all terms in the inside product are at most $1$ and at least one is at most
$\frac{1}{(\ell\cdot \varepsilon'/3-1)!}$, and so \[
\E{\Hb^{t'}_\alpha} \leq  \frac{{Q}^\ell \cdot
c^{3/\varepsilon'}}{(\ell\cdot \varepsilon'/3 - 1)!} + {Q}n^{-3}
\]
and so as there are at most $2^\ell$ distinct $\alpha$ with $\abs{\alpha}_* \le \ell$,
\begin{align*}
\sum_{\substack{\alpha \in \zpp\\ \ell - 3/\varepsilon' \le \abs{\alpha}_* \le
\ell - 3\\ |\alpha| \le 3/\varepsilon'}} \E[\Xb, \Bb_{t'}] {
\Hb_\alpha^{t'} } &\le \frac{(2{Q})^\ell \cdot
c^{3/\varepsilon'}}{(\ell\cdot \varepsilon'/3 - 1)!} + {Q}2^\ell n^{-3}\le \varepsilon/4
\end{align*}
provided $D$ is chosen to be a sufficiently large constant. Next, we bound the
second term. By Lemma~\ref{lm:ext-prob}, there is an event $\mathcal{E}_t$ for
each $t$ such that, for any $B_t \in \mathcal{E}_t$, \[
p(\alpha, \beta, B_t)\le \begin{cases}
\frac{O(\alpha\lbrack a\rbrack)}{n} &\mbox{if $\alpha\to \beta$ is an extension
of a path of size $a$}\\
\frac{O(\alpha \lbrack a\rbrack \cdot \alpha \lbrack b\rbrack)}{(n -
t)^2} &\mbox{if $\alpha\to \beta$ is a merge of paths of size $a$ and $b$.}
\end{cases} 
\] 
and so for each $\alpha$ such that a $\beta$ with $\abs{\beta}_* = \ell - 2$ is
reachable by an extension, the sum of $p(\alpha, \beta, B_t)$ over all such
$\beta$ is at most \[
\sum_{a \in \alpha} \frac{O(\alpha\lbrack a\rbrack)}{n} = \frac{O(\ell)}{n}.
\]
For $\alpha$ such that a $\beta$ with $\abs{\beta}_* = \ell - 2$ is
reachable by a merge we use the fact that $t \le n - n^{1 - 2\varepsilon} \le
n - \sqrt{n}$ and so $\frac{1}{(n - t)^2} \le 1/n$ to obtain that the sum of
$p(\alpha, \beta, B_t)$ over all such
$\beta$ is at most\[
\sum_{a, b \in \alpha} \frac{O(\alpha\lbrack a\rbrack \cdot \alpha\lbrack b\rbrack)}{n}
 = \frac{O(\ell^2)}{n}
 \]
 and so (again using the fact that $\Hb_\alpha^t$ is never negative)
 \begin{align*} 
 \sum_{t = 1}^{t'} \sum_{\substack{\beta \in \zpp\\ \abs{\beta}_* = \ell - 2\\
 |\beta| \le 2}}\sum_{\alpha \in \beta - 1} \E[\Xb, \Bb_t] {
 \Hb_\alpha^t \cdot p(\alpha, \beta, \Bb_t) } &\le  \frac{O(\ell^2)}{n}\sum_{t = 1}^{t'}
 \sum_{\substack{\alpha \in \zpp\\ \abs{\alpha}_* = \ell - 3\\ |\alpha| \le 3}}
 \E[\Xb, \Bb_t] { \Hb_\alpha^t } \\
 &+ \ell \sum_{t = 1}^{t'}\sum_{\substack{\alpha \in \zpp\\ \abs{\alpha}_* =
 \ell - 3 \\ |\alpha| \le 3}} \E[\Xb, \Bb_t] { \Hb_\alpha^t \middle|
 \overline{\mathcal{E}}_t}
 \Pb[\Bb_t]{\overline{\mathcal{E}}_t}
 \end{align*}
with the second sum coming from the fact that there are at most $\ell$
different $\beta$ that can be reached from any given $\alpha$ by an extension
or merge. For the first of these sums, we note that for any $\alpha$ with
$\abs{\alpha}_* = \ell - 3$ and $|\alpha| \le 3$, at least one path in $\alpha$
is length at least $\ell/3 - 1$ and so \[
\E[\Xb,\Bb_t] { \Hb_{\alpha}^t } \le \frac{D^\ell \cdot
c^3}{(\ell/3 - 1)!} + Dn^{-3}
\]
for all $t \le t'$. So by summing over the $t' \le n$ time steps and at most $2^\ell$ choices of $\alpha$, we get
\begin{align*}
\frac{O(\ell^2)}{n}\sum_{t = 1}^{t'}
 \sum_{\substack{\alpha \in \zpp\\ \abs{\alpha}_* = \ell - 3\\ |\alpha| \le 3}}
 \E[\Xb, \Bb_t] { \Hb_\alpha^t } &\le O(\ell^2)\frac{(2D)^\ell \cdot c^3}{(\ell/3 - 1)!} + \frac{O(\ell^2 D)}{n^3}\\
 &\le \varepsilon/4
\end{align*}
if $D$ is chosen to be large enough. For the second sum, note that by
Lemma~\ref{lm:ext-prob}, $\Pb{ \overline{\mathcal{E}_t}
} \le 1/n^{\ell+1}$, and so by applying
Lemma~\ref{lm:decomposition} and Lemma~\ref{lem:rbound},
\begin{align*}
\ell \sum_{t = 1}^{t'}\sum_{\substack{\alpha \in \zpp\\ \abs{\alpha}_* = \ell -
3 \\ |\alpha| \le 3}} \E[\Xb, \Bb_t] { \Hb_\alpha^t \middle|
\overline{\mathcal{E}}_t }
\Pb[\Bb_t]{\overline{\mathcal{E}}_t} &\le \frac{\ell}{n^{\ell +
1}} \sum_{t=1}^{t'} \sum_{\substack{\alpha \in \zpp\\ \abs{\alpha}_* = \ell - 3
\\ |\alpha| \le 3}} \E[\Xb, \Bb_t] { \Hb_\alpha^t \middle|
\overline{\mathcal{E}}_t}\\
&\le \frac{\ell}{n^{\ell + 1}}
\sum_{t=1}^{t'} \sum_{\substack{z \in \lbrace 0, 1\rbrace^t\\ \text{$z$ a union
of 1, 2, or 3 components}}}  \E[\Xb, \Bb_t] { \wt{\Fb}_t(z)^2 \middle|
\overline{\mathcal{E}}_t}\\
&= \frac{\ell}{n^{\ell + 1}} \sum_{t=1}^{t'} \sum_{\substack{z \in \lbrace 0,
1\rbrace^t\\ \text{$z$ a union of 1, 2, or 3 components}}}  \E[\Xb, \Bb_t]
{ \wt{\rb}(z;\Fb_0, \Fb_t)^2 \middle| \overline{\mathcal{E}}_t}\\
&\le \frac{\ell}{n^{\ell + 1}} \sum_{t=1}^{t'} \left(q(1) + q(2) + q(3)\right)\\
&\le \frac{3\ell \cdot c^3}{n^\ell}\\
&\le \varepsilon / 4
\end{align*}
provided $D$ is chosen to be large enough. Finally, $1/n \le \varepsilon/4$ if $D$ is chosen to be large enough, giving us \[
\E[\Xb, \Bb_n]{ \Hb_{\lbrace \ell \rbrace} } \le \varepsilon
\]
by summing these four bounds together.
\end{proofof}

\subsection{Proof of Theorem~\ref{thm:cycleslb}}\label{sec:cycleslb}
\begin{proofof}{Theorem~\ref{thm:cycleslb}}
Suppose there was a protocol solving the distributional version of \SC$(n,\ell)$
with probability $2/3$ and using $\min\left(\ell^{\ell/D}, n^{1 -
\varepsilon}\right)$ space, for some $D$ to be chosen later. By the min-max
theorem, there is a deterministic algorithm that, given a uniformly random
instance of the communication problem, returns the identity of a cycle and its
parity with probability $2/3$. Let $\Fb_t$ be the random indicator function
associated with the messages of this protocol, as in the discussion in the
previous sections. Let $Z \in \lbrace 0, 1\rbrace^n$ be the random variable
denoting the coefficient of the cycle whose parity the protocol returns if the
final message is $\Fb_n$, and let $\Fb=\F_n$ to simplify notation. Let $\Ab =
\lbrace x \in \lbrace
0, 1\rbrace^n : \Fb(x) = 1 \rbrace$. Then
\begin{align*}
\E[\Xb, \Bb_n]{ \Hb^n_{\lbrace \ell' \rbrace} } &\ge
\E[\Xb, \Bb_n]{  \wt{\Fb}(Z)^2 }\\
&= \E[\Xb, \Bb_n]{ \left(|\Ab|^{-1}\sum_{x \in \Ab} (-1)^{x \cdot
Z}\right)^2 } \text{.}
\end{align*}
Now, as $\Xb$ and $\Bb_n$ are both uniformly distributed, $2/3$ of the possible
pairs result in the algorithm giving a correct answer when given as input. Call
these ``good'' pairs. Then \[
\E[\Bb_n]{ \Pb[\Xb] { \text{$(X, \Bb_n)$ is good } \middle|
\Bb_n } } \ge 2/3
\]
and so as, for each value of $\Bb_n$, the possible realizations $A$ of $\Ab$ conditional on $\Bb_n$ partition the
possible values of $\Xb$, \[
\E[\Bb_n]{ \Pb[\Xb] { \text{$(x,\Bb_n)$ is good for at least $7/12$ of the $x
\in \Ab$ } \middle| \Bb_n } } \ge 1/12\text{.}
\]
If at least $7/12$ of the $x \in \Ab$ are good, in particular they all give the
same value of $(-1)^{x \cdot Z}$, so \[
\left| |\Ab|^{-1}\sum_{x \in \Ab} (-1)^{x \cdot Z} \right| \ge 1/6
\]
and therefore
\begin{align*}
\E[\Xb, \Bb_n]{ \Hb^n_{\lbrace \ell' \rbrace} } &\ge \E[\Bb_n]{ \E[\Xb] {  \left(|\Ab|^{-1}\sum_{x \in \Ab} (-1)^{x \cdot Z}\right)^2 \middle| \Bb_n } }\\
&\ge \E[\Bb_n]{ 1/12 \cdot (1/6)^2 }\\
&= 1/432
\end{align*}
which contradicts Lemma~\ref{lem:finalH} if $D$ is chosen to be large enough.
\end{proofof}

\subsection{Proof of Theorem~\ref{thm:compestlb}}\label{sec:main-thm}

\begin{proofof}{Theorem~\ref{thm:compestlb}}
Suppose we had an algorithm solving $(1,\ell)$ component estimation in the
$(2,0)$-batch random order streaming model with probability at least $2/3$ and
using $\min\paren*{\ell^{\bOm{\ell}}, n^{1 - \varepsilon}}$ space. Let $\zeta =
\poly(n^{-1})$ denote the timestamp precision assumed by the algorithm (as per
the discussion in Definition~\ref{dfn:batchmodel}, it is assumed that time
stamps are presented at this resolution, so they can be expressed in $\bO{\log
n}$ bits). We design a protocol for \SC$(n,\ell)$ using only $\bO{\log n}$
extra bits of space.

Let $V$ be the vertex set associated with the \SC($n, \ell'$) problem, where
$\ell' = \lfloor \ell/2 \rfloor$. We will use $V \times \lbrace 0, 1 \rbrace$
as the vertex set of our component estimation input. When player $i$ receives the
$i\nth$ edge $uv$ with bit label $x$, they give the component
estimation algorithm two edges (in random order, say), $(u,0)(u,x)$ and
$(u,1)(u,\overline{x})$ with timestamp $\tb_i$.  The timestamps $\tb_i$ are
given by drawing $n$ uniform random variables with precision $\zeta$ and
presenting them in ascending order. See Appendix~\ref{app:streamstamps} for how
this can be done in $\bO{\log n}$ space in the stream, by storing only the most
recent timestamp generated. They then send the state of the algorithm, along
with the most recent timestamp generated, to player $i+1$. The final player
reads off the vertex returned by the algorithm and the size of the component it
is reported to contain, returning $0$ if the size is $2\ell$ and $1$ if it is
$\ell$.

The stream ingested by the component collection algorithm will be a
$(2,0)$-batch random order stream, as the edges given to us were in random
order, and we are generating 2 edges for each one, and the batches are ordered
by randomly drawn timestamps (we let the timestamps of the two edges be equal
to the timestamp of the corresponding batch).  Moreover, the graph it generates
will, for each cycle in the communication problem, have one cycle of length
$2\ell'$ (if the parity of the cycle in the communication problem is even) or
two cycles of length $\ell'$ (otherwise).

So this graph has no component of size more than $2\ell' \le \ell$, so with
probability $2/3$ the component collection algorithm will correctly return a
vertex in $V \times \lbrace 0, 1\rbrace$ and the corresponding component. The
answer returned by the final player will therefore be a correct solution to the
\SC instance.
\end{proofof}

\subsection{Proof of Theorem~\ref{thm:rwlb}}\label{sec:rwlb}
\begin{proofof}{Theorem~\ref{thm:rwlb}}
Suppose we had an algorithm solving $(k, s, 1/10, 1/10)$ random walk generation for some $k$ and $s$ in the
$(2,0)$-batch random order streaming model. Let $\zeta = \poly(n^{-1})$ denote the timestamp precision
assumed by the algorithm.

Let $V$ be the vertex set associated with the \SC($n, \ell'$) problem, where we choose $\ell'=\lfloor \sqrt{k/C}\rfloor$ for a large constant $C>1$ for the first lower bound and $\ell'=\lfloor k/2\rfloor$ for the second lower bound. We will
use $V \times \lbrace 0, 1 \rbrace$ as the vertex set of our component
estimation input. When we receive the $i\nth$ edge $uv$ with bit label
$x$, we will give the component estimation algorithm two edges (in random
order, say), $(u,0)(u,x)$ and $(u,1)(u,\overline{x})$ with timestamp $\tb_i$.
The timestamps $\tb_i$ are given by drawing $n$ uniform random variables with
precision $\zeta$ and presenting them in ascending order. See
Appendix~\ref{app:streamstamps} for how this can be done in $\bO{\log n}$ space
in the stream.

The stream ingested by the random walk generation algorithm will be a
$(2,0)$-batch random order stream, as the edges given to us were in random
order, and we are generating 2 edges for each one, and the batches are ordered
by randomly drawn timestamps (we let the timestamps of the two edges be equal to the timestamp of the corresponding batch).  Moreover, the graph it generates will, for each
cycle in the communication problem, have one cycle of length $2\ell'$ (if the
parity of the cycle in the communication problem is
even) or two cycles of length $\ell'$ (otherwise).

For the first lower bound, let $\ell'=\lfloor \sqrt{k/C}\rfloor$ for a sufficiently large absolute constant $C$, so that $k=C(\ell')^2$. Generate a $1/10$-approximate sample of the walk of length $k$, with error at most $1/10$ in total variation distance.  The walk loops around the cycle that it starts in with probability at least $2/3$ as long as the constant $C$ is sufficiently large (indeed, this happens with probability at least $9/10$ for a true sample of the random walk of length $k$, and therefore at least with probability $(1-1/10)9/10-1/10\geq 2/3$ for a $1/10$-approximate sample with TVD error bounded by $1/10$). Let $(v, b)$ denote the starting vertex. If the cycle is of length $\ell'$, output $v$ and $\text{parity} = 0$.  If the cycle is of length $2\ell'$, output $v$ and $\text{parity}=1$.

For the second lower bound, let $\ell'=\lfloor k/2\rfloor$ and run $C4^k$ $1/10$-approximate random walks  of length $k$ started at uniformly random vertices, with error at most $1/10$ in total variation distance (for the joint distribution), for a sufficiently large constant $C>0$. With probability at least $2/3$ at least one of the walks will loop around the cycle that it started in (it suffices for the walk to take a step in the same direction for $k$ consecutive steps, which happens with probability $(1-1/10)2^{-k}\geq 4^{-k}$, so $C4^k$ independent repetitions suffice for one of the walks to cover the cycle with probability $9/10$; accounting for the at most $1/10$ error in total variation distance gives the result). Let $(v, b)$ denote the starting vertex.  If the cycle is of length $\ell$, output $v$ and $\text{parity} = 0$.  If the cycle is of length $2\ell'$, output $v$ and $\text{parity}=1$.

Now, this is the hard instance of Theorem~\ref{thm:compestlb}, with $\ell'$ larger than a sufficiently large absolute constant since $k$ is. Thus, solving it requires $\min\{(\ell')^{\Omega(\ell')}, n^{0.99}\}$ space, setting the $\varepsilon$ parameter to $0.01$.
Expressing this lower bound in terms of $k$, we now get that $(k, 1, 1/10)$-random walk generation requires at least $\min\{(\ell')^{\Omega(\ell')}, n^{0.99}\}=\min\{k^{\Omega(\sqrt{k})}, n^{0.99}\}$ space, and $(k, C4^k, 1/10)$-random walk generation requires at least $\min\{(\ell')^{\Omega(\ell')}, n^{0.99}\}=\min\{k^{\Omega(k)}, n^{0.99}\}$ space, as required.
\end{proofof}


\section{Component Collection and Counting}\label{sec:comp}

 \subsection{Algorithmic Techniques}
Many random-order streaming algorithms work, at a high level, in the following
way:
\begin{itemize}
\item Sample some connected structure from the stream in an order-dependent way
(for instance, ``growing'' a component by randomly choosing a vertex and then
keeping every edge either incident to the vertex or to an already-sampled
edge).
\item Weight the sampled structures by the inverse of the prior probability of
sampling them.
\end{itemize}
Such techniques make use of the fact that, in a fully random-order stream, the
probability of any given set of edges arriving in any given order can be
determined exactly and without any additional information about other edges in
the graph. Now consider a stream divided into \emph{known} batches of size $b$.
Such techniques can be applied here by increasing the number of edges we sample by a factor of $b$:
\begin{itemize}
\item Whenever we would keep an edge, instead keep the entire \emph{batch}
containing that edge.
\item When weighting a structure, adjust the prior probability of sampling it
accordingly.
\end{itemize}
This is possible because we know the batches---when we see an edge we know
which batch it was in, and the probability of a given set of batches arriving
in any given order can still be determined exactly.

But what can we do when those batches are unknown? The key observation
we apply is that, if the structures we are sampling are not too large and there
are not too many of them, we can \emph{guess} the batches and only err on
``irrelevant'' edges:
\begin{itemize}
\item Maintain a buffer of all edges with timestamps less than $w$ before the
present edge.
\item Whenever we would keep an edge, instead keep every edge within $w$ of it
in either direction.
\item When weighting structures, assume that any pair of edges that we kept and
that had timestamps separated by at most $2w$ were in the same batch.
\end{itemize}
It is clear that, at least, when we keep an edge we will keep every other edge
in the same batch. Furthermore, as long as our structures are not too large and
there are not too many of them, any pair of edges in the same structure will,
with probability $w$, either be in the same batch, or at least $2w$ away from
each other. So as long as the total number of these edges is not much larger
than $1/\sqrt{w}$, our batch guesses will probably be correct, and so we may
proceed as if we were in the known-batch setting.

\paragraph{Component Collection and Counting.} We apply the ``batch guessing'' technique described above to the problem of
counting and collecting bounded-size components in a batch random order graph stream.

Similarly to the component counting strategy of~\cite{PS18}, we approach this
problem by first uniformly sampling a set of vertices, and then for each such
vertex ``growing'' a connected subgraph $\Db_v$ containing $v$ by keeping every
edge with a path to $v$ in our already-sampled edges, as long as $V(\Db_v)$
never exceeds a given limit $k$. We then construct random variables $\Xb_v$
with the following properties:
\begin{enumerate}
\item $\E{\Xb_v \cdot \1bb(\text{$\Db_v = K_v$})} = 1$ whenever $V(K_v) \le k$.
\item $\E{\Xb_v \cdot \1bb(\text{$\Db_v \not= K_v$})}$ is small.
\end{enumerate}
Here $\1bb$ denotes the indicator function and $K_v$ denotes the actual
component containing $G$. This will allow us to approximately count the number
of size $\le k$ components, and therefore obtain a $\varepsilon n$ additive
approximation to the component count if $k \ge 1/\varepsilon$. Furthermore, if
a large enough fraction of the vertices of $G$ are in size $\le k$ components,
it will let us sample a size $k$ component (by sampling vertices $v$ from $G$
and then choosing a sampled subgraph $\Db_v$ with probability proportional to
$\Xb_v$). The total number of samples required will go as the inverse of the
variance of $\Xb_v$.

The approach of~\cite{PS18} was based on defining a canonical spanning tree for
each component, and then keeping components iff this spanning tree was
collected first \emph{and} it was entirely collected in the first $\lambda$
fraction of the stream. 

The weighting of this component in their estimator ($\Xb_v$ in our formulation)
is then given by $(k-1)!/\lambda^{k-1}$ when the component is kept. They then
make use of the fact that any subgraph $H$ that is \emph{not} the entirety of
$K_v$ will have a non-empty boundary, and therefore cannot be collected if any
of the boundary edges arrive in the last $1-\lambda$ fraction of the stream (as
they are incident to the spanning tree of $H$ but are not in $H$). This means
that if $\lambda$ is small enough, $\E{\Xb_v \cdot \1bb(\text{$\Db_v \not=
K_v$})}$ is small too. However, they need $\lambda$ to be at most
$k^{-\bTh{k^2}}$ (which is based on counting the number of possible
``canonical'' spanning trees with a given boundary size that can be rooted at
$b$), which in turn means that the probability of successfully collecting a
$k$-vertex component is at most $k^{-\bTh{k^3}}$, and so their algorithm needs
$\paren*{1/\varepsilon}^{\bO{1/\varepsilon^3}}\log n$ bits of space.

We instead define our weighting based on explicitly calculating the prior
probability of collecting a component, and by a different combinatorial
analysis of the number of subgraphs  rooted at $v$ with a given boundary
size, are able to have $\lambda = 1/\poly(k)$, for a
$(1/\varepsilon)^{\bO{1/\varepsilon}}\plog n$ space cost. 

This also has the virtue of translating easily to the known-batches model: when
growing a component, we keep an entire batch whenever we would keep an edge in
it, and then calculate the prior collection probability of a component based on
our knowledge of how it was collected.it was collected. We extend this to the
\emph{hidden} batch model through the ``batch guessing'' technique discussed
earlier in this section.

Our main results in this section are Theorem~\ref{theorem:compcounting} and Theorem~\ref{theorem:compfinding} below.

\begin{restatable}[Counting Components]{theorem}{compcounting}
\label{theorem:compcounting}
For all $\varepsilon, \delta \in (0, 1)$, there is a $(b,w)$-hidden batch
random order streaming algorithm that achieves an $\varepsilon n$ additive approximation to
$c(G)$ with $1 - \delta$ probability, using \[ \paren{1/\varepsilon
\delta}^{\bO{1/\varepsilon}} (b + wm)\plog(n) \] bits of space.
\end{restatable}

\begin{restatable}[Component Collection]{theorem}{compfinding}
\label{theorem:compfinding}
For all $\delta \in (0, 1)$, there is a $(b,w)$-hidden batch random order  streaming algorithm
such that, if at least a $\beta$ fraction of the vertices of $G$ are in
components of size at most $\ell$, returns a vertex in $G$ and the component
containing it with probability $1 - \delta$ over its internal randomness and
the order of the stream, using \[
(\ell/\beta\delta)^{\bO{\ell}}(b + wm)\plog n
\]
bits of space.
\end{restatable}
We start by stating a formal definition of the $(b, w)$-hidden batch stream model:
\begin{definition} [Hidden-batch random order stream model; formal definition]\label{dfn:batchmodel}
In the $(b,w)$-hidden-batch random order stream model the edge set of the input graph $G=(V, E)$ is presented as
follows:
\begin{enumerate}
\item An adversary partitions $E$ into batches $\Bc=\{B_1,\ldots, B_q\}$ of size at most $b$, so that
\[
E=\bigcup_{B\in \Bc} B,\text{~and~}|B|\leq b\text{~for all~}B\in \Bc.
\]
\item Each batch $B\in \Bc$ is assigned an uniformly distributed starting
time $\tb_B\sim \Uc\paren{\brac{0, 1}}$.
\item For each $B\in \Bc$, the adversary assigns each edge $e \in B$ a timestamp $\tb_e \in
\brac{\tb_B, \tb_B + w}$.
\item The items and timestamps $\set{(e, \tb_e) : e \in E}$ are presented to the
algorithm in non-decreasing order of $\tb_e$, with the adversary breaking ties.
\end{enumerate}
\end{definition}
While the batch timestamps are continuous random variables, we assume the
adversary presents the edge timestamps with precision $\poly(1/n)$, so each
timestamp can be stored in $\bO{\log n}$ space. We do, however, require that
the adversary present the edges in the order given by the edge timestamps, so in
particular if $w = 0$ the adversary does not have the option of re-arranging
batches that happen to fall within $\poly(n^{-1})$ of each other.

When $b = 1$ and $w = 0$ this therefore collapses to standard random order
streaming, as an algorithm presented with an ordinary random order stream can
generated a sequence of appropriate timestamps ``on the fly'' (see
Appendix~\ref{app:streamstamps} for details).

\subsection{Notation}
We will use $\sigma = (e, \tb_e)_{e \in E}$ to denote a $(b,w)$-hidden batch
stream, received in order of the timestamps $(\tb_e)_{e \in E}$.

Throughout we will use $K_v$ to refer to the component of $G$ containing $v$,
and $\Kc_v$ to refer to $\set{B \cap K_v : B \in \Bc, B \cap K_v \not=
\emptyset}$, the partitioning of $K_v$ into batches. When $v$ is unambiguous we
will sometimes drop the subscript.

We will use $\Ibb(p)$ to denote the variable that is 1 if the predicate $p$
holds and $0$ otherwise.

\subsection{Component Collection}

\subsubsection{The Real and the Idealized Algorithm}

In this section we describe an algorithm $\ColComp(v,k)$ for collecting a
subset $\Db$ of $K_v$, the component containing $v$ in $G$, along with a guess
$\Dc$ of how $\Db$ is partitioned into batches in $\Bc$. This algorithm will be
a primitive in our component counting and collection algorithms. 

To aid with the analysis of this algorithm, we will define a second
``idealized'' algorithm \ColCompIdeal$(v,k)$. This algorithm will be allowed to
know how the stream is partitioned into batches $B$ and have direct access to
the batch timestamps $\tb_B$. We will show that typically both algorithms will
have almost the same output, allowing us to analyze $\ColComp$ by way of
$\ColCompIdeal$.

$\ColComp$ will work as follows:
\begin{itemize}
\item Grow a subgraph $\Db$ from a vertex $v$, keeping any edge that connects
to $v$ through a path in the edges already added to $\Db$.
\item Whenever a new edge $e$ is added to the component as described above,
ensure that all edges from the batch containing $e$ are added to $\Db$ as well,
by adding every edge with a timestamp up to $w$ before or after $\tb_e$. In
order to facilitate this, we keep a buffer $W$ of all edges with timestamps up
to $w$ before the edge currently being processed.
\item If the component of $v$ in $\Db$ ever has more than $k$ vertices, return
$\perp$.
\item Otherwise, return the component of $v$ in $\Db$ (now discarding edges
that do not connect to $v$), the timestamp of the last edge to add a new
\emph{vertex} to it, and a guess $\Dc$ at how it is partitioned into batches
(based on assuming that any two edges that arrived within $w$ of each other
were in the same batch).
\end{itemize}

We now describe the algorithm formally -- see Algorithm~\ref{alg:collectcomp} below. The parameter $s$ is used to track which edges should be added to $\Db$ on the
grounds of having timestamps up to $w$ \emph{after} an edge in the component
containing $v$.

\begin{algorithm}[H]
\begin{algorithmic}[1]
\Procedure{CollectComponent}{$v, k$}
\State $W \gets \emptyset$ \Comment{Buffering edges from up to $w$ ago.}
\State $\Db\gets (V, \emptyset)$ \Comment{The subgraph we are building up.} 
\State \Comment{$\Db$ has at most $d\binom{k}{2}$ edges so can be stored as a sparse graph.}
\State $\Tb \gets 0$ \Comment{The last time at which a new vertex was added
to the component of $v$ in $\Db$.}
\State $s \gets 0$ \Comment{The last time an edge was added to the component of
$v$ in $\Db$.}
\For{$(e, \tb_e)$ from $\sigma$}
\State Remove all edges from $W$ with time stamps before $\tb_e - w$.
\State Add $(e, \tb_e)$ to $W$.
\If{an endpoint of $e$ is connected to $v$ in $\Db$}
\State \Comment{Add $e$ to subgraph  $\Db$ together with all edges up to $w$ before and after.}
\State $\Db \gets \Db \cup \set{f : (f, \tb_f) \in W}$ 
\State \Comment{Record timestamps for edges added to $\Db$.}
\State $s \gets \tb_e$
\ElsIf{$\tb_e \le s + w$} 
\State \Comment{Add $e$ since it might be in the same batch as $f$ based on timestamp}
\State $\Db \gets \Db \cup \set{e}$
\EndIf
\If{the component of $\Db$ containing $v$ has more than $k$ vertices}
\State \return $(\perp, \perp, \perp)$
\EndIf\ 
\EndFor
\State $\Db \gets$ the component containing $v$ in $\Db$.
\State $\Dc \gets$ the finest partition of $\Db$ such that $\forall e, f
\in \Db, |\tb_e - \tb_f| \le w \Rightarrow \exists P \in \Dc, e, f \in P$.
\State \return $(\Db, \Dc, \Tb)$ \Comment{A subcomponent of $K_v$, a guess at
how it is partitioned, and the last time a \emph{vertex} was added to it.}
\EndProcedure
\end{algorithmic}
\caption{Collecting a component of size at most $k$ in a hidden batch
stream.}
\label{alg:collectcomp}
\end{algorithm}

$\ColCompIdeal$ will be almost identical to $\ColComp$, except now we will know
the batches and so we will not need to guess which edges are in which batch.
Let $\sigma'$ denote the stream of batches and time stamps $(B, \tb_B)$,
ordered by $\tb_B$.

\begin{algorithm}[H]
\begin{algorithmic}[1]
\Procedure{CollectComponentIdeal}{$v, k$}
\State $\Sb\gets (V, \emptyset)$ \Comment{The subgraph we are building up.}
\State $\Sc \gets \emptyset$ \Comment{The set of batches intersecting $\Sb$.}
\State $\Ub \gets 0$ \Comment{The last time at which a new vertex was added to
the component of $v$ in $\Sb$.}
\For{$(B, \tb_B)$ from $\sigma'$}
\If{$\exists e \in B$ with at least one endpoint connected to $v$ in $\Sb$} 
\If{$\exists e \in B$ with exactly one endpoint connected to $v$ in $\Sb$}
\State $\Ub \gets \tb_B$
\EndIf
\State $\Sb \gets \Sb \cup B$
\State $\Sc \gets \Sc \cup \set{B}$
\If{the component of $\Sb$ containing $v$ has more than $k$ vertices}
\State \return $(\perp, \perp, \perp)$
\EndIf
\EndIf
\EndFor
\State $\Sb \gets$ the component containing $v$ in $\Sb$.
\For{$B \in \Sc$}
\State $B \gets B \cap \Sb$
\EndFor\ \Comment{To match $\ColComp$, we take all the batches we added to
$\Sb$ and intersect them with the final value of $\Sb$ (that is, the component
containing $v$ in $\Sb$).}
\State \return $(\Sb, \Sc, \Ub)$  \Comment{A subcomponent of $K_v$, its
partitioning into batches, and the last time a \emph{vertex} was added to it.}
\EndProcedure
\end{algorithmic}
\caption{Collecting a component of size at most $k$ in a stream with known
batches.}
\label{alg:collectcompid}
\end{algorithm}

Note that $\Sb$ uniquely determines $\Sc$, as $\Sc$ is the set of batches
containing an edge from $\Sb$, with all edges not in $\Sb$ removed.

\subsubsection{Correspondence to Ideal Algorithm}
In this section, we prove that $\ColComp$ is a close approximation of
\ColCompIdeal, which we will then analyze in the subsequent section.

For these lemmas, we will need to consider the ``boundary'' batches of $\Sb$.
\begin{definition}
\label{dfn:batchboundary}
The boundary batch set $\Fc$ of $\Sb$ consists of every batch $B \in \Bc$ such
that at least one of the following holds:
\begin{enumerate}
\item $B \cap \Sb \not= \emptyset$ (i.e.\ $B \in \Sc$).
\item $\exists e \in B$ such that $e$ is incident to either $v$ or some edge in
$\Sb$.
\end{enumerate}
\end{definition}
Note that this is determined uniquely by $\Sb$.

First we show that the component collected by $\ColComp$ is always at least the
component collected by $\ColCompIdeal$.
\begin{lemma}
\label{lm:DconS}
Let $\Db$, $\Sb$ be the subgraphs returned by \ColComp, \ColCompIdeal,
respectively. Then
\[
\Db \supseteq \Sb\text{.}
\]
\end{lemma}
\begin{proof}
For $t \in \brac{0,1}$, let $\Db_t$, $\Sb_t$ be the states of $\Db$ and $\Sb$,
respectively, after every edge (for $\Db$) or batch (for $\Sb$) with a
timestamp no greater than $t$ has been processed. It will therefore suffice to
prove that $\Db_{1+w} \supseteq \Sb_1$. (As ultimately $\Db$, $\Sb$ will be the
components of $v$ in $\Db_{1+w}$, $\Sb_1$, respectively, and taking the
component containing $v$ will preserve the superset relation.)

Fix any assignment of timestamps $(\tb_B)_{B \in \Bc}$, $(\tb_e)_{e \in E}$. We
will prove the following by (strong) induction on the order of the time stamps
$(\tb_B)_{B \in \Bc}$: for all $B \in \Bc$, $\Db_{\tb_B + w} \supseteq
\Sb_{\tb_B}$. As no more edges are added to $\Sb$ after $\max\set{\tb_B : B \in
\Bc}$, this will suffice.

For any $B \in \Bc$, suppose that this holds for all $B'$ with $\tb_{B'} < \tb_B$.
Then we have $\Db_{t + w} \supseteq \Sb_{t}$ for all $t < \tb_B$, as if $B', B''$
are any pair of badges with no batches arriving between them, $\Sb_{t}$ is
unchanged in the interval $\lbrack \tb_{B'}, \tb_{B''})$, while $\Db_{t}$ is
non-decreasing.

Now, if the batch $B$ was not added to $\Sb$ at time $\tb_B$, the result holds
immediately. So suppose B was added. Then $B$ contains at least
one edge $e$ that is incident to some edge $f$ in a batch $B'$ with $\tb_{B'} <
\tb_B$. By our inductive hypothesis, $f \in \Db_{\tb_{B'} + w}$. 

Suppose $f$ was added to $\Db$ at some time after $\tb_B$. Then as $t_{B'} < t_B$, this time was was in $\brac{t_B, t_B + w}$, and so every edge in $\brac{t_B, t_B + w}$ will be added to $\Db$ by the end of that window. So $B \subseteq \Db_{t_{B} + w}$ and therefore $\Db_{t_{B} + w} \supseteq \Sb_{t_B}$. 

Now suppose instead $f$ was added to $\Db$ before $\tb_B$. Then in particular
$f \in \Db_{s}$ for all $s < \tb_e$. Therefore, when $e$ arrives, it is added
to $\Db$ along with every edge that arrives within $w$ of $\tb_e$, including
all of $B$. So $B \subseteq \Db_{t_{B} + w}$ and therefore $\Db_{t_{B} + w}
\supseteq \Sb_{t_B}$.
\end{proof}

Next, we show that if the boundary batches of $\Sb$ are sufficiently
well-separated in time, \ColComp\ returns the same subgraph as \ColCompIdeal,
with the right partitioning and almost the same final time.
\begin{lemma}
\label{lm:fcond}
Let $(\Db, \Dc, \Tb)$ and $(\Sb, \Sc, \Ub)$ be returned by \ColComp,
\ColCompIdeal, respectively, and let $\Fc$ be as defined in
Definition~\ref{dfn:batchboundary}. If no pair $B, B' \in \Fc$ has $\abs{\tb_B
- \tb_{B'}} \le 2w$, then $\Db = \Sb,
\Dc = \Sc$, and $\abs{\Tb - \Ub} \le w$.
\end{lemma}
\begin{proof}
First we will prove that, under these conditions, $\Db = \Sb$. By
Lemma~\ref{lm:DconS}, it will suffice to prove that $\Db \subseteq \Sb$. As in
the previous proof, define $\Db_t$, $\Sb_t$ to be the states of $\Db$ and
$\Sb$, respectively, after every edge (for $\Db$) or batch (for $\Sb$) with a
timestamp no greater than $t$ has been processed. Let $\Fb = \bigcup \Fc$. It
will suffice to prove that $\Db_{1+w} \cap \Fb \subseteq \Sb_{1+w} \cap \Fb$,
as the component of $\Sb_{1+w}$ containing $v$ is contained in $\Fb$, and if
the component of $\Db_{1+w}$ containing $v$ included any edge not in $\Sb_{1+w}
\cap \Fb$, it would also include at least one edge in $\Fb \setminus
\Sb_{1+w}$, as $\Fb$ contains the entire boundary of the component of
$\Sb_{1+w}$ containing $v$ (recalling that $\Db$, $\Sb$ are the components of
$\Db_{1+w}$, $\Sb_1$ containing $v$, respectively, and $\Sb_1 = \Sb_{1+w}$
trivially).

Fix any assignment of timestamps $(\tb_B)_{B \in \Bc}$, $(\tb_e)_{e \in E}$ such
that the lemma criterion holds. We will prove the following by (strong)
induction on the order of the time stamps $(t_e)_{e \in E}$: for all $e \in G$,
$\Db_{\tb_e} \cap \Fb \subseteq \Sb_{t_e} \cap \Fb$. As no more edges arrive in
$(\max_{e \in E} \tb_e, w\rbrack$, this will give us $\Db_{1+w} \subseteq
\Sb_{1+w}$.

For any $e \in G$, suppose that this holds for all $f$ with $\tb_f \le \tb_e$.
Then we have $\Db_t \cap \Fb \subseteq \Sb_t \cap \Fb$ for all $t < \tb_e$, as
edges are only added to $\Db$ at times corresponding to the timestamp of some
edge.

Now, if no edges in $\Fb$ were added to $\Db$ at the time $\tb_e$, the result
holds immediately. So suppose $f \in \Fb$ was added. Then, one of the following holds:
\begin{enumerate}
\item $e$ is connected to $v$ through some path in $\Db_s$ for some $s < \tb_e$,
and $\tb_f \in \brac{\tb_e - w, \tb_e}$.
\item $f = e$, and there is some $g$ such that $\tb_g \in \brac{\tb_e - w,
\tb_e}$, and $g$ is connected to $v$ through some path in $\Db_s$ for some $s <
\tb_g$
\end{enumerate}
In the first case, the path connecting $e$ to $v$ in $\Db_s$ must be contained
in $\Fb$. To see this, note that every edge incident to $v$ is in $\Fb$ along
with every edge incident to the component containing $v$ in $\Sb_1$. So if the
paths was not containing in $\Fb$, consider the first edge of the path not in
$\Fb$. This cannot be the first edge of the path, as that edge is incident to
$v$. So consider the edge immediately preceding it. This edge is in $\Fb$ but
not in $\Sb_1$, as if it were in $\Sb_1$ every edge incident to it would be in
$\Fb$. So it is not in $\Sb_s$ and thus by our inductive hypothesis it is not
in $\Db_s$, contradiction.

As the path is contained in $\Fb$, it is contained in $\Sb_s$ by our inductive
hypothesis. In that case, $e$ is connected to $v$ through some path in $\Sb_s$,
and so $e \in \Fb$. Therefore, as the batches in $\Fc$ have timestamps
separated by $2w$, and the batch $B$ containing $e$ has $\tb_B \in \brac{\tb_e
- w, \tb_e}$, there is no other batch intersecting $\Sb_s$ with a time stamp
after $\tb_e-2w$, and so $e$ is also connected to $v$ through a path in
$\Sb_{\tb_e-2w} \subseteq \tb_{B-w}$. Therefore, $B$ was added to $\Sb$ at the
time $\tb_B \in \brac{\tb_e-w,\tb_e}$. As $f \in \Fb$, it is in $B$, as
otherwise the batch $B'$ containing it would have to have $\tb_{B'} < \tb_e -
2w$, which is inconsistent with $\tb_f$. So $f \in \Db_{\tb_B} \subseteq
\Db_{\tb_e}$, completing the proof for this case.

Now consider the second case. By the same argument as for $f$ in the first
case, $g$ and a path connecting it to $v$ are in $\Fb$, and so the path is in
$\Sb_s$. Furthermore, no batch intersecting $\Sb_s$ arrives in $\brac{\tb_g -
2w, \tb_g}$, as the batch $B$ containing $g$ is in $\Fc$ and $\tb_B \in
\brac{\tb_g - w, \tb_g}$. Therefore, there is a path connecting $g$ to $v$ in
$\Sb_{\tb_B - w}$ and so $B \subseteq \Sb_{\tb_B}$. As $f \in \Fb$, it is in
$B$, as otherwise the batch $B'$ containing it would have to have $\tb_{B'} <
\tb_e - 2w$, which is inconsistent with $\tb_f$. So $f \in \Db_{\tb_B}
\subseteq \Db_{\tb_e}$, completing the proof.

Now we will prove that the partitioning $\Dc = \Sc$. This follows from the fact
that $\Db = \Sb$ and the separation of the batch timestamps---if $e,f \in B
\in \Sc$, then $\abs{\tb_e - \tb_f} \le w$ and so they are in the same
partition in $\Dc$. Conversely, if $e, f$ are in the same partition in $\Dc$,
$\abs{\tb_e - \tb_f}$, then they are in the same batch $B$, as if they were in
different batches, both batches would be in $\Fc$ and would have timestamps
within $2w$ of each other. So $e, f$ are in the same partition in $\Sc$.

Finally, we prove that $\abs{\Tb - \Ub} \le w$. Note that these are the final
time a vertex is added to the component containing $v$ in $\Db$ or $\Sb$,
respectively. This will therefore follow directly from our proof that $\Db_s
\subseteq \Sb_s$, and the Lemma~\ref{lm:DconS} proof that $\Db_{s + w}
\supseteq \Sb_s$.
\end{proof}

We now show that the criterion of Lemma~\ref{lm:fcond} holds, and therefore
$\ColComp$ ``almost'' matches $\ColCompIdeal$, with high probability whenever
$\Fc$ is not too large.
\begin{lemma}
\label{lm:fprob}
Let $(\Db, \Dc, \Tb)$ and $(\Sb, \Sc, \Ub)$ be returned by \ColComp,
\ColCompIdeal, respectively, and let $\Fc$ be as defined in
Definition~\ref{dfn:batchboundary}. Then
n 
\[
\Pb{(\Db = \Sb) \wedge (\Dc = \Sc) \wedge (\abs{\Tb - \Ub} < w) \middle| \Sb}
\ge 1 - 2w\abs{\Fc}^2\text{.}
\]
\end{lemma}
\begin{proof}
Condition on the order in which the batches in $\Fc$ arrive. First note that
\emph{any} value of $(\tb_B)_{B \in \Bc}$ such that these particular batches
arrive in the given order is sufficient to fix the value of $\Fc$ and $\Sb$.

This means that, conditioned on $\Sb$, $\Fc$, and this order, the
unlabelled set of timestamps $\set{\tb_B : B \in \Fc}$ is distributed as
$\abs{\Fc}$ independent and uniform samples from $\brac{0,1}$. Therefore, the
probability that any pair of them are within $2w$ is at most
$\binom{\abs{\Fc}}{2}4w \le 2w\abs{\Fc}^2$, and so the result follows by
Lemma~\ref{lm:fcond}.
\end{proof}
 Our component
counting and collection algorithms, Algorithms~\ref{alg:countcomps}
and~\ref{alg:findcomp}, will use \ColComp{} by setting some small threshold
$\lambda$ and throwing away the result whenever $\Tb > \lambda$. Now, when
$\Fc$ is large, $\Tb$ and $\Ub$ will concentrate near $1$.  We use this, along
with Lemma~\ref{lm:fprob}, to show that, if $\lambda < 1/2$ and $w$ is small
enough in terms of $k$, either the batches and partitions will match between
$\ColComp$ and $\ColCompIdeal$, or they will both be thrown away.
\begin{lemma}
\label{lm:DapproxS}
Let $(\Db, \Dc, \Tb)$ and $(\Sb, \Sc, \Ub)$ be returned by \ColComp,
\ColCompIdeal, respectively.  For any $\lambda \le 1/2$, with probability $1 -
\bO{k^4w\log^2 1/w}$, either both $\Tb$ and $\Ub$ are greater than $\lambda$,
or $\Tb$ and $\Ub$ are both smaller than $\lambda$ and $(\Db, \Dc) = (\Sb,
\Sc)$.
\end{lemma}
\begin{proof}
Condition on $\Sb$ and therefore $\Fc$. Furthermore, condition on the order
in which the batches of $\Fc$ arrive. First, suppose $\abs{\Fc} \le 5k^2 + 10\log
1/w$. Then by Lemma~\ref{lm:fprob}, with probability $1 - \bO{wk^4 \log^2
1/w}$, \[
\Db = \Sb, \Dc = \Sc, \abs{\Tb - \Ub} < w
\]
and so the result will hold provided $\abs{\Ub - \lambda} > w$.  $\Ub$ is always
$\bb_B$ for some $B \in \Fc$, so again using the fact that the unlabelled set
$\set{\bb_B : B \in \Fc}$ is distributed as $\abs{\Fc}$ independent and uniform
samples, this happens with probability at least $1 - \bO{w\abs{\Fc})}$. So for any realization $S$ of $\Sb$ such that $\Sb = S$ implies
\[
\abs{\Fc} \le 5k^2 + 10\log 1/w
\]
we have
\[
\Pb{\Db = \Sb, \Dc = \Sc, (\Tb, \Ub < \lambda \vee \Tb, \Ub >
\lambda) \middle| \Sb = S} \ge 1 - \bOt{wk^4}\text{.}
\]
Now suppose $\abs{\Fc} > 5k^2 + 10\log 1/w$. In particular, as $\Db$ and $\Sb$
have edges incident to at most $k$ vertices, and $\Db \supseteq \Sb$, at least
$9k^2/2 + 10\log 1/w$ of the batches in $\Fc$ have an edge with exactly one
endpoint in $\Sb$ and $\Db$. Call this set $\Fc'$.

Each batch in $\Fc'$ must arrive before $\Tb$ and $\Ub$, as otherwise it
would've been included in $\Sb$ or $\Db$ (since after these times they each
have reached their final vertex set), and so $\Tb, \Ub \ge \max_{B \in \Fc'}
\tb_B$.

Using again the fact that the unlabelled set $\set{\bb_B : B \in \Fc}$ is
distributed as $\abs{\Fc}$ independent and uniform samples, this means that
$\Tb, \Ub \le 1/2$ only if at least $9/10$ of $\Fc$ has timestamps $\le 1/2$,
which happens with probability at most
\begin{align*}
\binom{\abs{\Fc}}{9\abs{\Fc}/10}2^{-9\abs{\Fc}/10} &=
\binom{\abs{\Fc}}{\abs{\Fc}/10} 2^{-9\abs{\Fc}/10}\\
&\le
\paren*{\frac{e\abs{\Fc}}{\abs{\Fc}/10}}^{\abs{\Fc}/10}2^{-9\abs{\Fc}/10}\\
&= \paren*{\frac{10e}{2^9}}^{\abs{\Fc}}\\
&\le 2^{-\log 1/w}\\
&= w
\end{align*}
and so  for any realization $S$ of $\Sb$ such that $\Sb = S$ implies \[
\abs{\Fc} > 5k^2 + 10\log 1/w 
\]
we have
\[
\Pb{\Tb, \Ub > 1/2 \ge\lambda \middle| \Sb = S} \ge 1 - w\text{.}
\]
As $\Sb$ uniquely determines $\Fc$, the lemma therefore holds when conditioning
on any realization of $\Sb$.
\end{proof}

\subsubsection{Space complexity of \ColComp}{}
\begin{lemma}
\label{lm:compcollectionspace}
Algorithm~\ref{alg:collectcomp} can be implemented in $\bO{k^2(b + mw)\log^2
n}$ bits of space in expectation.
\end{lemma}
\begin{proof}
At all times in the execution of $\ColComp$, $W$ contains, at most, edges with
timestamps up to $w$ before that of the last edge processed. Other than $W$,
the algorithm has to keep, up to $\binom{k}{2}$ times, all edges with
timestamps within $w$ of some specified edge.

Therefore, the space usage of the algorithm is at most \[
\bO{k^2 \cdot M^* \cdot \log n}
\]
bits, where $M^*$ is the largest number of edges with timestamps within any
width-$2w$ window in the stream. To bound the expectation of $M^*$, we start by
noting that, as each batch $B$ has at most $b$ edges, all with timestamps in $\brac{\tb_B, \tb_B + w}$,
\begin{align*}
M^* &\le \max_{B \in \Bc} \paren*{\abs{B} + \sum_{\substack{B' \in \Bc:\\
\abs{\tb_{B'} - \tb_B} \le 3w}}\abs{B'}}\\
&\le b + \max_{B \in \Bc}\sum_{\substack{B' \in \Bc:\\ \abs{\tb_{B'} - \tb_B}
\le 3w}}\abs{B'}
\end{align*}
Now, fix some $B \in \Bc$. For all $B' \in \Bc \setminus \set{B}$, let
$\Ab_{B'}$ be the random variable that is $\abs{B'}$ if $\abs{\tb_{B'} - \tb_B}
\le 3w$ and $0$ otherwise. Then the variables $\Ab_{B'}$ are independent, are size at most $b$, and the sum of their expectations is at most $wm$ while
\begin{align*}
\sum_{B' \in \Bc \setminus \set{B}} \E{\Ab_{B'}^2} &= \sum_{B' \in \Bc
\setminus \set{B}}3w\abs{B'}^2\\
&\le 3wmb\text{.}
\end{align*}
So by the Bernstein inequalities, for all $t$, $\sum_{B' \in \Bc
\setminus \set{B}} \Ab_{B'} \le m + t$ with probability at least \[
e^{-\bOm{\frac{t^2}{wmb + bt}}}
\]
and so in particular, it is at most $\bO{\paren{wm + b} \log n}$ with
probability at least $1 - m^{-2}$, and so by a union bound $\abs{M^*} =
\bO{\paren{wm + b} \log n}$ with probability at least $1 - m^{-1}$.

We therefore have
\begin{align*}
\E{M^*} &\le \bO{\paren{wm + b} \log n} + m \cdot m^{-1}\\
&= \bO{\paren{wm + b} \log n}
\end{align*}
and so the lemma follows.
\end{proof}

\subsection{Properties of Idealized Component Collection}
In this section we will establish some properties of $\ColCompIdeal$ that will
be useful for both \emph{counting} and \emph{collecting} components.

Let $\lambda > 0$ be some parameter to be defined later. 
\begin{definition}
\label{dfn:batchorder}
For any $v \in V$, $H \subseteq G$, $\lambda > 0$, the batch probability
$p_v^\lambda(H)$ is the probability that both of the following happen:
\begin{itemize}
\item The batches intersecting $H$ arrive in \emph{batch order}---any order
such that, for every batch $B$ intersecting $H$, there is a path from $v$ to an
edge in $B$ consisting entirely of edges in batches intersecting $H$ with
timestamps before $\tb_B$.
\item $H$ is covered by time $\lambda$---for every $w \in V(H)$, there is a
path from $v$ to $w$ consisting entirely of edges in batches intersecting $H$
with timestamps before $\lambda$.
\end{itemize}
\end{definition}
We will use this to define a family of random variables $\Xb_v$. Let 
$$
(\Sb_v, \Sc_v, \Ub_v) =
\ColCompIdeal(v, k).
$$ Then we define \[
\Xb_v = \begin{cases}
0 &\text{~if~} \Sb_v = \perp\\
0 &\text{~if~} \Ub > \lambda\\
1/p_v^\lambda(\Sb_v) &\mbox{otherwise.}
\end{cases}
\]
Note that this can be determined entirely from the output of $\ColCompIdeal$,
without knowing anything else about the stream.

We then define \[
\Xb_v = \Yb_v + \Zb_v
\]
where $\Yb_v$ is $\Xb_v$ when $\Sb_v$ is the component of $G$ containing $v$,
and zero otherwise. Note that while $\Yb_v$, $\Zb_v$ are determined by the
stream, they cannot be identified from the output of \ColCompIdeal{}
alone.

We want to prove that the variables $\Yb_v$, corresponding to $\Sb_v$ being
``correct'', have ``nice'' properties---good expectation, bounded variance, and
approximate independence. Meanwhile, we want to prove that the ``error''
variables $\Zb_v$ are small in expectation.

\subsubsection{Correct Component Contribution}
\begin{lemma}
\label{lm:Yexp}
\[
\E{\Yb_v} = \begin{cases}
1 & \mbox{if $v$ is in a component with $\le k$ vertices.}\\
0 & \mbox{otherwise.}
\end{cases}
\]
\end{lemma}
\begin{proof}
If $v$ is \emph{not} in a component with $\le k$ vertices it is $0$ by
definition. Otherwise, let $H$ be the component containing it and $\Hc = \set{B
: B \in \Bc, B \cap H \not=\emptyset}$. Then it will be $1/p_v^\lambda(H) $ if
$\Sb_v = H$, $\Ub \le \lambda$, and $0$ otherwise. $\Sb_v = H$ iff the batches
in $\Hc$ arrive in batch order (as defined in Definition~\ref{dfn:batchorder}),
while $\Ub \le \lambda$ iff $H$ is covered by time $\lambda$. The probability
that both of these happen is exactly $p_v^\lambda(H)$.
\end{proof}

To bound the variance of $\Yb_v$, we will need some lower bounds on
$p_v^\lambda(H)$ when $H$ is the component containing $v$.

\begin{lemma}
\label{lm:pmin}
For any $v \in V$, let $H$ be the component of $G$ containing $v$. Then \[
p_v^\lambda(H) \ge (\lambda/k^2)^k\text{.}
\]
\end{lemma}
\begin{proof}
Consider a depth-first search tree for $H$. Suppose that the batches
intersecting the edges in this tree arrive in the order corresponding to the
depth-first search, before every other batch intersecting $H$, and before time
$\lambda$. Then:
\begin{itemize}
\item The batches intersecting $H$ arrive in batch order.
\item Every vertex in $H$ is covered by a tree of edges in these batches, each
of which has timestamp before $\lambda$.
\end{itemize}
Recall that batch order is defined in Definition~\ref{dfn:batchorder}.

So $p_v^\lambda(H)$ is at least the probability that this occurs. Now, as there
are no more than $\binom{k}{2} \le k^2$ distinct batches intersecting $H$, and
at most $k-1 < k$ batches intersecting the tree, the probability that the
batches intersecting the tree arrive in the depth-first search order and before
every other batch intersecting $H$ is at least $(1/k^2)^k$.

Now note that unconditionally, the probability that a given set of fewer than
$k$ batches would all arrive before time $\lambda$ is at least $\lambda^k$, and
conditioning on them being the first $k$ of the batches intersecting $H$ to
arrive only increases this probability. So the probability that both events
hold is at least $(\lambda/k^2)^k$, completing the proof.
\end{proof}

\begin{lemma}
\label{lm:Yvar}
\[
\Var{\Yb_v} \le (k/\lambda)^{\bO{k}}
\]
\end{lemma}
\begin{proof}
If $v$ is in a component with more than $k$ vertices, $\Yb_v$ is always 0.
Otherwise, by Lemma~\ref{lm:pmin}, we have
\begin{align*}
\Var{\Yb_v} &\le \E{\Yb_v^2}\\
&\le 1/(\lambda/k^2)^{2k}
\end{align*}
completing the proof.
\end{proof}

\begin{lemma}
\label{lm:Yindepcon}
For any vertex $v \in V$, if $v$ is in a component with more than $k$
vertices, let $\Kc_v = \emptyset$, otherwise let it be the set of batches that
contain at least one edge in the component of $G$ containing $v$. For any $U
\subseteq V$, if the sets $(\Kc_v)_{v \in U}$ are disjoint, the variables
$(\Yb_v)_{v \in U}$ are independent.
\end{lemma}
\begin{proof}
Suppose $v$ is in a component with more than $k$ vertices. Then $\Yb_v = 0$
always and therefore it is independent of $\Yb_u$ for all $u \in V$. 

Otherwise, for a batch to affect the output of $\ColCompIdeal(v,k)$, there
must be a path from it to $v$ in $G$. Therefore, $\Yb_v$ depends only on the
timestamps of batches intersecting its component, that is $\Kc_v$. So the
variables $(\Yb_v)_{v \in U}$ are independent provided the sets $(\Kc_v)_{v
\in U}$ are disjoint.
\end{proof}
\begin{lemma}
\label{lm:Yindep}
Let $S$ be a set of $r$ vertices sampled uniformly (with replacement) from
$V$.  Then with probability at least $1 - r^2bk^3/n$ over the choice of $S$,
the variables $(\Yb_v)_{v \in U}$ are independent conditioned on $S$.
\end{lemma}
\begin{proof}
For each $v \in U$, let $K_v$ be the component containing $v$. By
Lemma~\ref{lm:Yindepcon}, it will suffice to show that the sets $\set{\Kc_v
: v \in U, V(H_v) \le k}$ are disjoint with this probability. 

For each vertex $v$ in a component with at most $k$ vertices, there are at
most $k^2$ different batches intersecting this component, and therefore at
most $k^2b$ edges in these batches, and therefore at most $k^2b$ components
such that if $w$ is in that component, $\Kc_v \cap \Kc_w \not= \emptyset$.

So if we fix a $v$ in a component with at most $k$ vertices and then select a
$w$ from $V$, there are at most $k^3b$ choices of $w$ such that $w$ is in a
component with at most $k$ vertices such that $\Kc_v \cap \Kc_w \not=
\emptyset$. So for a randomly selected pair $v, w$ this happens with probability at most $\frac{k^3b}{n}$.

The proof then follows by taking a union bound over the $\le r^2$ pairs of
vertices in $U$.
\end{proof}
\subsubsection{Bounding the Contribution of Bad Components}
We want to prove that the expectation of the variables $\Zb_v$ is very small.
For this we need the fact that, for any vertex $v$, if there are too many
different ``wrong'' components that we might find connecting $v$, as can be the
case when $v$ is in a large component, each of these components also has a
large boundary, and so is unlikely to be found.
\begin{lemma}
\label{lm:boundary}
For any $\Hc \subset \Bc$, $v \in V$, let the $v$-\emph{boundary} of $\Hc$ be
the set of batches in $\Bc \setminus \Hc$ that contain at least one edge that
is connected to $v$ by a path of edges in $\bigcup \Hc$.

For any $v \in V$, $h, a \in \Nbb$, the number of size-$h$ subsets $\Hc$ of
$\Bc$ such that
\begin{enumerate}
\item every batch in $\Hc$ contains an edge $e$ such that there is a path from $v$ to an endpoint of $e$ in $\bigcup \Hc$
\item the $v$-boundary of $\Hc$ contains exactly $a$ batches
\end{enumerate}
is at most \[
\binom{h + a}{a}\text{.}
\]
\end{lemma}
\begin{proof}
In this proof we will make use of edge \emph{contraction}---to contract by an
edge $uv$, $uv$ is removed from the graph and the vertices $u,v$ are
identified with each other, and so any edge incident to either is now incident
to the merged vertex. This may result in the graph becoming a multigraph (if
$u,v$ are incident to the same edge) and having self-loops (if the graph is
already a multigraph, and one of multiple edges between $u$ and $v$ is
contracted). We will therefore prove the lemma for multigraphs with
self-loops, and it will follow for simple graphs as a special case.

Note that when contracting multiple edges, it does not matter in which order
we contract them.

We proceed by induction on $h + a$. If $h+a = 0$, the result follows
automatically. So suppose $h + a > 0$ and the result holds for all smaller
values of $h+a$. Then there is at least one batch $B$ containing an edge
incident to $v$. We will use the inductive hypothesis to bound the number of choices of $\Hc$ that contain $B$, and the number that do not.

First, consider any choice of $\Hc$ that does not contain $B$. Consider the
graph $G' = (V, E \setminus B)$ with batching $\Bc' = \Bc \setminus \set{B}$.
Then, each such choice of $\Hc$ is a subset of $\Bc'$ whose $v$-boundary with
respect to $G', \Bc'$ contains $a - 1$ batches. Moreover, each batch in $\Hc$
still contains an edge $e$ with a path from $v$ to an endpoint of $e$ in
$\bigcup \Hc$. So by applying our inductive hypothesis to $G', \Bc'$, there are
at most $\binom{h + a - 1}{a - 1}$ such choices of $\Hc$.

Secondly, consider any choice of $\Hc$ that \emph{does} contain $B$. Consider
the graph $G^*$ obtained by contracting every edge in $B$, with batching $\Bc^*$
given by removing $B$ and contracting each of its edges for the other batches.
Then there is a one-to-one correspondence between such choices of $\Hc$ and
subsets $\Hc^*$ of $\Bc^*$, again given by removing $B$ and contracting each of
its edges for the other batches. Each such subset will have a $v$-boundary
containing exactly $a$ batches, but will only contain $h - 1$ batches.
Furthermore, as contracting edges preserves connectedness, every batch in
$\Hc^*$ will contain an edge $e$ such that there is a path from $v$ to an
endpoint of $e$ in $\bigcup \Hc^*$. So, by applying our inductive hypothesis to
$G^*, \Bc^*$, there are at most $\binom{h + a - 1}{a}$ such choices of
$\Hc^*$, and therefore of $\Hc$.

The lemma then follows from the fact that \[
\binom{h + a - 1}{a - 1} + \binom{h + a - 1}{a} = \binom{h + a}{a}\text{.}
\]
\end{proof}
Now we are ready to bound the expectation of the ``bad'' contributions $\Zb_v$,
corresponding to $\Xb_v$ when the subgraph $\Sb$ returned by \ColCompIdeal\ is
\emph{not} equal to the actual component $K_v$.
\begin{lemma}
\label{lm:Zexp}
As long as $\lambda \le 1/2ek^2$, \[
\E{\Zb_v} = O(\lambda k^4)\text{.}
\]
\end{lemma}
\begin{proof}
$\Zb_v$ is non-zero precisely when $\ColCompIdeal(v,k)$ returns $(\Sb_v, \Sc,
\Ub)$ such that $\Ub \le \lambda$ and $\Sb_v \not= C_v$, where $C_v$ is the
component containing $v$ in $G$. For any possible value of $\Sb_v$, this
requires that
\begin{enumerate}
\item the set of batches $\Hc$ that intersect $\Sb_v$ arrive in batch order (as defined in Definition~\ref{dfn:batchorder})
\item $\Sb_v$ is covered by time $\lambda$
\item every batch $B$ in the $v$-boundary of $\Hc$ has $\tb_B < \lambda$
\end{enumerate}
with the latter being necessary because as $\Sb_v$ is covered by time
$\lambda$, any batch with timestamp at least $\lambda$ and connected to $v$ by
edges in $\bigcup \Hc$ has an edge incident to $\Sb_v$, and will therefore be
in $\Hc$ and therefore not in its $v$-boundary.

So the probability that it happens is at most $p_v^\lambda(\Sb_v)\lambda^a$ if
the $v$-boundary of $\Hc$ contains $a$ batches. So for each possible $\Sb_v$
with a batch set with boundary $a$, the expected contribution to $\Zb_v$ from
it is at most $\lambda^a$.

As $\Hc$ determines $\Sb_v$ exactly (since $\Sb_v$ is the component containing
$v$ in $\bigcup \Hc$, we can therefore use Lemma~\ref{lm:boundary} to bound the
expected contribution to $\Zb_v$ from collecting size-$h$ sets of batches $\Hc$
with size-$a$ boundaries by $\binom{h + a}{a}\lambda^a$.

As $\Sb_v \not= C_v$ requires $\Hc$ to have at least one batch in its
$v$-boundary, we can therefore bound $\E{\Zb_v}$ by
\begin{align*}
\sum_{h=0}^{\binom{k}{2}} \sum_{a = 1}^\infty \binom{h + a}{a}\lambda^a & \le
\sum_{h=0}^{\binom{k}{2}} \sum_{a = 1}^\infty (\lambda e(h/a + 1))^a\\
&\le k^2 \sum_{a = 1}^\infty (\lambda e(k^2 + 1))^a\\
&= \bO{\lambda k^4}
\end{align*}
provided $\lambda \le 1/2ek^2$.
\end{proof}
\subsection{Component Counting}
For any $v \in V$, we will use $c_v$ to denote the size of the component
containing $v$. In particular this means that $c(G)$, the number of components
in $G$ is $\sum_{v \in V} \frac{1}{c_v}$. 

Let $\lambda \in (0,1/2)$ and $r \in \Nbb$.  We now present an algorithm for
component counting based on $r$ copies of $\ColComp$. Informally, the algorithm
works as follows:
\begin{itemize}
\item Sample $r$ vertices $(\vb_i)_{i=1}^r$.
\item For each vertex $v$, run $\ColComp(v, k)$, rejecting the answer if the
timestamp $\Tb$ returned is greater than $\lambda$.
\item Approximate the component count $c(G)$ by $\frac{n}{r}\sum_{i=1}^r
\frac{\Xb_{\vb_i}}{\cb_{\vb_i}}$, where $\cb_{\cb_i}$ is the number of vertices
in the component that \ColComp\ sampled at $\vb_i$.  We don't have direct access
to the variables $\Xb_v$ but we can do this (with good enough probability)
because \ColComp\ normally almost-matches \ColCompIdeal\ and therefore $\Xb_v$
is usually $1/p_v^\lambda(\Db)$.
\end{itemize}
This will usually give us a good approximation to
$\frac{n}{r}\sum_{i=1}^r\frac{1}{c_{\vb_i}}$, because $\Xb_v = \Yb_v + \Zb_v$
is dominated by $\Yb_v$, corresponding to the case when the returned component
$\Db$ is actually $K_v$ (and so $\cb_{v} = c_v$), provided $\lambda$ is small
enough (so that $\Zb_v$ is small in expectation) and $r$ is large enough (so
that $\sum_{i=1}^r \frac{1}{c_v}\Yb_{\vb_i}$ concentrates around its
expectation over $(\Xb_{\vb_i})_{i=1}^r$).

Then, if $r$ is big enough, $\frac{n}{r}\sum_{i=1}^r\frac{1}{c_{\vb_i}}$ will
usually approximate $\sum_{v \in V} \frac{1}{c_v} = c(G)$ well enough, so we
are done.

We now formally describe the algorithm.
\begin{algorithm}[H]
\begin{algorithmic}[1]
\Procedure{CountComponents}{$k, \lambda, r$}
\State $\Cb \gets 0$
\For{$i \in \brac{r}$}
\State $\vb_i \gets \Uc(V)$
\State $(\Db_i, \Dc_i, \Tb_i) \gets \ColComp(\vb_i, k)$
\If{$\Tb_i \le \lambda$}
\State $\cb_{\vb_i} \gets \abs{V(\Db_i)}$
\State $\Cb \gets \Cb +
\frac{n}{r\cb_{\vb_i}}\cdot\frac{1}{p_{\vb_i}^\lambda(\Db_i)}$ \Comment{Usually $\frac{n\Xb_{\vb_i}}{r\cb_{\vb_i}}$.}
\EndIf
\EndFor
\EndProcedure
\end{algorithmic}
\caption{Counting the number of components in a graph.}
\label{alg:countcomps}
\end{algorithm}
\noindent
Recall that $\frac{1}{p_{\vb_i}^\lambda(\Db_i)} = \Xb_{\vb_i}$ whenever
$(\Db_i, \Dc_i)$ match the subgraph and batching given by $(\Sb,\Sc, \Ub) =
\ColCompIdeal(\vb_i,k)$ and the time stamps $\Tb_i$, $\Ub$ are either both
smaller than or both larger than $\lambda$. We start by showing that when this
is the case, the algorithm approximates $c(G)$ with good probability.

We start by showing that, if our inner loop simply added
$\frac{n}{r}\cdot\frac{1}{c_v}$ for each $v$ (or zero if $c_v > k$), we would
get a good approximation to $c(G)$ with good probability.

\begin{lemma}
\label{lm:csample}
With probability at least
$ 1 - k^2/r$ over $(\vb_i)_{i=1}^r$, \[
\abs{c(G) - \frac{n}{r}\sum_{i=1}^r\frac{1}{c_{\vb_i}}\Ibb(c_{\vb_i} \le k)} <
\frac{2}{k}n\text{.} 
\]
\end{lemma}
\begin{proof}
As the $\vb_i$ are sampled independently, the variables \[
\paren*{\frac{n}{c_{\vb_i}}\Ibb(c_{\vb_i} \le k)}_{i=1}^r
\]
are independent. Each has expectation at least $c(G) - \frac{1}{k}n$, as \[
\E{\frac{n}{\cb_{\vb_i}}} = c(G)
\]
and any time $c_{\vb_i} > k$, $\frac{n}{\cb_{\vb_i}} \le \frac{1}{k}n$.

Moreover, each has variance at most $n^2$, and so by Chebyshev's inequality
their average will be within $\frac{1}{k}n$ of their expectation with
probability $1 - k^2/r$.
\end{proof}
Next, we use the fact that the ``good'' part of $\Xb_{\vb_i}$, $\Yb_{\vb_i}$,
is zero whenever $\Db_i$ is not the right guess for the component containing
$\vb_i$ to show that $\frac{1}{\cb_v}\Yb_{\vb_i}$ is a good enough substitute
for $\frac{1}{c_v}\Ibb(c_v \le k)$.
\begin{lemma}
\label{lm:YaccC}
With probability at least $1 - r^2bk^3/n - (k/\lambda)^{O(k)}/\sqrt{r}$ over
$(\vb_i)_{i=1}^r$ and the order of the stream, \[
\abs{\frac{n}{r}\sum_{i=1}^r \frac{1}{\cb_v}\Yb_{\vb_i} -
\frac{n}{r}\sum_{i=1}^r\frac{1}{c_v}\Ibb(c_v \le k)} \le \frac{1}{k}n\text{.}
\]
\end{lemma}
\begin{proof}
For any $v$, $\Yb_{v}$ is zero whenever $c_v > k$ and otherwise by
Lemmas~\ref{lm:Yexp},~\ref{lm:Yvar}, it has expectation $1$ and variance
$(k/\lambda)^{\bO{k}}$. Furthermore, whenever $\Yb_v$ is non-zero, $\cb_v =
c_v$.

So \[
\frac{n}{r}\sum_{i=1}^r \E{\frac{1}{\cb_v}\Yb_{\vb_i} \middle| (\vb_i)_{i=1}^r}
=
\frac{n}{r}\sum_{i=1}^r\frac{1}{c_v}\Ibb(c_v \le k)
\]
and by Lemma~\ref{lm:Yindep}, with probability at least $1 - r^2bk^3/n$ over
$(\vb_i)_{i=1}^r$, the variables $\paren*{\frac{1}{\cb_v}\Yb_{\vb_i}}_{i=1}^r$
are independent with variances at most $(k/\lambda)^{\bO{k}}$. So by taking a
union bound with Chebyshev's inequality, the lemma follows.
\end{proof}
This leaves an error term $\abs{\frac{n}{r}\sum_{i=1}^r
\frac{1}{\cb_v}\Zb_{\vb_i}}$ to deal with. As the expectation of $\Zb_{\vb_i}$
is small when $\lambda$ is small enough, we can show that this is usually small
through applying Markov's inequality.
\begin{lemma} 
\label{lm:ZsmallC}
Fix any value of $(\vb_i)_{i=1}^r$. As long as $\lambda \le 1/2ek^2$, with
probability at least $1 - O(\lambda k^5)$ over the order of the stream, \[
\abs{\frac{n}{r}\sum_{i=1}^r \frac{1}{\cb_v}\Zb_{\vb_i}} \le
\frac{1}{k}n\text{.}
\]
\end{lemma}
\begin{proof}
By Lemma~\ref{lm:Zexp}, whenever $\lambda \le 1/2ek^2$, $\E{\Zb_v} = \bO{\lambda
k^4}$ for all $v \in V$, so this follows by a direct application of Markov's
inequality.
\end{proof}
This tells us that $\frac{n}{r}\sum_{i=1}^r \frac{1}{\cb_v}\Xb_{\vb_i}$ is a
good approximation to $c(G)$ with good probability, and so we can use
Lemma~\ref{lm:DapproxS} to lower bound our success probability.
\begin{lemma}
\label{lm:compcoutput}
As long as $\lambda \le 1/2ek^2$, with probability at least \[
1 - \bO{rk^4w\log^21/w + r^2bk^3/n + (k/\lambda)^{O(k)}/\sqrt{r} +
\lambda k^5}
\]
over $(\vb_i)_{i=1}^r$ and the order of the stream, \[
\abs{\Cb - c(G)} \le \frac{4}{k}n\text{.}
\]
\end{lemma}
\begin{proof}
By Lemma~\ref{lm:DapproxS} and a union bound, with probability $1 -
\bO{k^4w\log^21/w}$, taking the output of an instance of $\ColComp$ and
proceeding iff the timestamp output is at most $\lambda$ will give the same
result as doing so with an instance of $\ColCompIdeal$. Therefore, by taking
a union bound over the $r$ iterations of the inner loop, with probability $1 -
\bO{rk^4w\log^21/w}$, \[
\Cb = \frac{n}{r}\sum_{i=1}^r \frac{1}{\cb_v}\Xb_{\vb_i}\text{.}
\]
So recalling that $\Xb = \Yb + \Zb$, and taking a union bound over
Lemmas~\ref{lm:csample}, \ref{lm:YaccC}, and~\ref{lm:ZsmallC}, the lemma
follows.
\end{proof}
We now prove Theorem~\ref{theorem:compcounting}, restated here for convenience of the reader. The theorem follows by carefully choosing our algorithm
parameters in terms of $w, b$, and $m$.

\compcounting*
\begin{proof}
Assume that $\delta, \varepsilon \le 1/2$ (if they are in $(1/2,1)$, the result
will follow from the $1/2$ case). We start by setting $k = 4/\varepsilon$, and
$\lambda = \bTh{\delta/k^5}$ such that the $\bO{\lambda k^5}$ term in
Lemma~\ref{lm:compcoutput} is at most $\delta/4$ and $\lambda \le 1/2ek^2$.
Then, we set $r = \paren{1/\varepsilon\delta}^{\bTh{1/\varepsilon}}$ such that
the $\bO{(k/\lambda)^{O(k)}/\sqrt{r}}$ term is at most $\delta/4$.

Now, consider the $\bO{r^2bk^3/n}$ term. If it is greater than $\delta/4$, we have \[
n\log n \le \paren{1/\varepsilon\delta}^{\bO{1/\varepsilon}}b\log n
\]
and so the theorem follows immediately by using a union-find to exactly calculate the components of $G$.

If the $\bO{rk^4w\log^21/w}$ term is greater than $\delta/4$, we start by
noting that \[
w \log^2 1/w =  \bO{(1/m + w)\log^2 1/n}
\] as $w \log^2 1/w =
\bO{m^{-1}\log m}$ when $w \le 1/m$), so we have \[
n \log n \ge \paren{1/\varepsilon\delta}^{\bO{1/\varepsilon}}(1 + wm)\log^3 n
\]
and so the theorem again follows from using a union-find.

If neither of these hold, Lemma~\ref{lm:compcoutput} tells us that running
$\CountComp(k, \lambda, r)$ will give a $\varepsilon n$ additive approximation
to $c(G)$ with probability $1 - \delta$. As the space needed is that required
to run $r$ copies of $\ColComp$, by Lemma~\ref{lm:compcollectionspace} we
achieve the desired space.
\end{proof}
\subsection{Component Collection}
Let $\lambda \in (0,1/2)$ and $r \in \Nbb$.  We now present an algorithm for
collecting components  based on $r$ copies of $\ColComp$. Informally, the algorithm
works as follows:
\begin{itemize}
\item Sample $r$ vertices $(\vb_i)_{i=1}^r$.
\item For each vertex $v$, run $\ColComp(v, k)$, rejecting the answer if the
timestamp $\Tb$ returned is greater than $\lambda$.
\item Sample one of the components $\Db$ returned by these with probability
weighted by approximately $\Xb_v$ (using the fact that \ColComp\ normally
almost-matches \ColCompIdeal\ and therefore $\Xb_v$ is usually
$1/p_v^\lambda(\Db)$.
\end{itemize}
This will usually give us an actual component, because $\Xb_v = \Yb_v + \Zb_v$
is dominated by $\Yb_v$, corresponding to the case when the returned component
$\Db$ is actually $K_v$, provided $\lambda$ is small enough (so that $\Zb_v$ is
small in expectation) and $r$ is large enough (so that $\sum_{i=1}^r
\Yb_{\vb_i}$ concentrates around its expectation).

We now formally describe the algorithm.
\begin{algorithm}[H]
\begin{algorithmic}[1]
\Procedure{FindComponent}{$k, \lambda, r$}
\State $p \gets 0$
\For{$i \in \brac{r}$}
\State $\vb_i \gets \Uc(V)$
\State $(\Db_i, \Dc_i, \Tb_i) \gets \ColComp(\vb_i, k)$
\If{$\Db_i \not= \perp \wedge \Tb_i < \lambda$}
\State $p_i \gets 1/p_{\vb_i}^\lambda(\Db_i))$ \Comment{Usually $\Xb_{\vb_i}$.}
\State $p \gets p + p_i$
\Else
\State $p_i \gets 0$
\EndIf
\EndFor
\If{$p = 0$}
\State \return $\perp$
\EndIf
\State $(\vb^*, \Db^*) \gets (\vb_i, \Db_i)$ with probability $p_i/p$ for each
$i$.
\State \return $(\vb^*, \Db^*)$
\EndProcedure
\end{algorithmic}
\caption{Collecting a component in a graph.}
\label{alg:findcomp}
\end{algorithm}
Recall that $\frac{1}{p_{\vb_i}^\lambda(\Db_i)} = \Xb_{\vb_i}$ whenever
$(\Db_i, \Dc_i)$ match the subgraph and batching given by $(\Sb,\Sc, \Ub) =
\ColCompIdeal(\vb_i,k)$ and the time stamps $\Tb_i$, $\Ub$ are either both
smaller than or both larger than $\lambda$. We start by showing that when this
is the case, the algorithm returns an actual component of $G$ with good
probability.

When it holds (and assuming at least one run does not return $\perp)$, the
probability of returning a real component will be proportional to
$\sum_{i=1}^r\Yb_{\vb_i}$, as these are the $\Xb_{\vb_i}$ such that $\Db_i$ is
the component containing $\vb_i$. Meanwhile the probability of returning a bad
component will be proportional to $\sum_{i=1}^r \Zb_{\vb_i}$. So we need to
prove that $\sum_{i=1}^r\Yb_{\vb_i}$ is non-zero and large relative to
$\sum_{i=1}^r \Zb_{\vb_i}$.

First, we need a good enough fraction of the vertices sampled to be in size
$\le k$ components, as otherwise the $\Yb_{\vb_i}$ will be identically zero.
\begin{lemma}
\label{lm:betafrac}
Suppose a $\beta$ fraction of vertices of $G$
are in components of size at most $k$. Then with probability $1 -
e^{-r\lambda/2\beta}$ over $(\vb_i)_{i=1}^r$, at least a
$\beta - \sqrt{\lambda}$ fraction of the vertices $(\vb_i)_{i=1}^r$ are in
components of size at most $k$.
\end{lemma}
\begin{proof}
The vertices are sampled independently, so this follows directly by the
Chernoff bounds.
\end{proof}
Given this, we show that $\sum_{i=1}^r \Yb_{\vb_i}$ is reasonably large.
\begin{lemma}
\label{lm:findYacc}
Suppose a $\beta$ fraction of vertices of $G$
are in components of size at most $k$. Then with probability at least $1 -
e^{-r\lambda/2\beta} - r^2bk^3/n - (k/\lambda)^{O(k)}/\sqrt{r}$ over
$(\vb_i)_{i=1}^r$ and the order
of the stream, \[
\sum_{i=1}^r \Yb_{\vb_i} \ge r(\beta - 2\sqrt{\lambda}) \text{.}
\]
\end{lemma}
\begin{proof}
For any $v$, $\Yb_{v}$ is zero whenever it is in a component of size greater
than $k$ and otherwise by Lemmas~\ref{lm:Yexp},~\ref{lm:Yvar}, it has
expectation $1$ and variance $(k/\lambda)^{\bO{k}}$. 

So \[
\sum_{i=1}^r \E{\Yb_{\vb_i} \middle| (\vb_i)_{i=1}^r}
=
\sum_{i=1}^r\Ibb(c_v \le k)
\]
and by Lemma~\ref{lm:Yindep}, with probability at least $1 - r^2bk^3/n$ over
$(\vb_i)_{i=1}^r$, the variables $\paren*{\Yb_{\vb_i}}_{i=1}^r$ are independent
with variances at most $(k/\lambda)^{\bO{k}}$. So by taking a union bound with
Chebyshev's inequality and the result of Lemma~\ref{lm:betafrac}, the lemma
follows.
\end{proof}
We show that $\sum_{i=1}^r \Zb_{\vb_i}$ is small (when $\lambda$ is small
enough) by invoking the bound on the expectation of individual $\Zb_v$.
\begin{lemma}
\label{lm:findZsmall}
Fix any value of $(\vb_i)_{i=1}^r$. As long as $\lambda \le 1/2ek^2$, with
probability at least $1 - O(\sqrt{\lambda} k^4)$ over the order of the stream,
\[
\sum_{i=1}^r\Zb_{\vb_i} \le r\sqrt{\lambda} \text{.}
\]
\end{lemma}
\begin{proof}
By Lemma~\ref{lm:Zexp}, whenever $\lambda \le 1/2ek^2$, $\E{\Zb_v} = \bO{\lambda
k^4}$ for all $v \in V$, so this follows by a direct application of Markov's
inequality.
\end{proof} 
So now we have that, with good enough probability, $\sum_{i=1}^r
\Yb_{\vb_i}$ is large relative to $\sum_{i=1}^r\Zb_{\vb_i}$ and so we use
the fact that the output of $\ColComp$ usually almost matches the output of
$\ColCompIdeal$ to lower bound the probability with which our algorithm outputs
a valid component.
\begin{lemma}
\label{lm:compfindoutput}
As long as $\lambda \le 1/2ek^2$, with probability at least \[
1 - \bO{e^{-r\lambda/2\beta} + rk^4w\log^21/w + r^2bk^3/n +
(k/\lambda)^{O(k)}/\sqrt{r} + \sqrt{\lambda}k^4 + \sqrt{\lambda}/\beta}
\]
over $(\vb_i)_{i=1}^r$ and the order of the stream, $\Db^*$ is the component of
$G$ containing $\vb^*$.
\end{lemma}
\begin{proof}
By Lemma~\ref{lm:DapproxS} and a union bound, with probability $1 -
\bO{k^4w\log^21/w}$, taking the output of an instance of $\ColComp$ and
proceeding iff the timestamp output is at most $\lambda$ will give the same
result as doing so with an instance of $\ColCompIdeal$. Therefore, by taking
a union bound over the $r$ iterations of the inner loop, with probability $1 -
\bO{rk^4w\log^21/w}$, \[
(p_i)_{i=1}^r = \paren{\Xb_{\vb_i}}_{i=1}^r\text{.}
\]
So recalling that $\Xb = \Yb + \Zb$, and taking a union bound over
Lemmas~\ref{lm:findYacc} and~\ref{lm:findZsmall}, we have that with probability at least \[
1 - \bO{e^{-r\lambda/2\beta} +  k^4w\log^21/w + r^2bk^3/n +
(k/\lambda)^{O(k)}/\sqrt{r} + \sqrt{\lambda}k^4}
\]
over $(\vb_i)_{i=1}^r$ and the order of the stream,
\begin{align*}
\frac{\sum_{i=1}^r \Yb_{\vb_i}}{\sum_{i=1}^r\Zb_{\vb_i}} &\ge \frac{r(\beta -
2\sqrt{\lambda})}{r(\beta - \sqrt{\lambda})}\\
&= 1 - \bO{\sqrt{\lambda} / \beta}
\end{align*}
So if this holds, the algorithm will output a correct component with
probability $1 - \bO{\sqrt{\lambda}/\beta}$, as $\Yb_i > 0$ iff $\Db_i$ is the
component of $G$ containing $\vb_i$. The lemma therefore follows from taking
one final union bound.
\end{proof}

Finally Theorem~\ref{theorem:compfinding} follows by carefully choosing the algorithm parameters in
terms of $w$, $b$, and $m$.
\compfinding*

\begin{proof}
Assume that $\delta \le 1/2$ (if it is in $(1/2,1)$, the result will follow from
the $1/2$ case). We start by setting $k = \ell$, and $\lambda =
\bTh{\delta^2/k^8 + \delta^2/\beta^2}$ such that the
$\bO{\sqrt{\lambda}/\beta}$ and $\bO{\sqrt{\lambda k^4}}$ terms in
Lemma~\ref{lm:compfindoutput} sum to at most $\delta/3$ and $\lambda \le
1/2ek^2$.  Then, we set $r = \paren{\ell/\beta\delta}^{\bTh{1/\varepsilon}}$
such that the $\bO{(k/\lambda)^{O(k)}/\sqrt{r}}$ and $e^{-r\lambda/2\beta}$
terms are at most $\delta/3$.

Now, consider the $\bO{r^2bk^3/n}$ term. If it is greater than $\delta/4$, we have \[
\ell n\log n \le \paren{\ell/\beta\delta}^{\bO{1/\varepsilon}}b\log n
\]
and so the theorem follows immediately by keeping the first $\ell$ edges
incident to each vertex we see.

If the $\bO{rk^4w\log^21/w}$ term is greater than $\delta/4$, we start by
noting that \[
w \log^2 1/w =  \bO{(1/m + w)\log^2 1/n} 
\] as $w \log^2 1/w =
\bO{m^{-1}\log m}$ when $w \le 1/m$, so we have \[
\ell n \log n \ge \paren{\ell/\beta\delta}^{\bO{1/\varepsilon}}(1 + wm)\log^3 n
\]
and so the theorem again follows.
 
If neither of these hold, Lemma~\ref{lm:compfindoutput} tells us that running
$\FindComp(k, \lambda, r)$ will give a vertex in $G$ and the component
containing it with probability $1 - \delta$. As the space needed is that
required to run $r$ copies of $\ColComp$, by Lemma~\ref{lm:compcollectionspace}
we achieve the desired space.
\end{proof}

\section*{Acknowledgements}
Ashish Chiplunkar was partially supported by the Pankaj Gupta New Faculty
Fellowship. John Kallaugher and Eric Price were supported by NSF Award
CCF-1751040 (CAREER). Michael Kapralov was supported by ERC Starting Grant
759471.

John was also supported by Laboratory Directed Research and Development program
at Sandia National Laboratories, a multimission laboratory managed and operated
by National Technology and Engineering Solutions of Sandia, LLC., a wholly
owned subsidiary of Honeywell International, Inc., for the U.S. Department of
Energy's National Nuclear Security Administration under contract DE-NA-0003525.
Also supported by the U.S.  Department of Energy, Office of Science, Office of
Advanced Scientific Computing Research, Accelerated Research in Quantum
Computing program.
\bibliographystyle{alpha}
\bibliography{paper}

\begin{appendix}
\section{Generating Timestamps in the Stream}
\label{app:streamstamps}
In this section we show how an algorithm can generate $n$ timestamps in a
streaming manner, corresponding to drawing $n$ uniform random variables from
$(0,1)$ and then presenting each in order with $\poly(1/n)$ precision, using
$\bO{\log n}$ bits of space.

Let $(\Xb_i)_{i=1}^n$ denote $n$ variables drawn independently from
$\Uc(0,1)$ and then ordered so that $\Xb_i \le \Xb_{i+1}$ for all $i \in
\brac{n-1}$. By standard results on the order statistics (see
e.g.\ page 17 of~\cite{DN03}), the distribution of $(\Xb_i)_{i=j+1}^n$ depends
only on $\Xb_j$, and in particular they are distributed as drawing $(n-j)$
samples from $(\Xb_j, 1)$.

So then, to generate $(\Xb_i)_{i=1}^n$ with $\poly(1/n)$ precision in the
stream it will suffice to, at each step $i+1$, use $\Xb_{i}$ to generate
$\Xb_{i+1}$ (as sampling from the minimum of $k$ random variables to
$\poly(1/n)$ precision can be done in $\bO{\log n}$ space). We will only need
to store one previous variable at a time, to $\poly(1/n)$ precision, and so
this algorithm will require only $\bO{\log n}$ space.

\end{appendix}

\end{document}